\documentclass[11pt]{article}

\usepackage[T1]{fontenc}
\usepackage{bm, commath}
\usepackage{natbib}
\usepackage{caption}
\usepackage{graphicx}
\usepackage{amsmath, amsfonts, amsthm, amsfonts, dsfont, mathrsfs}
\usepackage{float}
\usepackage{subcaption}
\usepackage{booktabs,siunitx}
\usepackage{url}
\usepackage[colorlinks,citecolor=blue,urlcolor=blue,linkcolor=blue]{hyperref}
\usepackage{multirow}
\usepackage{amssymb}
\usepackage{bm}
\usepackage{bbm}
\usepackage{authblk}

\usepackage{custom_tex} 
\allowdisplaybreaks 

\usepackage[page]{appendix}

\usepackage{listings}
\usepackage{bbm}
\usepackage{textcomp}
\usepackage[margin=1in]{geometry}
\usepackage{authblk}
\usepackage[ruled, vlined, lined, boxed, commentsnumbered]{algorithm2e}
\SetKwInput{KwParam}{Parameter}
\SetAlgoCaptionLayout{centerline}

\usepackage{sectsty}
\usepackage[onehalfspacing]{setspace}
\usepackage[usenames,dvipsnames]{color}
\usepackage[utf8]{inputenc}
\usepackage{float}
\usepackage{tikz}
\usetikzlibrary{shapes,decorations,arrows,calc,arrows.meta,fit,positioning}
\tikzset{
    -Latex,auto,node distance =1 cm and 1 cm,semithick,
    state/.style ={ellipse, draw, minimum width = 0.7 cm},
    point/.style = {circle, draw, inner sep=0.04cm,fill,node contents={}},
    bidirected/.style={Latex-Latex,dashed},
    el/.style = {inner sep=2pt, align=left, sloped}
}

\newtheorem{theorem}{Theorem}[section]
\newtheorem{lemma}[]{Lemma}[section]
\newtheorem{proposition}[]{Proposition}[section]
\newtheorem{corollary}[]{Corollary}[section]
\newtheorem{definition}[]{Definition}[section]
\newtheorem{assumption}{Assumption}

\usepackage{mathtools}

\usepackage{xcolor}

\numberwithin{equation}{section}

\makeatletter
\renewcommand{\algocf@captiontext}[2]{#1\algocf@typo. \AlCapFnt{}#2} 
\def\@algocf@capt@plain{top}
\renewcommand{\algocf@makecaption}[2]{%
  \addtolength{\hsize}{\algomargin}%
  \sbox\@tempboxa{\algocf@captiontext{#1}{#2}}%
  \ifdim\wd\@tempboxa >\hsize
  \hsK,ip .5\algomargin%
  \parbox[t]{\hsize}{\algocf@captiontext{#1}{#2}}
  \else%
  \global\@minipagefalse%
  \hbox to\hsize{\box\@tempboxa}
  \fi%
  \addtolength{\hsize}{-\algomargin}%
}
\makeatother

\newcommand{\pibase}{\pi^{\textnormal{b}}}
\newcommand{\piimp}{\pi^{\uparrow}}
\newcommand{\hpiimp}{\hat{\pi}^{\uparrow}}
\newcommand{\hpiimpstar}{\hat{\pi}^{\uparrow,\star}}
\newcommand{\tpiimp}{\tilde{\pi}^{\uparrow}}
\newcommand{\nuca}{\texttt{NUCA}}
\renewcommand{\bar}{\overline}

\newcommand{\sra}{\textnormal{\texttt{SRA}}}

\newcommand{\rel}{/}
\newcommand{\Qfunc}{Q}
\newcommand{\Vfunc}{V}

\newcommand{\scrP}{\mathscr{P}}

\newcommand{\obs}{\textnormal{\texttt{obs}}}

\newcommand{\vc}{\textnormal{\texttt{vc}}}
\newcommand{\opt}{\textnormal{\texttt{opt}}}
\newcommand{\prox}{\textnormal{\texttt{approx}}}
\newcommand{\gen}{\textnormal{\texttt{gen}}}

\newcommand{\risk}{\calR}

\newcommand{\contr}{\mathscr{C}}
\newcommand{\imp}{\textnormal{Imp}}
\newcommand{\pirisk}{\pi^{\textnormal{\texttt{R}}}}
\newcommand{\pival}{\pi^{\textnormal{\texttt{V}}}}
\newcommand{\Lval}{L^{\textnormal{\texttt{V}}}}
\newcommand{\Uval}{U^{\textnormal{\texttt{V}}}}
\newcommand{\IV}{\textnormal{\texttt{IV}}}
\newcommand{\sens}{\textnormal{\texttt{sens}}}
\newcommand{\std}{\textnormal{\texttt{std}}}
\newcommand{\prosp}{\textnormal{\texttt{prosp}}}
\newcommand{\train}{\textnormal{\texttt{train}}}
\newcommand{\hpiq}{\hat\pi^{\textnormal{\texttt{Q}}}}
\renewcommand{\hat}{\widehat}
\renewcommand{\tilde}{\widetilde}

\newcommand{\Cimp}{\contr^{\piimp\rel\pibase}}

\newcommand{\numberthis}{\addtocounter{equation}{1}\tag{\theequation}}

\title{\sffamily\bfseries Estimating and Improving Dynamic Treatment Regimes With a Time-Varying Instrumental Variable
}
\author{\textsf{Shuxiao Chen}\thanks{Email:
\href{mailto:shuxiaoc@wharton.upenn.edu}{shuxiaoc@wharton.upenn.edu}}~}
\author{\textsf{Bo Zhang}\thanks{Email:
\href{mailto:bozhan@wharton.upenn.edu}{bozhan@wharton.upenn.edu}}}

\affil{{Department of Statistics, The Wharton School, University of Pennsylvania}} 

\date{\today}

\begin{document}
\maketitle

\sectionfont{\bfseries\large\sffamily}%

\subsectionfont{\bfseries\sffamily\normalsize}%





\vspace{0.5 cm}
\noindent
\textsf{{\bf Abstract}: 
Estimating dynamic treatment regimes (DTRs) from retrospective observational data is challenging as some degree of unmeasured confounding is often expected. 
In this work, we develop a framework of estimating properly defined ``optimal'' DTRs with a time-varying instrumental variable (IV) when unmeasured covariates confound the treatment and outcome, rendering the potential outcome distributions only partially identified. 
We derive a novel Bellman equation under partial identification, use it to define a generic class of estimands (termed \emph{IV-optimal DTRs}), and study the associated estimation problem. 
We then extend the IV-optimality framework to tackle the policy improvement problem, 
delivering \emph{IV-improved DTRs} that are guaranteed to perform no worse and potentially better than a pre-specified baseline DTR.
Importantly, our IV-improvement framework opens up the possibility of 
\emph{strictly} improving upon DTRs that are optimal under the no unmeasured confounding assumption (NUCA). 
We demonstrate via extensive simulations the superior performance of IV-optimal and IV-improved DTRs over the DTRs that are optimal only under the NUCA. 
In a real data example, we embed retrospective observational registry data into a natural, two-stage experiment with noncompliance using a time-varying IV and estimate useful IV-optimal DTRs that assign mothers to high-level or low-level neonatal intensive care unit based on their prognostic variables.}%

\vspace{0.3 cm}
\noindent
\textsf{{\bf Keywords}: Causal inference, Dynamic treatment regime, Instrumental variable, Offline reinforcement learning, Retrospective observational data}


\addtocontents{toc}{\protect\setcounter{tocdepth}{0}} 


\section{Introduction}
Estimating single-stage individualized treatment rules (ITRs) and the more general multiple-stage dynamic treatment regimes (DTRs) 
has attracted a lot of interest from diverse disciplines. 
Estimating optimal policies (ITRs or DTRs) can be challenging when data come from retrospective observational databases where some degree of unmeasured confounding is often expected. 
In these scenarios, an instrumental variable (IV) is a useful tool to infer the treatment effect.
Motivated by recent works on estimating optimal ITRs using an instrumental variable \citep{cui2019semiparametric, qiu2020optimal, pu2020estimating} and literature on policy improvement (\citealp{kallus2018interval,kallus2020confounding,kallus2020minimax}), we study in this article how to leverage information contained in a time-varying instrumental variable to estimate properly-defined ``optimal'' DTRs and improve upon pre-specified baseline DTRs. 

An instrumental variable is valid if it is associated with the treatment, affects the outcome only through its association with the treatment, and is independent of the unobserved treatment-outcome confounding variables, possibly conditional on a rich set of observed covariates. 
One subtlety in IV-based analysis lies in that even a valid IV cannot always identify the mean potential outcome; rather, a valid IV along with appropriate, application-driven IV identification assumptions places certain restrictions on the potential outcome distributions. 
This line of research is known as \emph{partial identification of probability distributions} (\citealp{manski2003partial}). 
This subtlety is inherited by the policy estimation problem with an IV. 
In particular, when the conditional average treatment effect (CATE) is not point identified from data, the optimal policy that maximizes the value function cannot be identified either, necessitating researchers to target alternative optimality criteria. 
While such criteria have been proposed in the single-stage setting from different perspectives (see, e.g., \citealt{cui2019semiparametric,cui2021machine, pu2020estimating}), the literature on the more complicated, multiple-stage setting is scarce.

Our first primary interest in this article is to extend optimality criteria in single-stage settings (\citealp{murphy2003optimal,cui2019semiparametric,cui2021machine,pu2020estimating}) and develop an optimality criterion that is tailored to general \emph{sequential} decision problems and incorporates the rich information contained in a time-varying IV.
Our optimality criterion, termed \emph{IV-optimality}, is based on a carefully weighted version of the partially identified $Q$-function and value function subject to the distributional constraints imposed by the IV.
This criterion is distinct from the framework of \citet{han2019optimal}, who endows the collection of partially identified DTRs with a partial order and characterizes the set of maximal elements (see \citealt{zhang2020selecting} for similar ideas).
It also distinct from the recent work in the reinforcement learning literature \citep{liao2021instrumental} that directly models the transition dynamics and uses instrumental variables to identify the relevant causal parameters.
In particular, we do not impose Markovian assumptions and do not pose parametric models a priori.
We then take a hybrid approach of $Q$-learning \citep{watkins1992q,schulte2014q} and weighted classification \citep{zhang2012estimating,zhao2012estimating} to target this optimality criterion and establish non-asymptotic rate of convergence of the proposed estimators.


The IV-optimality framework also motivates a conceptually simple yet highly informative variant framework that allows researchers to leverage a time-varying IV to improve upon a baseline DTR. 
The policy improvement problem was first considered in a series of papers by \citet{kallus2018interval,kallus2020confounding,kallus2020minimax} under a ``Rosenbaum-bounds-type'' sensitivity analysis model \citep{rosenbaum2002observational}. 
Despite its novelty and usefulness in a range of application scenarios,
a sensitivity-analysis-based policy improvement framework does have a few limitations. 
First, each improved policy is indexed by a sensitivity parameter $\Gamma$ that controls the degree of unmeasured confounding. 
Since the sensitivity parameter is not identified from the observed data, it is often unclear which improved policy best serves the purpose. 
Second, as pointed out by \citet{heng2020interactions}, ``Rosenbaum-bounds-type'' sensitivity analysis model does not fully take into account unmeasured confounding heterogeneity; see also \citet{bonvini2019sensitivity}. 
More importantly, one major limitation of a sensitivity-analysis-based policy improvement framework is that it \emph{cannot} improve upon the \emph{NUCA-optimal policy}, i.e. the policy that is optimal under the no unmeasured confounding assumption \citep{rosenbaum1983central, Robins1992}. 
Intuitively, this is because a sensitivity analysis model and a fixed sensitivity parameter only introduce a partial order (rather than a total order) among all candidate policies and the NUCA-optimal policy always remains a maximal element in this partial order \citep{zhang2020selecting}. 
Therefore, their framework provably cannot improve upon the NUCA-optimal policy, a policy of major interest in many application scenarios.
See Section \ref{sec: ITR IV improve} for a detailed discussion.
As we will demonstrate in this article, an IV-based policy improvement framework solves all aforementioned problems simultaneously.

The rest of the article is organized as follows. 
We describe a real data application in Section \ref{subsec: intro application}. 
Section \ref{sec: review} reviews alternative optimality criteria in the single-stage setting and provides a general IV-optimality framework that is amenable to being extended to the multiple-stage setting.
Section \ref{sec: ITR IV improve} considers improving upon a baseline ITR with an IV. 
Building upon the preparations in Sections \ref{sec: review} and \ref{sec: ITR IV improve}, we describes the IV-optimality framework for policy estimation
in the multiple-stage setting in Section \ref{sec: extension to DTR} and extends this framework to tackle the policy improvement problem in Section \ref{sec:dtr_improve}.
Section \ref{sec: theory} studies the theoretical properties of the proposed methods. 
We conducted extensive simulations in Section \ref{sec: simulation} and revisited the application in Section \ref{sec: application}. Section \ref{sec: discussion} concludes with a brief discussion. For brevity, all proofs are deferred to the Supplementary Material.

\subsection{Application: A Natural, Two-Stage Experiment Derived from Retrospective Registry Data}
\label{subsec: intro application}

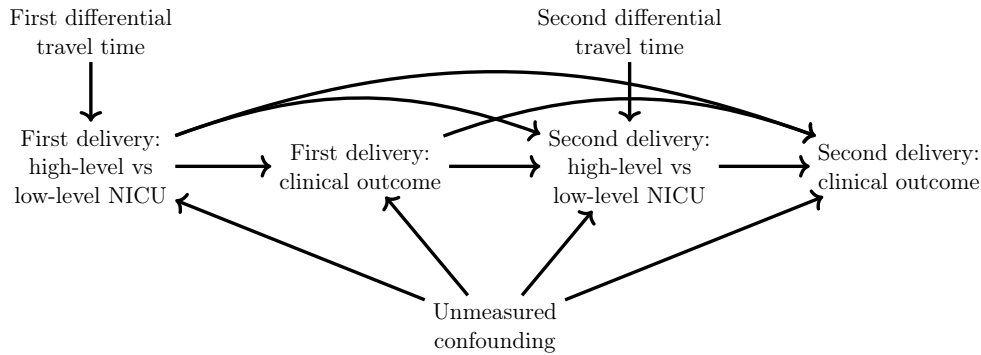
\begin{figure}[t!]
    \centering
    \resizebox{0.8\textwidth}{!}{
    \begin{tikzpicture}[scale=0.9, every text node part/.style={align=center}
]
\node at (0,0) (A1) {First delivery: \\high-level vs \\ low-level NICU};
\node at (5,0) (R1) {First delivery: \\ clinical outcome};
\node at (10,0) (A2) {Second delivery: \\high-level vs \\ low-level NICU};
\node at (15,0) (R2) {Second delivery: \\ clinical outcome};
\node at (7.5,-3) (U) {Unmeasured \\confounding};

\node at (0,2.5) (Z1) {First differential\\ travel time};
\node at (10,2.5) (Z2) {Second differential \\ travel time};

\draw[ultra thick, ->] (A1) -- (R1);
\draw[ultra thick, ->] (R1) -- (A2);
\draw[ultra thick, ->] (A2) -- (R2);
\draw[ultra thick, ->] (A1) to [out=20, in=160] (A2);
\draw[ultra thick, ->] (A1) to [out=20, in=160](R2);
\draw[ultra thick, ->] (R1) to [out=20, in=160](R2);

\draw[ultra thick, ->] (U) -- (A1);
\draw[ultra thick, ->] (U) -- (R1);
\draw[ultra thick, ->] (U) -- (A2);
\draw[ultra thick, ->] (U) -- (R2);

\draw[ultra thick, ->] (Z1) -- (A1);
\draw[ultra thick, ->] (Z2) -- (A2);

\end{tikzpicture}}
    \caption{\small A DAG illustrating the NICU application. Observed baseline covariates like mothers' demographics, neighborhood characteristics, etc, are omitted for clearer presentation.}
    \label{fig: IV DAG application}
\end{figure}

\citet{lorch2012differential} constructed a retrospective cohort study to investigate the effect of delivery hospital on premature babies (gestational age between $23$ and $37$ weeks) and found a significant benefit to neonatal outcomes when premature babies were delivered at hospitals with high-level neonatal intensive care units (NICU) compared to those without NICUs. \citet{lorch2012differential} used the differential travel time to the nearest high-level versus low-level NICU as an IV so that the outcome analysis is less confounded by mothers' self-selection into high-level NICUs. Put another way, the differential travel time creates a natural experiment with noncompliance: mothers who live relatively close to a high-level NICU were encouraged to deliver, although not necessarily delivered, at a high-level NICU. More recently, \citet{michael2020instrumental} considered mothers who delivered exactly two babies from 1996 to 2005 in Pennsylvania, and investigated the cumulative effect of delivering at high-level NICUs on neonatal survival probability using the same differential travel time IV.

Currently, there is still limited capacity at high-level NICUs, so it is not practical to direct all mothers to these high-technology, high-volume hospitals. Understanding which mothers would most significantly benefit from delivering at a high-level NICU helps design optimal perinatal regionalization systems that designate hospitals by the scope of perinatal service provided and designate where infants are born or transferred according to the level of care they need at birth (\citealp{lasswell2010perinatal}; \citealp{kroelinger2018comparison}). Indeed, previous studies seemed to suggest that although high-level NICUs significantly reduced deaths for babies of small gestational age, they made little difference for almost mature babies like 37 weeks (\citealp{yang2014estimation}).

Mothers who happened to \emph{relocate} during two consecutive deliveries of babies constitute a \emph{natural} two-stage, two-arm randomized controlled trial with noncompliance and present a unique chance to investigate an optimal dynamic treatment regime. 
See Figure \ref{fig: IV DAG application} for the directed acyclic graph (DAG) illustrating this application.
We will revisit the application after developing theory and methodology precisely suited for this purpose.

\section{
Estimating Individualized Treatment Rules with an Instrumental Variable
}
\label{sec: review}
\subsection{Optimal and NUCA-Optimal ITRs} \label{subsec:opt_and_nuca_opt_ITR}
We first define the estimand of interest, an \emph{optimal ITR}, under the potential outcome framework (\citealp{neyman1923application, rubin1974estimating}), and briefly review how to estimate the optimal ITR under the no unmeasured confounding assumption. 

Consider a single-stage decision problem where one observes the covariates $\bsX = \bsx \in \calX \subseteq\bbR^{d}$, takes a binary action $a\in\{\pm 1\}$, and receives the response $Y (a)\in \calY \subseteq\bbR$. 
Here, $Y(a)$ denotes the potential outcome under action $a$.
This decision-making process is formalized by the notion of individualized treatment rule (or a single-stage policy), which is a map $\pi$ from available prognostic variables $\bsx$ to a treatment decision $a = \pi(\bsx)$.
Denote by $p^\star_a(\cdot|\bsx)$ the conditional law of $Y(a)$ given the realization of $\bsX = \bsx$. 
The quality of $\pi$ is quantified by its \emph{value}:
\begin{equation}
\label{eq: value in single stage}
  V^\pi =  \bbE[Y(\pi(\bsX))] = 
  \bbE_{\bsX}\bbE_{Y \sim p^\star_{\pi(\bsX)}(\cdot|\bsX)}[Y],
\end{equation}  
where the outer expectation is taken with respect to the law of the covariates $\bsX$ and the inner expectation with respect to the potential outcome distribution $p^\star_{a}(\cdot|\bsX)$ with $a$ set to $\pi(\bsX)$. Intuitively, $V^\pi$ measures the expected value of the response if the population were to follow $\pi$. 
Given a class of candidate ITRs $\Pi$, which we refer to as a policy class, let $V^\star = \max_{\pi \in \Pi} V^\pi$ denote the maximal value of an ITR when restricted to $\Pi$.
An ITR is said to be optimal with respect to $\Pi$ if it achieves $V^\star$. 

A well-known result of \citet{zhang2012estimating} (see also \citealt{zhao2012estimating}) asserts the {duality} between value maximization and risk minimization, in the sense that any optimal ITR that maximizes the value $V^\pi$ also minimizes the following \emph{risk} (and vice versa):
\begin{equation}
\label{eq: risk in single stage}
\risk^\pi= \bbE_{\bsX}\left[ |\contr^\star(\bsX)|\cdot \indc{ \pi(\bsX) \neq \sgn(\contr^\star(\bsX)) } \right],
\end{equation}
where $\contr^\star(\bsx) = \bbE_{Y\sim p^\star_{+1}(\cdot| \bsx)}\left[Y\right] - \bbE_{Y\sim p^\star_{-1}(\cdot|\bsx)}\left[Y\right]$ is the CATE.

It is not hard to show that the sign of the CATE, $\sgn(\contr^\star(\bsx))$, is the Bayes ITR (i.e., the optimal ITR when $\Pi$ consists of all Boolean functions). 
Thus, the risk $\risk^\pi$ admits a natural interpretation as a weighted misclassification error: if the decision made by $\pi$ disagrees with the Bayes ITR, then $|\contr^\star(\bsx)|$-many units of loss are incurred.

Suppose we have a dataset consisting of i.i.d.~samples from the law of the triplet $(\bsX^\obs, A^{\obs}, Y^\obs)$.
In parallel to the potential outcome distribution $p^\star_{a}(\cdot|\bsx)$, let $p^\obs_a(\cdot|\bsx)$ denote the conditional law of $\{Y^\obs|A^\obs=a, \bsX^\obs = \bsx\}$. One can then define counterparts of the value and the risk as in \eqref{eq: value in single stage} and \eqref{eq: risk in single stage} respectively, but with $p^\star_a(\cdot|\bsx)$ replaced by $p^\obs_a(\cdot|\bsx)$. For instance, we can define
\begin{equation}
\label{eq:nuca_opt_single_stage}
  \risk^{\obs, \pi}= \bbE_{\bsX}\left[ |\contr^\obs(\bsX)|\cdot \indc{ \pi(\bsX) \neq \sgn(\contr^\obs(\bsX)) } \right], 
\end{equation}
where 
$\contr^\obs(\bsx) = \bbE_{Y\sim p^\obs_{+1}(\cdot| \bsx)}\left[Y\right] - \bbE_{Y\sim p^\obs_{-1}(\cdot|\bsx)}\left[Y\right]$.
Unlike $p^\star_{a}(\cdot|\bsx)$, which is defined on the potential outcomes, 
the distribution $p^{\obs}_{a}(\cdot|\bsx)$ and hence $\risk^{\obs, \pi}$ are \emph{always} identified from the observed data, thus rendering the task of minimizing $\calR^{\obs, \pi}$ feasible using the observed data. 
Under a version of the no unmeasured confounding assumption, the two distributions $p^\star_a(\cdot|\bsX)$ and $p^\obs_a(\cdot|\bsX)$ agree, and thus a minimizer of $\risk^{\obs, \pi}$ is indeed an optimal ITR in that it also minimizes $\risk^\pi$.
To make the distinction clear, we refer to any minimizer of $\calR^{\obs, \pi}$ as a {NUCA-optimal} ITR (i.e., it is only optimal in the conventional sense under the NUCA) and denote it as $\pi_\nuca$.
Many estimation strategies targeting the NUCA-optimal ITR have been proposed in the literature; see, e.g., structural equation models and its variants (\citealp{murphy2001marginal,murphy2003optimal}), outcome weighted learning and its variants (\citealp{zhao2012estimating, zhao2019efficient, athey2020policy}), tree-based methods (\citealp{laber2015tree, zhang2018interpretable}), among others.


\subsection{Instrumental Variables and IV-Optimal ITRs}
\label{subsec: IV optimality review}
The no unmeasured confounding assumption is often a heroic assumption when data come from retrospective observational studies and should be made with caution. 
Indeed, \citet{pu2020estimating} demonstrates via extensive simulations that a NUCA-optimal ITR could have poor generalization performance when the NUCA fails. 

In classical causal inference literature, an instrumental variable is a widely-used tool to infer the causal effect from noisy observational data (\citealp{AIR1996, rosenbaum2002covariance, Imbens2004, hernan2006instruments}). 
A random variable $Z$ is said to be a valid IV if it satisfies the {core IV assumptions}: Stable Unit Treatment Value Assumption (SUTVA), correlation between IV and treatment, exclusion restriction (ER), and IV unconfoundedness conditional on the observed covariates (\citealp{AIR1996, baiocchi2014instrumental}). 
In general, even a valid IV cannot point identify the mean conditional potential outcomes $\mathbb{E}[Y(\pm 1)|\bsX=\bsx]$ or the CATE $\contr^\star(\bsx)$ (\citealp{robins1996identification, balke1997bounds, manski2003partial,swanson2018partial}); 
hence, neither the value nor the risk of an ITR can be point identified with an IV without additional assumptions. \cite{wang2018bounded}, \cite{cui2019semiparametric} and \cite{qiu2020optimal} establish a set of sufficient conditions, under which the CATE can be point identified and the optimal ITR can be estimated with an IV.

Despite the incapability of point identifying the causal effects, a valid IV can still be useful
in that even under minimal identification assumptions, it can produce meaningful \emph{partial identification intervals/bounds} for the CATE. 
That is, one can construct two functions $L(\bsx)$ and $U(\bsx)$, both of which are functionals of the observed data distribution (and thus estimable from the data), such that the CATE satisfies $\contr^\star(\bsx)\in [L(\bsx), U(\bsx)]$ almost surely.
Important examples include the Balke-Pearl bounds \citep{balke1997bounds} and the Manski-Pepper bounds \citep{manski1998monotone}. 

Given the partial identification interval $[L(\bsx), U(\bsx)]$, the risk of an ITR $\pi$ defined in \eqref{eq: risk in single stage} is bounded between 
\begin{equation*}
  \underline{\risk}^\pi = \bbE_{\bsX}\bigg[ \inf_{\contr(\bsX)\in [L(\bsX), U(\bsX)]}|\contr(\bsX)|\cdot \indc{ \pi(\bsX) \neq \sgn(\contr(\bsX)) } \bigg]
\end{equation*}  
and 
\begin{equation}
  \label{eq:worst_case_risk_single_stage}
  \overline{\risk}^\pi = \bbE_{\bsX}\bigg[ \sup_{\contr(\bsX)\in [L(\bsX), U(\bsX)]}|\contr(\bsX)|\cdot \indc{ \pi(\bsX) \neq \sgn(\contr(\bsX)) } \bigg].
\end{equation}
\citet{pu2020estimating} argued that a sensible criterion is to minimize the expected worst-case risk $\overline{\mathcal{R}}^{\pi}$, and the resulting minimizer is termed an \emph{IV-optimal ITR}, as this ITR is ``worst-case risk-optimal'' with respect to the partial identification interval induced by an IV and its associated identification assumptions. 
Note that an IV-optimal ITR is not an optimal ITR without further assumptions.

\cite{cui2021machine} proposed an alternative set of optimality criteria from the perspective of the partially identified value function. 
Instead of constructing two functions $L(\bsx), U(\bsx)$ that contains the CATE, one may alternatively construct functions $\underline{Q}(\bsx, a)$, $\overline{Q}(\bsx, a)$ that sandwiches the conditional mean potential outcome $\bbE[Y(a)|\bsX=\bsx]$, and the value defined in \eqref{eq: value in single stage} satisfies $\bbE[\underline{Q}(\bsX, \pi(\bsX))]\leq V^\pi \leq \bbE[\overline{Q}(\bsX, \pi(\bsX))]$. 
\cite{cui2021machine} advocated maximizing some carefully-chosen middle ground between the lower and upper bounds of the value: 
\begin{equation}
\label{eqn: middle ground of value interval}
  V^{\pi}_{\lambda} = \bbE_{\bsX}\bigg[ \lambda(\bsX, \pi(\bsX))\cdot \underline{Q}(\bsX, \pi(\bsX)) + [1-\lambda(\bsX, \pi(\bsX))] \cdot  \overline{Q}(\bsX, \pi(\bsX))\bigg],
\end{equation} 
where $\lambda(\bsx, a)$ is a pre-specified function that captures a second-level individualism, i.e., how optimistic/pessimistic individuals with covariates $\bsX = \bsx$ are. 

As pointed out by \cite{cui2021machine}, minimizing the maximum risk $\overline{R}^\pi$ is \emph{not} equivalent to maximizing the minimum value $V^\pi_{\lambda=0}$, even when the bounds $L(\bsx), U(\bsx)$ of the CATE are obtained via bounds on the conditional values, i.e., $L(\bsx) = \overline{Q}(\bsx, +1) - \underline{Q}(\bsx, -1)$ and $U(\bsx) = \underline{Q}(\bsx, +1) - \overline{Q}(\bsx, -1)$.
Rather, minimizing $\overline{R}^\pi$ is equivalent to maximizing the midpoint of the minimum and maximum values, namely $V^\pi_{\lambda= 1/2}$.

\subsection{A General IV-Optimality Framework}
\label{subsec:ITR general framework}

In this section, we present a general framework that incorporates the extra information in IVs for better policy estimation. Conceptually, a valid IV and the associated identification assumptions impose \emph{distributional constraints} on potential outcome distributions $p^\star_a(\cdot|\bsx)$. For example, under assumptions leveraged in \cite{cui2019semiparametric} and \cite{qiu2020optimal}, $p^\star_a(\cdot|\bsx)$ can be expressed as functionals of the observed data distribution. 
As another example, if less stronger assumptions are imposed so that point identification is impossible, the partial identification results assert that $p^\star_a(\cdot|\bsx)$ is ``weakly bounded'', in the sense that for a sufficiently regular function $f$, we can find two functions $\underline{Q}(\bsx, a; f)$ and $\overline{Q}(\bsx, a; f)$ such that 
\begin{equation}
    \label{eq:weakly_boundedness}
    \underline{Q}(\bsx, a; f) \leq \int f(y) p^\star_a(dy|\bsx) \leq \overline{Q}(\bsx, a; f).
\end{equation}    
If $f$ is the identity function, then the above display is precisely the partial identification intervals of the mean conditional potential outcome $\bbE[Y(a)|\bsX=\bsx]$ that appeared in \eqref{eqn: middle ground of value interval}.
In both examples, a valid IV along with the identification assumptions allows us to specify a collection of distributions $\calP_{\bsx, a}$, so that 
$p^\star_a(\cdot|\bsx) \in \calP_{\bsx, a}$. In the first example, the set $\calP_{\bsx, a}$ is a singleton consisting of the ground truth potential outcome distribution $p^\star_a(\cdot|\bsx)$, whereas in the second example, the set $\calP_{\bsx, a}$ consists of all distributions that are weakly bounded in the sense of \eqref{eq:weakly_boundedness}. 
In words, $\calP_{\bsx, a}$ is the collection of all possible potential outcome distributions $p^\star_{a}(\cdot|\bsx)$ that are compatible with the putative IV and the associated IV identification assumptions, and we refer to it as an \emph{IV-constrained set}.

We may further equip an IV-constrained set $\calP_{\bsx, a}$ with a prior distribution $\scrP_{\bsx, a}$. Here, $\scrP_{\bsx,a}$ is a probability distribution on $\calP_{\bsx, a}$, the latter of which itself is a collection of probability distributions.  
For readers with a machine learning background, this is reminiscent of Baxter's model of inductive bias learning \citep{baxter2000model}: in his language, $\calP_{\bsx, a}$ is called an \emph{environment}, whose elements are called \emph{tasks}, and it is assumed that nature can sample a task from $\scrP_{\bsx, a}$, which is a probability distribution on the environment.
To have a fully rigorous treatment, we equip $\calP_{\bsx, a}$ with a metric (e.g., Wasserstein metric) and work with the induced Borel sigma algebra (\citealp{parthasarathy2005probability}). 

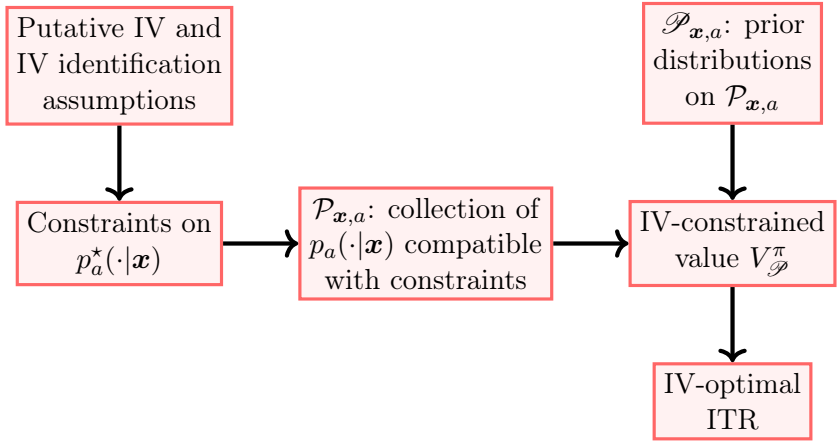
\begin{figure}[t!]
    \centering
    \begin{tikzpicture}[
squarednode/.style={rectangle, draw=red!60, fill=red!5, very thick}, 
every text node part/.style={align=center}
]
\node[squarednode] at (1,4) (IV) {Putative IV and \\ IV identification \\assumptions};
\node[squarednode] (constraint) [below=of IV] {Constraints on \\ $p^\star_{a}(\cdot|\bsx)$};
\node[squarednode] (P) [right=of constraint] {$\calP_{\bsx, a}$: collection of \\$p_{a}(\cdot |\bsx)$ compatible \\ with constraints};
\node[squarednode] (Estimand) [right=of P] {IV-constrained \\ value $V^\pi_{\scrP}$};
\node[squarednode] (Prior) [above=of Estimand]{$\scrP_{\bsx,a}$: prior \\ distributions \\ on $\calP_{\bsx, a}$};
\node[squarednode] (Estimator) [below=of Estimand] {IV-optimal \\ ITR};
\draw[ultra thick, ->] (IV) -- (constraint);
\draw[ultra thick, ->] (constraint) -- (P);
\draw[ultra thick, ->] (P) -- (Estimand);
\draw[ultra thick, ->] (Estimand) -- (Estimator);
\draw[ultra thick, ->] (Prior) -- (Estimand);
\end{tikzpicture}
    \caption{\small A schematic plot of IV-optimality. A collection of prior distributions $\{\scrP_{\bsx, a}: \bsx\in \bbR^d, a = \pm 1\}$ is imposed on the IV-constrained sets $\{\calP_{\bsx, a}: \bsx\in \bbR^d, a = \pm 1\}$, which gives rise to the IV-constrained value $V^\pi_{\scrP}$. An ITR is IV-optimal with respect to this collection of priors if it maximizes $V^\pi_{\scrP}$ over a policy class $\Pi$.}
    \label{fig: scheme flowchat}
\end{figure}

Given the prior distributions $\{\scrP_{\bsx, a}: \bsx\in\calX, a =\pm 1\}$, we can define the \emph{IV-constrained value} of a policy $\pi$ as
\begin{equation}
  \label{eq:itr_general_estimand}
 V^\pi_{\scrP} = 
  \bbE_{\bsX}\bbE_{p_{\pi(\bsX)} \sim \scrP_{\bsX, \pi(\bsX)}}\bbE_{Y \sim p_{\pi(\bsX)}}[Y],
\end{equation}
where we write $p_{\pi(\bsx)} = p_{\pi(\bsx)}(\cdot|\bsx)$ for notational simplicity.
From now on, we will refer to the collection of criteria given by maximizing $V^\pi_{\scrP}$ as \emph{IV-optimality}. 
An ITR is said to be \emph{IV-optimal} with respect to the prior distributions $\{\scrP_{\bsx, a}: \bsx\in\bbR^d, a =\pm 1\}$ and a policy class $\Pi$ if it maximizes $V^\pi_\scrP$ among all $\pi\in\Pi$. Flow chart in Figure \ref{fig: scheme flowchat} summarizes this conceptual framework.

It is clear that the formulation in \eqref{eqn: middle ground of value interval} can be recovered by carefully choosing the prior distributions. 
In particular, let $\underline{p}_a(\cdot|\bsx)$ and $\overline{p}_a(\cdot|\bsx)$ be distributions that witness the partial identification bounds $\underline{Q}(\bsx, a)$ and $\overline{Q}(\bsx, a)$, respectively:
\begin{equation}
    \label{eq:witness_ub_and_lb_single_stage}
\underline{Q}(\bsx, a) = \int y \underline{p}_a(dy|\bsx), \qquad \underline{Q}(\bsx, a) = \int y \overline{p}_a(dy|\bsx).
\end{equation}
Then the criterion \eqref{eqn: middle ground of value interval} is recovered by considering the following two-point priors:
\begin{equation*}
    \scrP^{\texttt{two-point}}_{\bsx,a}
    =
    \lambda(\bsx, a)\cdot \delta_{\underline{p}_a(\cdot|\bsx)}
    + 
    [1-\lambda(\bsx, a)]\cdot \delta_{\overline{p}_a(\cdot|\bsx)},
\end{equation*}
where $\delta_p$ is a point mass at $p$.
As discussed near the end of Section \ref{subsec: IV optimality review}, setting $\lambda(\bsx, a)$ uniformly equal to $1/2$ recovers the original IV-optimality criterion considered in \citet{pu2020estimating} that minimizes the worst-case risk \eqref{eq:worst_case_risk_single_stage}. In fact, under certain regularity conditions, one can show that the reverse is also true: the formulation \eqref{eq:itr_general_estimand} for a specified collection prior distributions can also be recovered from \eqref{eqn: middle ground of value interval} by a careful choice of $\lambda$. 
In view of such an equivalence, the criterion \eqref{eq:itr_general_estimand} should \emph{not} be regarded as a generalization of \eqref{eqn: middle ground of value interval}. Rather, it is a convenient tool amenable to being generalized to the multiple-stage setting. A proof of this equivalence statement will appear in Section \ref{sec: extension to DTR}. We defer how to estimate an IV-optimal ITR to Section \ref{sec: extension to DTR}.

\section{Improving Individualized Treatment Rules with an Instrumental Variable}
\label{sec: ITR IV improve}
Compared to estimating an optimal ITR, a less ambitious goal it to improve upon a baseline ITR $\pibase$, so that the improved ITR is no worse and potentially better than $\pibase$. 
In this section, we show how to achieve this goal with an IV in the single-stage setup, and prepare readers for our main results concerning policy improvement in the general multiple-stage setup in Section \ref{sec:dtr_improve}.


The goal of ``never being worse'' is reminiscent of the min-max risk criterion in \eqref{eq:worst_case_risk_single_stage}.
In view of this, it is natural to consider minimizing
the maximum \emph{excess risk} with respect to the baseline ITR $\pibase$, subject to the IV-informed partial identification constraints. This strategy is summarized in the following definition.

\begin{definition}[Risk-based IV-improved ITR]
\label{def:risk_based_improvement}
Let $\pibase$ denote a baseline ITR, $\Pi$ a policy class, and $[L(\bsx), U(\bsx)]$ an IV-informed partial identification interval of the CATE $\contr^\star(\bsx)$.
A \emph{risk-based IV-improved ITR} is any solution to the following optimization problem:
\begin{equation}
    \label{eq:pi_improved_wrt_any_baseline}
    \min_{\pi\in\Pi}\bbE_{\bsX}\bigg[ \sup_{\contr(\bsX)\in [L(\bsX), U(\bsX)]}|\contr(\bsX)|
    \cdot \bigg(\indc{ \pi(\bsX) \neq \sgn(\contr(\bsX)) } - \indc{ \pibase(\bsX) \neq \sgn(\contr(\bsX)) }\bigg) \bigg].
\end{equation}
\end{definition}



When $\Pi$ consists of all Boolean functions, \eqref{eq:pi_improved_wrt_any_baseline} admits the following explicit solution.

\begin{proposition}[Formula for risk-based IV-improved ITR]
\label{prop:itr_policy_imp_bayes_policy_classification}
Define
    \begin{equation}
    \label{eq: bayes rule}
        \pirisk(\bsx)= \indc{L(\bsx) > 0} - \indc{U(\bsx)< 0} + \indc{L(\bsx)\leq 0\leq U(\bsx)}\cdot \pibase(\bsx).
    \end{equation}
    Then $\pirisk$ is a risk-based IV-improved ITR when $\Pi$ consists of all Boolean functions.
\end{proposition}

The ITR in \eqref{eq: bayes rule} admits a rather intuitive explanation: it takes action $+1$ when the partial identification interval is positive (i.e., $0 < L(\bsx) \leq  U(\bsx)$), and it takes action $-1$ when the interval is negative (i.e., $L(\bsx) \leq U(\bsx) < 0$), and it follows the baseline ITR $\pibase$ otherwise. 

A closely related criterion to that appeared in Definition \ref{def:risk_based_improvement} is to maximize the minimum ``excess value'' with respect to $\pibase$, subject to the distributional constraints imposed by the putative IV and its associated identification assumptions.
This strategy is detailed as follows.

\begin{definition}[Value-based IV-improved ITR]
\label{def:value_based_improvement}
Let $\pibase$ denote a baseline ITR, $\Pi$ a policy class, and $\{\calP_{\bsx, a}: \bsx\in \bbR^d, a\pm 1\}$ a collection of IV-constrained sets.
A \emph{value-based IV-improved ITR} is any solution to the following optimization problem:
\begin{equation}
\label{eq:PIO_value}
    \max_{\pi\in\Pi}
    \bbE_{\bsX}\bigg[ 
    \inf_{ \substack{
    p_{\pi(\bsX)}\in\calP_{\bsX, \pi(\bsX)}
    \\ p_{\pibase(\bsX)}\in\calP_{\bsX, \pibase(\bsX)} }
    }
    \bigg\{\bbE_{\substack{Y\sim p_{\pi(\bsX)} , Y'\sim p_{\pibase(\bsX)}}} \left[Y - Y'\right] \bigg\} \bigg].
\end{equation}
\end{definition}

The above formulation is known as distributionally robust optimization in the optimization literature \citep{delage2010distributionally}. 
When the IV-constrained sets are derived from partial identification intervals $[\underline{Q}(\bsx, a), \overline{Q}(\bsx, a)]$ and $\Pi$ consists of all Boolean functions, the optimization problem \eqref{eq:PIO_value} admits the explicit solution below, analogous to the one given in Proposition \ref{prop:itr_policy_imp_bayes_policy_classification}.



\begin{proposition}[Formula for value-based IV-improved ITR]
\label{prop:itr_policy_imp_bayes_policy_value}
    Assume $\underline{p}_a(\cdot|\bsx)$ and $\overline{p}_a(\cdot|\bsx)$, the two distributions defined in \eqref{eq:witness_ub_and_lb_single_stage} that witness the partial identification bounds $\underline{Q}(\bsx, a)$, $\overline{Q}(\bsx, a)$, are both inside $\calP_{\bsx, a}$.
    Define
    \begin{equation}
    \label{eq: bayes rule II}
        \pival(\bsx)= \indc{\Lval(\bsx) > 0} - \indc{\Uval(\bsx)< 0} + \indc{\Lval(\bsx)\leq 0\leq \Uval(\bsx)}\cdot \pibase(\bsx).
    \end{equation}
    where
        $\Lval(\bsx) = \underline{Q}(\bsx, +1) - \overline{Q}(\bsx, -1)$ and
        $\Uval(\bsx) = \overline{Q}(\bsx, +1) - \underline{Q}(\bsx, -1)$.
    Then $\pival$ is a value-based IV-improved policy when $\Pi$ consists of all Boolean functions.
\end{proposition}

Propositions \ref{prop:itr_policy_imp_bayes_policy_classification} and \ref{prop:itr_policy_imp_bayes_policy_value} together reveal an interesting duality between worst-case excess risk minimization and worst-case excess value maximization. 
When the partial identification interval for the CATE is derived directly from the partial identification intervals for the mean conditional potential outcomes, so that $\Lval= L$ and $\Uval = U$, then the two ITRs defined in \eqref{eq: bayes rule} and \eqref{eq: bayes rule II} agree, and they simultaneously satisfy the two IV-improvement criteria presented in Definitions \ref{def:risk_based_improvement} and \ref{def:value_based_improvement}. Hence, we will not distinguish between two types of IV-improved ITRs.

Such a duality is reminiscent of the duality between risk minimization and value maximization in the classical policy estimation problem under the NUCA discussed in Section \ref{subsec:opt_and_nuca_opt_ITR}. 
Moreover, as discussed in Section \ref{subsec: IV optimality review}, minimizing maximum risk and maximizing minimum value subject to IV-informed partial identification intervals are not equivalent. 
Thus, it is curious to see that such a duality is restored in the policy improvement problem. We defer a discussion of how to estimate IV-improved ITRs to the more general dynamic treatment regimes setting studied in Section \ref{sec:dtr_improve}.


We conclude this section with a comparison between the IV-improved ITR and the improved ITR derived in a series of works by \citet{kallus2018interval,kallus2018confounding,kallus2020minimax}.
In this line of work, the authors considered minimizing the maximum excess risk subject to a ``Rosenbaum-bounds-type'' sensitivity analysis model, and the expression of their improved ITR is similar to \eqref{eq:pi_improved_wrt_any_baseline} and \eqref{eq:PIO_value}, but with $L(\bsx), U(\bsx)$ replaced by the bounds under the sensitivity analysis model.
Despite the apparent similarity, there is a profound, practical difference between sensitivity-analysis-based and IV-based policy improvement. 
Since a sensitivity analysis model only \emph{relaxes} the NUCA, the CATE derived under NUCA (i.e., $\contr^\obs$) is always contained in a sensitivity analysis model. Hence, their framework can never improve upon the NUCA-optimal rule $\pi^\nuca$, defined as the minimizer of \eqref{eq:nuca_opt_single_stage}.
More explicitly, let $\imp^\IV$ be a (potentially multi-valued) \emph{policy improvement operator} that sends a baseline ITR $\pibase$ to its IV-improved counterparts, and let $\imp^{\sens}$ be the corresponding policy improvement operator considered in \cite{kallus2018confounding}. 
We necessarily have $\pi^\nuca \in \imp^\sens(\pi^\nuca)$, meaning that the NUCA-optimal rule can never be improved by $\imp^\sens$.
On the other hand, we will demonstrate via extensive simulations in Section \ref{sec: simulation} that $\imp^{\IV}(\pi^\nuca)$ may yield a strictly better ITR than $\pi^\nuca$ as an IV contains additional information.

\section{Estimating Dynamic Treatment Regimes with an Instrumental Variable}
\label{sec: extension to DTR}

\subsection{Optimal and SRA-Optimal DTRs}
\label{subsec: DTR setup}
We now consider a general $K$-stage decision making problem \citep{murphy2003optimal,schulte2014q}. 
At the first stage we are given baseline covariates $\bsX_1 = \bsx_1 \in \calX$, make a decision $a_1 = \pi_1(\bsx_1)$, and observe a reward $R_1(a_1) = r_1$. 
At stage $k = 2, \cdots, K$, let $\vec{a}_{k-1} = (a_1, \cdots, a_{k-1})$ denote treatment decisions from stage $1$ up to $k-1$, and $\vec{R}_{k - 1}(\vec{a}_{k-1}) = (R_1(a_1), R_2(\vec{a}_2), \cdots, R_{k-1}(\vec{a}_{k-1}))$ the rewards from stage $1$ to $k-1$. 
Meanwhile, denote by $\bsX_k(\vec{a}_{k - 1})$ new covariate information (e.g., time-varying covariates, auxiliary outcomes, etc) that would arise after stage $k-1$ but before stage $k$ if the patient were to follow the treatment decisions $\vec{a}_{k-1}$, and $\Vec{\bsX}_k(\vec{a}_{k-1}) = (\bsX_1, \bsX_2(a_1), \cdots, \bsX_k(\vec{a}_{k-1}))$ the entire covariate history up to stage $k$. 
Given the realizations of $\vec R_{r-1} = \vec r_{k-1}, \vec\bsX_{k} = \vec\bsx_k$,
a possibly data-driven decision $a_k = \pi_k(\vec{a}_{k - 1}, \vec{r}_{k-1}, \vec{\bsx}_k)$ is then made and we observe reward $R_k(\vec a_k) = r_k$. 
Our goal is to estimate a dynamic treatment regime $\pi = \{\pi_k\}_{k=1}^K$, such that the expected value of cumulative rewards $\sum_{k=1}^K r_k$ is maximized if the population were to follow $\pi$.

To simplify the notations, the historical information available for making the $k$-th decision at stage $k$ is denoted by
  $
  \vec{\bsH}_1 = \bsX_1
  $
  for $k = 1$
  and 
  $
  \vec\bsH_k = \big(\vec a_{k-1}, \vec R_{k-1}(\vec a_{k-1}), \vec\bsX_k(\vec a_{k-1})\big)
  $  
  for any 
  $
  2\leq k \leq K.
  $
The information $\vec\bsH_k$ consists of $\vec \bsH_{k-1}$ and $\bsH_k $: $\vec\bsH_{k-1}$ is the information available at the previous stage and $\bsH_k = \big(a_{k-1}, R_{k-1}(\vec{a}_{k-1}), \bsX_k(\vec a_{k-1})\big)$ is the new information generated after decision $a_{k-1}$ is made. 
Let $p^\star_{a_K}(\cdot|\vec\bsh_K)$ be the conditional law of $R(\vec a_K)$ given a specific realization of the historical information $\vec\bsH_K = \vec\bsh_k$.
We define the following \emph{action-value function} (or $Q$-function) at stage $K$:
\begin{equation}
  \label{eq:true_Q_stage_K}
  \Qfunc_K(\vec\bsh_K, a_K)  = \bbE[R_K(\vec a_K) \mid \vec \bsH_K=\vec\bsh_k] = \bbE_{R_K \sim p^\star_{a_K}(\cdot |\vec\bsh_K)}[R_K],
\end{equation}
where we emphasize that the expectation is taken over the potential outcome distribution $p^\star_{a_K}(\cdot|\vec\bsh_K)$.
For a policy $\pi$, its \emph{value function} at stage $K$ is then defined as 
\begin{equation*}
  \Vfunc_K^\pi(\vec\bsh_K)  = Q_K\big( \vec\bsh_K, \pi_K(\vec\bsh_K) \big).
\end{equation*}
Note that the value function at stage $K$ depends on $\pi$ only through $\pi_K$.

Next, we define $Q$-functions and value functions at a generic stage $k\in[K]$ recursively (denote $[K] = \{1, \hdots, K\}$). In particular, for $k = K-1, \hdots, 1$, we let $p^\star_{a_k}(\cdot, \cdot | \vec\bsh_k)$ denote the joint law of the reward $R_k(\vec a_k)$ and new covariate information $\bsX_{k+1}(\vec a_k)$ observed immediately after decision $a_k$ has been made conditional on $\vec\bsH_k = \vec\bsh_k$.
The $Q$-function of $\pi$ at stage $k$ are then defined as: 
\begin{align}
  \Qfunc^\pi_k(\vec\bsh_k, a_k) & = \bbE [R_k(\vec a_k) + {\Vfunc^\pi_{k+1}(\vec\bsH_{k+1})}
  \mid {\vec\bsH_k = \vec\bsh_k}]\nonumber\\
  \label{eq:true_Q_func_k}
  & = \bbE_{(R_k, \bsX_{k+1})\sim p^\star_{a_k}(\cdot, \cdot|\vec\bsh_k)}[R_k + \Vfunc^\pi_{k+1}(\vec\bsH_{k+1})],
\end{align}
where $\vec \bsH_{k+1} = (\vec\bsh_k , a_{k}, R_{k}, \bsX_{k+1})$ is a function of $R_k$ and $\bsX_{k+1}$.
The corresponding value function of $\pi$ is taken to be 
$\Vfunc^\pi_k(\vec\bsh_k)  =  \Qfunc^\pi_k\big(\vec\bsh_k, \pi_k(\vec\bsh_k)\big).$
Similar to \eqref{eq:true_Q_stage_K}, the expectation is taken over the potential outcome distribution $p^\star_{a_k}(\cdot, \cdot|\vec\bsh_k)$. 
We interpret the $Q$-function $Q^\pi_k(\vec\bsh_k, a_k)$ as the cumulative rewards collected by executing $a_k$ at stage $k$ and follow $\pi$ from stage $k+1$ and onwards. 
In contrast, the value function $V^\pi_k(\vec\bsh_k)$ is the cumulative rewards collected by executing $\pi$ from stage $k$ and onwards.
Therefore, $Q^\pi_k$ depends on $\pi$ only through $\pi_{(k+1):K} = (\pi_{k+1},\hdots, \pi_K)$ and $V^\pi_k$ depends on $\pi$ only through $\pi_{k:K}=(\pi_k, \hdots, \pi_K)$.\footnote{For notational simplicity we interpret quantities whose subscripts do not make sense (e.g., $a_0, r_0$ and $\bsx_{K+1}$) as ``null'' quantities and their occurrences in mathematical expressions will be disregarded.}

With a slight abuse of notation, we let $\Pi = \Pi_1\times \cdots \Pi_K$ be a policy class.
A DTR is said to be optimal with respect to $\Pi$ if it maximizes $V^\pi_1(\bsx_1)$ for all fixed $\bsx_1$ (and hence $\bbE[V^\pi_1(\bsX_1)]$ for an arbitrary covariate distribution) over the policy class $\Pi$. 

Assume for now that $\Pi$ consists of all DTRs (i.e., each $\Pi_k$ consists of all Boolean functions). 
A celebrated result from control theory states that the \emph{dynamic programming} approach below, also known as \emph{backward induction}, yields an optimal DTR $\pi^\star$ (\citealp{murphy2003optimal, sutton2018reinforcement}):
\begin{align}
  \label{eq:opt_dtr_dynamic_prog_expression}
  \pi^\star_K(\vec\bsh_K) & = \argmax_{a_K\in\{\pm 1\}} \Qfunc_K(\vec\bsh_K, a_K), 
  ~~~~\pi^\star_k(\vec\bsh_k) = \argmax_{a_k\in\{\pm 1\}} \Qfunc^{\pi^\star}_k(\vec\bsh_k, a_k), ~~ k = K-1, \hdots, 1.
\end{align}
More explicitly, $\pi^\star$ satisfies $\Vfunc^{\pi^\star}_k(\vec\bsh_k) = \max_{\pi} \Vfunc^\pi_k(\vec\bsh_k)$ for any $k\in[K]$ and any configuration of the historical information $\vec\bsh_k$, and is always well-defined as $Q^{\pi^\star}_k$ depends on $\pi^\star$ only through $\pi^\star_{(k+1):K}$.

The foregoing discussion is based on the potential outcome distributions. Suppose that we have collected i.i.d.~data from the law of the random trajectory $(\bsX_k^\obs, A_k^\obs, R_k^\obs)_{k = 1}^K$. 
Let 
$\{p^\obs_{a_k}(\cdot, \cdot|\vec\bsh_k)\}_{k=1}^{K-1}$ 
be the conditional laws of the observed rewards and new covariate information identified from the observed data: 
\begin{equation*}
  (R_K^\obs \mid \vec\bsH_K^\obs=\vec\bsh_k, A^\obs_K = a_K ) \sim p^\obs_{a_K}(\cdot |\vec\bsh_K ), 
  ~~ (R_k^\obs, \bsX_{k+1}^\obs \mid  \vec\bsH_k^\obs = \vec\bsh_k, A_k^\obs = a_k ) \sim p^\obs_{a_k}(\cdot, \cdot | \vec\bsh_k),
\end{equation*}  
for $k\leq K-1$. Here, $\vec\bsH_k^\obs$ denotes the observed historical information observed up to stage $k$.
Suppose that we obtain a DTR using the dynamic programming approach described in \eqref{eq:opt_dtr_dynamic_prog_expression} with $p^\obs_{a_k}$ in place of $p^\star_{a_k}$.
Under a version of the sequential randomization assumptions (\citealp{robins1998marginal}) and with additional assumptions of consistency and positivity, we have
$
  p^\star_{a_k} = p^\obs_{a_k},
$
from which it follows that the DTR obtained is in fact an optimal DTR \citep{murphy2003optimal, schulte2014q}.
In view of this fact, we will refer to this policy as an \emph{SRA-optimal} DTR and denote it as $\pi_\sra$. Methods that estimate SRA-optimal DTRs have been well studied in the literature. See \cite{murphy2003optimal, zhao2015new,tao2018tree,zhang2018c}, among many others.

\subsection{IV-Optimality for DTRs and Dynamic Programming Under Partial Identification}

Suppose that in addition to the observed trajectory $(\bsX_k^\obs, A_k^\obs, R_k^\obs)_{k = 1}^K$, we have access to a time-varying instrumental variable $\{Z_k\}_{k = 1}^K$.
Similar to the single-stage setting in Section \ref{subsec:ITR general framework}, the IV, along with its associated identification assumptions, imposes distributional constraints on the potential outcome distributions $p^\star_{a_k}(\cdot, \cdot|\vec\bsh_k)$. 
That is, we can specify, for each action $a_k$ and historical information $\vec\bsh_k$ at stage $k$, an IV-constrained set $\calP_{\vec\bsh_k, a_k}$, which contains the ground truth potential outcome distribution $p^\star_{a_k}(\cdot, \cdot | \vec\bsh_k)$.
Again, two primary examples are that $\calP_{\vec\bsh_k, a_k}$ is a singleton under point identification (\citealt{michael2020instrumental}), and that $\calP_{\vec\bsh_k, a_k}$ contains weakly bounded distributions by the partial identification intervals in the sense of \eqref{eq:weakly_boundedness}. For ease of exposition, we treat $\calP_{\vec\bsh_k, a_k}$ as a generic set of distributions compatible with the IV and identification assumptions for now; examples of $\calP_{\vec\bsh_k, a_k}$ will be given in Section \ref{subsec: estimate IV-optimal DTR} when we formally describe estimation procedures.
 
It is essential to have a time-varying IV (e.g., daily precipitation) rather than a time-independent IV (e.g., sickle cell trait) in the multiple-stage setting; see Supplementary Materials \ref{subsec: appendix DAG} for details.
 
Given the IV-constrained set $\calP_{\vec\bsh_k, a_k}$, we impose a prior distribution $\scrP_{\vec\bsh_k, a_k}$ on it. This notion generalizes the single-stage setting in Section \ref{subsec:ITR general framework}.
We use $p_{a_k}(\cdot,\cdot|\vec\bsh_k) \sim \scrP_{\vec\bsh_k, a_k}$ to denote sampling a distribution $p_{a_k}(\cdot, \cdot |\vec\bsh_k )\in \calP_{\vec\bsh_k, a_k}$ from $\scrP_{\vec\bsh_k, a_k}$. 
We will use the shorthand $p_{a_k}$ for $p_{a_k}(\cdot,\cdot|\vec\bsh_k)$ where there is no ambiguity.

Under these notation, we introduce \emph{IV-constrained} counterparts of the conventional $Q$- and value functions defined in \eqref{eq:true_Q_stage_K}--\eqref{eq:true_Q_func_k}.

\begin{definition}[IV-constrained $Q$- and value function]\label{def:DTR IV-constrained Q and value function}
For each stage $k\in[K]$, each action $a_k\in\{\pm 1\}$, and each configuration the historical information $\vec\bsH_k$, let $\scrP_{\vec\bsh_k, a_k}$ be a prior distribution on the IV-constrained set $\calP_{\vec\bsh_k, a_k}$. 
The IV-constrained $Q$-function and the corresponding value function of a DTR $\pi$ with respect to the collection prior distributions $\{\scrP_{\vec\bsh_k, a_k}\}$ at stage $K$ are
\begin{align*}
\Qfunc_{\scrP, K}(\vec\bsh_K, a_K) 
&= \bbE_{p_{a_K} \sim \scrP_{\vec\bsh_K, a_K}}\bbE_{R_K \sim p_{a_K}}[R_K],\\
\text{and}\qquad\Vfunc_{\scrP, K}^\pi(\vec\bsh_K) &= \Qfunc_{\scrP, K} \big(\vec\bsh_K,  \pi_K(\vec\bsh_K) \big),
\end{align*}
respectively. 
Recursively, at stage $k = K - 1, K - 2, \cdots, 1$, the IV-constrained $Q$-function and the corresponding value function are
\begin{align*}
  \Qfunc^\pi_{\scrP, k}(\vec\bsh_k, a_k) 
   &= 
   \bbE_{p_{a_k} \sim \scrP_{\vec\bsh_k, a_k}}\bbE_{(R_k, \bsX_{k+1})\sim p_{a_k}}[R_k + \Vfunc^\pi_{\scrP, k+1}(\vec\bsH_{k+1})],\\
  \text{and}\qquad\Vfunc^\pi_{\scrP, k}(\vec\bsh_k) &=  \Qfunc^\pi_{\scrP, k}\big(\vec\bsh_k, \pi_k(\vec\bsh_k)\big),
\end{align*}
respectively, where we recall that $\vec \bsH_{k+1} = (\vec\bsh_k , a_{k}, R_{k}, \bsX_{k+1})$ is a function of $R_k$ and $\bsX_{k+1}$.
\end{definition}

Definition \ref{def:DTR IV-constrained Q and value function} generalizes the notion of IV-constrained value in the single-stage setting.
Similar to the conventional $Q$- and value functions, the IV-constrained $Q$-function $Q^\pi_{\scrP,k}$ depends on $\pi$ only through $\pi_{(k+1):K}$ and the IV-constrained value function $V^\pi_{\scrP, k}$ depends on $\pi$ only through $\pi_{k:K}$.

Definition \ref{def:IV-Optimal DTR} follows from Definition \ref{def:DTR IV-constrained Q and value function} and defines the notion of IV-optimality for DTRs.

\begin{definition}[IV-optimal DTR]\label{def:IV-Optimal DTR}
%
A DTR $\pi^\star$ is said to be IV-optimal with respect to the collection of prior distributions $\{\scrP_{\vec\bsh_k, a_k}\}$ and a policy class $\Pi$ if it satisfies 
\begin{equation}
  \label{eq:iv_opt_dtr}
  \pi^\star \in \argmax_{\pi\in\Pi} \Vfunc^\pi_{\scrP, 1}(\bsx_1)
\end{equation}
for every fixed $\bsh_1 = \bsx_1\in\calX$.
\end{definition}

Note that a priori, we do not know if an IV-optimal policy $\pi^\star$ exists, as the above definition requires $\pi^\star$ to maximize the IV-constrained value function for every fixed $\bsx_1$ (and thus $\bbE_{\bsX_1}[V^\pi_{\scrP, 1}(\bsX_1)]$ for any law of $\bsX_1$). However, as we will see in our first main result below, the optimization problem \eqref{eq:iv_opt_dtr} can be solved via a modified dynamic programming algorithm if $\Pi$ consists of all policies.

\begin{theorem}[Dynamic programming for the IV-optimal DTR]
\label{thm:iv_opt_dtr_dynamic_prog}
Let $\pi^\star$ be recursively defined as follows:
\begin{align*}
  \pi^\star_K(\vec\bsh_K) & = \argmax_{a_K\in\{\pm 1\}} \Qfunc_{\scrP, K}(\vec\bsh_K, a_K), 
  ~~~~\pi^\star_k(\vec\bsh_k) = \argmax_{a_k\in\{\pm 1\}} \Qfunc^{\pi^\star}_{\scrP, k}(\vec\bsh_k, a_k), ~k \leq K-1.
\end{align*}
Then the DTR $\pi^\star$ satisfies
\begin{equation}
  \label{eq:iv_opt_dtr_k}
  V^{\pi^\star}_{\scrP, k}(\vec\bsh_k) = \max_{\pi} V^\pi_{\scrP, k}(\vec\bsh_k)
\end{equation}
for any stage $k$ and any configuration of the historical information $\vec\bsh_k$, where the maximization is taken over all DTRs.
\end{theorem}

The above result is a generalization of the classical dynamic programming algorithm \eqref{eq:opt_dtr_dynamic_prog_expression} under partial identification. Note that the policy $\pi^\star$ in the above theorem is always well-defined as $\Qfunc^{\pi^\star}_{\scrP_k}$ depends on $\pi^\star$ only through $\pi^\star_{(k+1):K}$.

\subsection{An Alternative Characterization of IV-Optimal DTRs} 
\label{subsec:alternative characterization}
Estimating an IV-optimal DTR based on Theorem \ref{thm:iv_opt_dtr_dynamic_prog} alone is difficult, primarily because the definitions of the IV-constrained $Q$- and value functions involve integration over the prior distributions $\{\scrP_{\vec\bsh_k, a_k}\}$. 
In this subsection, we provide an alternative characterization of an IV-optimal DTR which illuminates a practical estimation strategy. 

Recall that in the single-stage setting, there is an equivalence between the IV constrained value \eqref{eq:itr_general_estimand} and the convex combination of the lower and upper bounds of the value \eqref{eqn: middle ground of value interval}. The latter expression \eqref{eqn: middle ground of value interval} is easier to compute as long as the weights of the convex combination are specified. Below, we extend such an equivalence relationship to the multiple-stage setting and use it to get rid of the intractable integration over prior distributions. 

We start by defining the \emph{worst-case} and \emph{best-case} IV-constrained $Q$- and value functions as well as their weighted versions as follows.

\begin{definition}[Worst-case, best-case, and weighted $Q$- and value functions]\label{def:worst/best Q and value function}
Let the prior distributions $\{\scrP_{\vec\bsh_k, a_k}\}$ be specified and let $\{\lambda_k(\vec\bsh_k, a_k)\}$ be a collection of weighting functions taking values in $[0, 1]$.
The worst-case, best-case, and weighted $Q$-functions at stage $K$ with respect to the specified prior distributions and weighting functions are respectively defined as
\begin{align*}
  \label{eq:worst Q stage K}
  &\underline{\Qfunc}_{K}(\vec\bsh_K, a_K) = \inf_{\substack{p_{a_K}\in \calP_{\vec\bsh_K, a_K}}} \bbE_{R_K\sim p_{a_K}} [R_K] , \numberthis\\
  \label{eq:best Q stage K}
  &\overline{\Qfunc}_{K}(\vec\bsh_K, a_K) = \sup_{\substack{p_{a_K} \in \calP_{\vec\bsh_K, a_K} }} \bbE_{R_K\sim p_{a_K}} [R_K], \numberthis\\
  &Q_{\vec{{\lambda}}, K}(\vec\bsh_K, a_K) = \lambda_K(\vec\bsh_K, a_K) \cdot \underline\Qfunc_K(\vec\bsh_K, a_K) + \big(1 - \lambda_K(\vec\bsh_K, a_K) \big)\cdot \overline{\Qfunc}_K(\vec\bsh_K, a_K),
\end{align*}
The corresponding worst-case, best case, and weighted value functions, denoted as $\underline{\Vfunc}_K^\pi(\vec\bsh_K)$, $\overline{\Vfunc}_K^\pi(\vec\bsh_K)$, and $\Vfunc_{\vec{{\lambda}}, K}^\pi(\vec\bsh_K)$, respectively, are
obtained by setting $a_k = \pi_K(\vec\bsh_K)$ in the above displays. 
Recursively, at stage $k = K - 1, \cdots, 1$, we define the worst-case, best-case, and weighted $Q$-functions as
\begin{align}
  \label{eq:worst_Q_func_k}
  \underline\Qfunc^\pi_{\vec{\lambda}, k}(\vec\bsh_k, a_k) & = \inf_{p_{a_k} \in \calP_{\vec\bsh_k, a_k}} \bbE_{(R_k, \bsX_{k+1})\sim p_{a_k}}[R_k + \Vfunc^\pi_{\vec{\lambda}, k+1}(\vec\bsH_{k+1})], \\
  \label{eq:best_Q_func_k}
  \overline\Qfunc^\pi_{\vec{\lambda}, k}(\vec\bsh_k, a_k) & = \sup_{p_{a_k}\in \calP_{\vec\bsh_k, a_k}} \bbE_{(R_k, \bsX_{k+1})\sim p_{a_k}}[R_k + \Vfunc^\pi_{\vec{\lambda}, k+1}(\vec\bsH_{k+1})],\\
  Q_{\vec{\lambda}, k}^\pi(\vec\bsh_k, a_k) 
  & = \lambda_k(\vec\bsh_k, a_k) \cdot \underline\Qfunc_{\vec{\lambda}, k}^\pi(\vec\bsh_k, a_k) + \big(1 - \lambda_k(\vec\bsh_k, a_k) \big)\cdot \overline{\Qfunc}_{\vec{\lambda}, k}^\pi(\vec\bsh_k, a_k), \nonumber
\end{align}
respectively.
Again, replacing $a_k$ in the above $Q$-functions with $\pi_k(\vec\bsh_k)$ yields their corresponding value functions $\underline{\Vfunc}^\pi_{\vec{\lambda}, k}, \overline\Vfunc^\pi_{\vec{\lambda}, k}$, and $\Vfunc^\pi_{\vec{\lambda}, k}$.
\end{definition}

Definition \ref{def:worst/best Q and value function} generalizes the criterion \eqref{eqn: middle ground of value interval}. According to Definition \ref{def:worst/best Q and value function}, the worst-case and best-case $Q$-functions $\underline\Qfunc^\pi_{\vec{\lambda}, k}$ and $\overline\Qfunc^\pi_{\vec{\lambda}, k}$ depend on the weighting functions only through $\lambda_t(\vec\bsh_{t}, a_t)$ for $k+1 \leq t \leq K\}$, whereas the corresponding value functions depend on the weighting functions only through $\lambda_t(\vec\bsh_{t}, a_t)$ for $k \leq t \leq K$. 
For notational simplicity, in the rest of the paper, we add superscript $\pi$ and subscript $\vec\lambda$ in $Q$-functions at stage $K$ (e.g., we write $\underline{Q}_K = Q^\pi_{\vec\lambda, K}$), although they have no dependence on $\pi$ and the weighting functions.

Proposition \ref{prop:struc_of_iv-opt_dtr} below establishes the equivalence between the weighted $Q$-functions (resp.~value functions) and their IV-constrained counterparts in Definition \ref{def:DTR IV-constrained Q and value function}.

\begin{proposition}[Equivalence between weighted and IV-constrained $Q$- and value functions]
\label{prop:struc_of_iv-opt_dtr}
The following two statements hold:
\begin{enumerate}
  \item Fix any collection of weighting functions $\{\lambda_k(\vec\bsh_k, a_k)\}$. Assume that in Definition \ref{def:worst/best Q and value function}, the infimums and supremums when defining weighted $Q$- and value functions are all attained. 
  Then there exists a collection of prior distributions $\{\scrP_{\vec\bsh_k, a_k}\}$ such that
  \begin{equation*}
  \Qfunc^\pi_{\vec{\lambda}, k} = \Qfunc^\pi_{\scrP, k}, \qquad \Vfunc^\pi_{\vec{\lambda}, k} = \Vfunc^\pi_{\scrP, k}, ~~\forall k\in[K].
  \end{equation*}
  \item Reversely, if we fix any collection of prior distributions $\{\scrP_{\vec\bsh_k, a_k}\}$, then there exists a collection of weighting functions $\{\lambda_k(\vec\bsh_k, a_k)\}$ such that the above display holds true.
\end{enumerate}
\end{proposition}

The above theorem effectively translates the problem of specifying a collection of prior distributions to specifying a collection of weighting functions, thus allowing one to bypass the integration over the prior distributions. This theorem, along with the dynamic programming algorithm presented in Theorem \ref{thm:iv_opt_dtr_dynamic_prog}, leads to the following alternative characterization of the IV-optimal DTR.

\begin{corollary}[Alternative characterization of the IV-optimal DTR]
\label{cor:struc_of_iv-opt_dtr}
Under the setting in Part 1 of Proposition \ref{prop:struc_of_iv-opt_dtr}, if we recursively define $\pi^\star$ as
\begin{equation}
  \label{eq:struc_of_iv-opt_dtr}
  \pi^\star_k(\vec\bsh_k) = \sgn\left\{\contr^{\pi^\star}_{\vec \lambda, k}(\vec\bsh_k) \right\} = \sgn\left\{Q_{\vec{\lambda}, k}^{\pi^\star}(\vec\bsh_k, +1) - Q_{\vec{\lambda}, k}^{\pi^\star}(\vec\bsh_k, -1) \right\}, ~~ k = K, K - 1, \hdots, 1,
\end{equation}  
where the two quantities
\begin{align*}
  Q_{\vec{\lambda}, k}^{\pi^\star}(\vec\bsh_k, +1)
  & = \lambda_k(\vec\bsh_k, +1) \cdot \underline{\Qfunc}_{\vec{\lambda} ,k}^{\pi^\star} (\vec\bsh_k, +1) + \big(1-\lambda_k(\vec\bsh_k, +1)\big) \cdot \overline{\Qfunc}_{\vec{\lambda}, k}^{\pi^\star} (\vec\bsh_k, +1),\\
  Q_{\vec{\lambda}, k}^{\pi^\star}(\vec\bsh_k, -1)
  & = \lambda_k(\vec\bsh_k, -1) \cdot \underline{\Qfunc}_{\vec{\lambda}, k}^{\pi^\star} (\vec\bsh_k, -1) + \big(1-\lambda_k(\vec\bsh_k, -1)\big) \cdot \overline{\Qfunc}_{\vec{\lambda}, k}^{\pi^\star} (\vec\bsh_k, -1) 
\end{align*}
depend on $\pi^\star$ only through $\pi^\star_{(k+1):K}$ and are therefore well-defined, then $\pi^\star$ satisfies \eqref{eq:iv_opt_dtr_k} for any stage $k$ and any configuration of the historical information $\vec\bsh_k$.
In particular, $\pi^\star$ is an IV-optimal DTR in the sense of Definition \ref{def:IV-Optimal DTR} when $\Pi$ consists of all Boolean functions.
\end{corollary}



Compared to Theorem \ref{thm:iv_opt_dtr_dynamic_prog}, Corollary \ref{cor:struc_of_iv-opt_dtr} gives an alternative, analytic and constructive characterization of the IV-optimal DTR. 
The proof is an immediate consequence of Theorem \ref{thm:iv_opt_dtr_dynamic_prog} and Proposition \ref{prop:struc_of_iv-opt_dtr}, and omitted.

\paragraph{An illustrative example when $K=2$.}
Figure \ref{fig: illustrate two stages} illustrates the decision process when $K = 2$. At the second stage, given $\vec{\bsh}_2^+ = ({\bsh}_1, a_1=+1, R_1(+1) = r_1^+, \bsX_2(+1) = \bsx_2^+)$ 
and $\lambda_2(\vec\bsh_2^+, \pm 1)$, an IV-optimal action is made based on comparing two weighted $Q$-functions $Q_{\vec\lambda, 2}(\vec\bsh_2^+, +1) = 0.8$ and $Q_{\vec\lambda, 2}(\vec\bsh_2^+, -1) = 0.5$. Since $0.8>0.5$, we have $\pi^\star_2(\vec\bsh_2^+) = +1$. 
The decision given $\vec\bsh_2^- = (\bsh_1, a_1 = -1, R_1(-1) = r_1^-, \bsX_2(-1) = \bsx_2^-)$ and $\lambda_2(\vec\bsh_2^-, \pm 1)$ is similar, and we have $\pi^\star_2(\vec\bsh_2^-) = -1$ in this case.

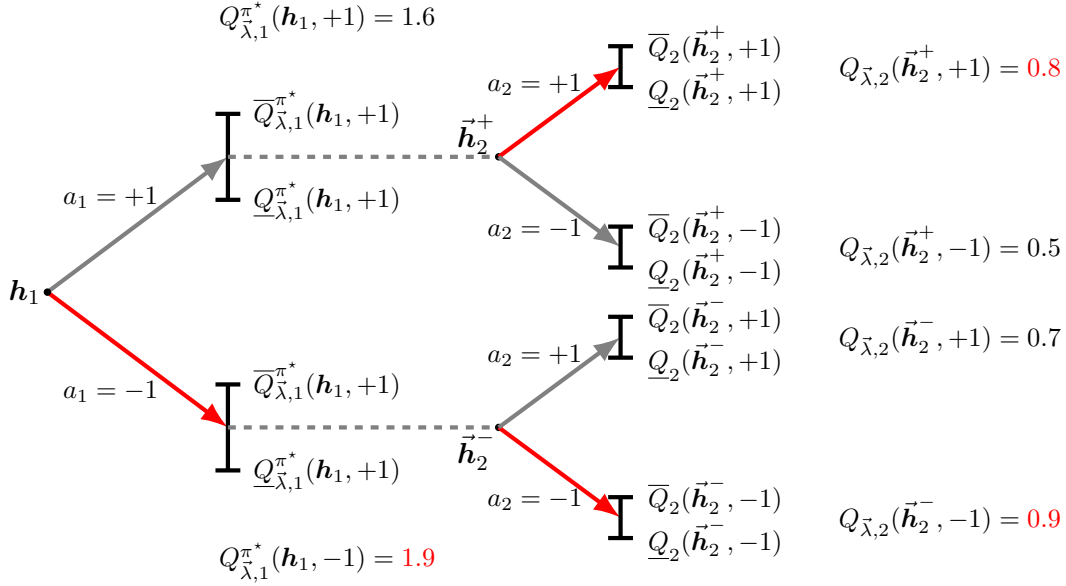
\begin{figure}[t!]
    \centering
    \begin{tikzpicture}[scale=0.6, every text node part/.style={align=center},
    node distance=2cm]

\draw[gray, ultra thick] (-2,0) -- (2,3);
\draw[red, ultra thick] (-2,0) -- (2,-3);
\filldraw[black] (-2,0) circle (2pt);
\node at (-2.5, 0) {${\bsh_1}$};
\node[align=left] at (-0.6,2.1) {\small $a_1 = +1$};
\node[align=left] at (-0.6,-2.2) {\small $a_1 = -1$};
\draw[|-|, ultra thick] (2,2) -- (2,4);
\draw[|-|, ultra thick] (2,-2) -- (2,-4);

\node at (4.2, 6) {\small${\Qfunc}_{\vec\lambda, 1}^{\pi^\star} (\bsh_1, +1) =\textcolor{black}{1.6} $};
\node at (4.2, 4) {\small $\overline{\Qfunc}^{\pi^\star}_{\vec{\lambda}, 1} (\bsh_1, +1) $};
\node at (4.2, 2) {\small $\underline{\Qfunc}^{\pi^\star}_{\vec{\lambda}, 1} (\bsh_1, +1)  $};

\node at (4.2, -2) {\small $\overline{\Qfunc}^{\pi^\star}_{\vec{\lambda}, 1} (\bsh_1, +1)  $};
\node at (4.2, -4) {\small $\underline{\Qfunc}^{\pi^\star}_{\vec{\lambda}, 1} (\bsh_1, +1)  $};
\node at (4.2, -6) {\small${\Qfunc}_{\vec\lambda, 1}^{\pi^\star} (\bsh_1, -1) =\textcolor{red}{1.9} $};

\draw[-, gray, ultra thick, dashed] (2, 3) -- (8, 3);
\draw[-, gray, ultra thick, dashed] (2, -3) -- (8, -3);

\filldraw[black] (8, 3) circle (2pt);
\node at (7.5, 3.5) {$\vec{\bsh}_2^+$};
\draw[red, ultra thick] (8,3) -- (10.7,5);
\draw[gray, ultra thick] (8,3) -- (10.7,1);
\draw[|-|, ultra thick] (10.7,4.5) -- (10.7,5.5);
\draw[|-|, ultra thick] (10.7,0.5) -- (10.7,1.5);
\node[align=left] at (8.8,4.6) {\small $a_2 = +1$};
\node[align=left] at (8.8,1.4) {\small $a_2 = -1$};
\node at (12.8, 5.5) {\small$\overline{\Qfunc}_{2} (\vec\bsh_2^+, +1)$};
\node at (18, 5) {\small${\Qfunc}_{\vec\lambda, 2} (\vec\bsh_2^+, +1) =\textcolor{red}{0.8} $};
\node at (12.8, 4.5) {\small$\underline{\Qfunc}_{2} (\vec\bsh_2^+, +1)  $};
\node at (12.8, 1.5) {\small$\overline{\Qfunc}_{2} (\vec\bsh^+_2, -1)  $};
\node at (18, 1) {\textcolor{black}{\small${\Qfunc}_{\vec\lambda, 2} (\vec\bsh_2^+, -1) = 0.5  $}};
\node at (12.8, 0.5) {\small$\underline{\Qfunc}_{2}(\vec\bsh^+_2, -1)  $};

\filldraw[black] (8, -3) circle (2pt);
\node at (7.5, -3.5) {$\vec{\bsh}^-_2$};
\draw[red, ultra thick] (8,-3) -- (10.7,-5);
\draw[gray, ultra thick] (8,-3) -- (10.7,-1);
\draw[|-|, ultra thick] (10.7,-4.5) -- (10.7,-5.5);
\draw[|-|, ultra thick] (10.7,-0.5) -- (10.7,-1.5);
\node[align=left] at (8.8,-4.6) {\small $a_2 = -1$};
\node[align=left] at (8.8,-1.4) {\small $a_2 = +1$};
\node at (12.8, -4.5) {\small$\overline{ \Qfunc}_{2} (\vec\bsh^-_2, -1)  $};
\node at (18, -5) {\small${\Qfunc}_{\vec\lambda, 2} (\vec\bsh_2^-, -1) =\textcolor{red}{0.9} $};
\node at (12.8, -5.5) {\small$\underline{\Qfunc}_{2} (\vec\bsh^-_2, -1)  $};
\node at (12.8, -0.5) {\small$\overline{ \Qfunc}_{2} (\vec\bsh^-_2, +1)  $};
\node at (18, -1) {\small${\Qfunc}_{\vec\lambda, 2} (\vec\bsh_2^-, +1) =\textcolor{black}{0.7} $};
\node at (12.8, -1.5) {\small$\underline{\Qfunc}_{2} (\vec\bsh^-_2, +1)  $};

\end{tikzpicture}
\caption{
\small
An illustrative example of an IV-optimal DTR when $K=2$. 
According to the values of the weighted $Q$-functions at the second stage, we have $\pi^\star_2(\vec\bsh_2^+) = +1$ and $\pi^\star_2(\vec\bsh_2^-)= -1$.
Given $\pi^\star_2$, the weighted $Q$-functions at the first stage can be computed, and the IV-optimal action is $\pi^\star_1(\bsh_1) = -1$.
}
\label{fig: illustrate two stages}
\end{figure}

Given the knowledge of $\pi^\star_2$, we can compute the weighted value function $V^{\pi^\star}_{\vec\lambda, 2}(\cdot)$ of $\pi^\star$ at the second stage.
This allows us to construct a ``pseudo-outcome'' at the first stage, denoted as $PO_1(r_1, \bsx_2|{\bsh}_1, a_1) = r_1 + V_2^{\pi^\star}({\bsh}_1, a_1, r_1, \bsx_2)$.
Importantly, $PO_1$ depends on $\pi^\star_2$, as it is the cumulative rewards if we observe $\bsh_1$ at the first stage, take an immediate action $a_1$, and then act according to the IV-optimal decision $\pi^\star_2$ at the second stage. 

The partial identification interval for the expected value of $PO_1(R_1(a_1), \bsX_2(a_1)|\bsh_1, a_1)$, where the expectation is taken over the potential outcome distribution of $R_1(a_1)$ and $\bsX_2(a_1)$, is precisely the worst-case and best-case $Q$-functions $\underline{Q}^{\pi^\star}_{\vec\lambda, 1}(\bsh_1, a_1)$ and $\overline{Q}^{\pi^\star}_{\vec\lambda, 1}(\bsh_1, a_1)$. To this end, we can compute the weighted $Q$-function $Q^{\pi^\star}_{\vec\lambda,1}(\bsh_1, a_1)$ if $\lambda_1(\bsh_1, a_1)$ is specified, from which we can decide $\pi^\star_1$. Reading the numbers off Figure \ref{fig: illustrate two stages}, we conclude that $\pi^\star_1(\bsh_1) = -1$.

\begin{figure}[t!]
\centering
\begin{subfigure}[b]{0.30\textwidth}
  \centering
  \begin{tikzpicture}[scale=0.6, node distance=2cm]
  \draw[gray, ultra thick, red] (0,0) -- (2,3);
  \draw[gray, ultra thick] (0,0) -- (2,-3);
  \filldraw[black] (0,0) circle (2pt);
  \node at (-0.5, 0) {\small $\vec{\bsh_k}$};
  \node[align=left] at (0.5,1.8) {\small $+1$};
  \node[align=left] at (0.5,-1.8) {\small $-1$};
  \draw[|-|, ultra thick] (2,2) -- (2,4);
  \draw[|-|, ultra thick] (2,-2) -- (2,-4);
  \node at (2.5, 5.5) {\small ${\Qfunc}^{\pi^\star}_{\vec\lambda, k} (\vec\bsh_k, +1) =\textcolor{red}{1.0} $};
  \node at (5.3, 2) {\small $\underline{\Qfunc}_{\vec{\lambda} ,k}^{ \pi^\star} (\vec\bsh_k, +1) = 1.0$};
  \node at (5.3, 4) {\small $\overline{\Qfunc}_{\vec{\lambda} ,k}^{  \pi^\star} (\vec\bsh_k, +1) = 1.5$};
  \node at (5.3, -2) {\small $\overline{\Qfunc}_{\vec{\lambda} ,k}^{ \pi^\star} (\vec\bsh_k, -1) = 5.0$};
  \node at (5.3, -4) {\small $\underline{\Qfunc}_{\vec{\lambda} ,k}^{\pi^\star} (\vec\bsh_k, -1) = 0.8$};
  \node at (2.5,  -5.5) {\small ${\Qfunc}^{\pi^\star}_{\vec\lambda, k} (\vec\bsh_k, -1) =\textcolor{black}{0.8} $};
  \filldraw[red] (2, 2) circle (3pt);
  \filldraw[red] (2, -4) circle (3pt);
  \end{tikzpicture}
  \caption[]
  {{\footnotesize Worst-Case: $\lambda_k(\vec\bsh_k, \pm 1) = 1$.}}
  \label{fig: illustrate stage k worst case}
\end{subfigure}
\hfill
\begin{subfigure}[b]{0.3\textwidth}  
  \centering 
  \begin{tikzpicture}[scale=0.6, node distance=2cm]
  \draw[gray, ultra thick] (0,0) -- (2,3);
  \draw[gray, ultra thick, red] (0,0) -- (2,-3);
  \filldraw[black] (0,0) circle (2pt);
  \node at (-0.5, 0) {\small $\vec{\bsh_k}$};
  \node[align=left] at (0.5,1.8) {\small $+1$};
  \node[align=left] at (0.5,-1.8) {\small $-1$};
  \draw[|-|, ultra thick] (2,2) -- (2,4);
  \draw[|-|, ultra thick] (2,-2) -- (2,-4);
  \node at (2.5,  5.5) {\small ${\Qfunc}^{\pi^\star}_{\vec\lambda, k} (\vec\bsh_k, +1) =\textcolor{black}{1.5} $};
  \node at (5.3, 2) {\small $\underline{\Qfunc}_{\vec{\lambda} ,k}^{ \pi^\star} (\vec\bsh_k, +1) = 1.0$};
  \node at (5.3, 4) {\small $\overline{\Qfunc}_{\vec{\lambda} ,k}^{  \pi^\star} (\vec\bsh_k, +1) = 1.5$};
  \node at (5.3, -2) {\small $\overline{\Qfunc}_{\vec{\lambda} ,k}^{ \pi^\star} (\vec\bsh_k, -1) = 5.0$};
  \node at (5.3, -4) {\small $\underline{\Qfunc}_{\vec{\lambda} ,k}^{\pi^\star} (\vec\bsh_k, -1) = 0.8$};
  \node at (2.5, -5.5) {\small ${\Qfunc}^{\pi^\star}_{\vec\lambda, k} (\vec\bsh_k, -1) =\textcolor{red}{5.0} $};
  \filldraw[red] (2, 2) circle (3pt);
  \filldraw[red] (2, -2) circle (3pt);
  \end{tikzpicture}
  \caption[]%
  {{\footnotesize Best-Case: $\lambda_k(\vec\bsh_k, \pm 1) = 0$.}}
  \label{fig: illustrate stage k best case}
 \end{subfigure}
 \hfill
\begin{subfigure}[b]{0.3\textwidth}   
  \centering 
  \begin{tikzpicture}[scale=0.6, node distance=2cm]
  \draw[gray, ultra thick] (0,0) -- (2,3);
  \draw[gray, ultra thick, red] (0,0) -- (2,-3);
  \filldraw[black] (0,0) circle (2pt);
  \node at (-0.5, 0) {\small $\vec{\bsh_k}$};
  \node[align=left] at (0.5,1.8) {\small $+1$};
  \node[align=left] at (0.5,-1.8) {\small $-1$};
  \draw[|-|, ultra thick] (2,2) -- (2,4);
  \draw[|-|, ultra thick] (2,-2) -- (2,-4);
  \filldraw[red] (2, 3) circle (3pt);
  \filldraw[red] (2, -3) circle (3pt);
  \node at (2.5,  5.5) {\small ${\Qfunc}^{\pi^\star}_{\vec\lambda, k} (\vec\bsh_k, +1) =\textcolor{black}{1.25} $};
  \node at (5.3, 2) {\small $\underline{\Qfunc}_{\vec{\lambda} ,k}^{\pi^\star} (\vec\bsh_k, +1) = 1.0$};
  \node at (5.3, 4) {\small $\overline{\Qfunc}_{\vec{\lambda} ,k}^{ \pi^\star} (\vec\bsh_k, +1) = 1.5$};
  \node at (5.3, -2) {\small $\overline{\Qfunc}_{\vec{\lambda} ,k}^{\pi^\star} (\vec\bsh_k, -1) = 5.0$};
  \node at (5.3, -4) {\small $\underline{\Qfunc}_{\vec{\lambda},k}^{\pi^\star} (\vec\bsh_k, -1) = 0.8$};
  \node at (2.5,  -5.5) {\small ${\Qfunc}^{\pi^\star}_{\vec\lambda, k} (\vec\bsh_k, -1) =\textcolor{red}{2.90} $};
  \end{tikzpicture}
  \caption[]%
  {{\footnotesize Min-Max: $\lambda_k(\vec\bsh_k, \pm 1) =\frac{1}{2}$.}}
  \label{fig: illustrate stage k min max}
\end{subfigure}
\caption[]
{\small IV-optimal DTR at stage $k$ with different choices of weighting functions.
The weighted $Q$-functions are represented by red dots in each arm, and the corresponding IV-optimal action is colored in red. Figure \ref{fig: illustrate stage k worst case} corresponds to the worst-case perspective, Figure \ref{fig: illustrate stage k best case} the best-case perspective, and Figure \ref{fig: illustrate stage k min max} the min-max perspective. In each case, the IV-optimal action is $\sgn\{Q_{\vec{\lambda}, k}^{\pi^\star}(\vec\bsh_k, +1) - Q_{\vec{\lambda}, k}^{\pi^\star}(\vec\bsh_k, -1) \}$.} 
\label{fig: illustrate stage k}
\end{figure}
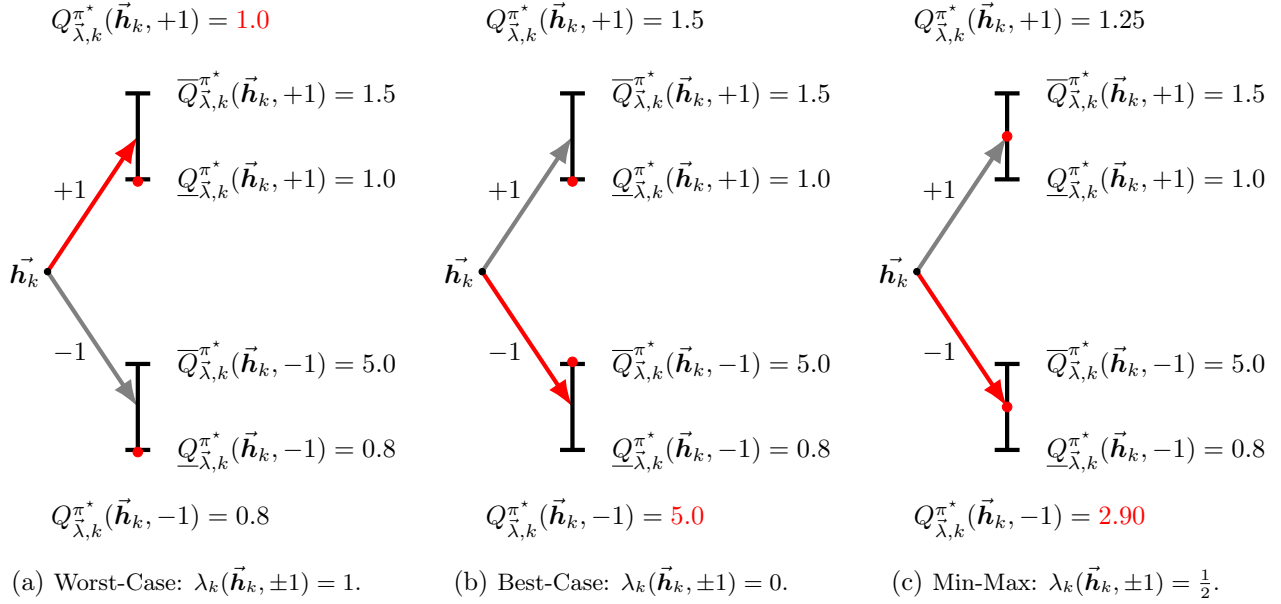

\paragraph{Choosing weighting functions.}
The specification of weighting functions reveals one's level of optimism.
Suppose that the future weighting functions $\{\lambda_t: t\geq k+1\}$ has been specified. At stage $k$, if one adopts a worst-case perspective and would like to maximize the worst-case gain at this stage (fixing the weighting function specifications at all future stages),
then it suffices to compare the two worse-case $Q$-functions at stage $k$, namely $\underline{\Qfunc}_{\vec{\lambda} ,k}^{\pi^\star} (\vec\bsh_k, +1)$ and $\underline{\Qfunc}_{\vec{\lambda} ,k}^{\pi^\star} (\vec\bsh_k, -1)$.
And the pessimistic action is taken to be sign of the difference of the two worst-case $Q$-functions,
which corresponds to taking $\lambda_k(\vec\bsh_k, \pm 1) = 1$ (see Figure \ref{fig: illustrate stage k worst case}). Alternatively, if one adopts a best-case perspective at stage $k$ and would like to maximize the best-case gain at this stage, then one shall compare the two best-case $Q$-functions, namely $\overline{\Qfunc}_{\vec{\lambda} ,k}^{\pi^\star} (\vec\bsh_k, +1)$ and $\overline{\Qfunc}_{\vec{\lambda} ,k}^{\pi^\star} (\vec\bsh_k, -1)$. The optimistic action is then taken to be the difference between the two best-case $Q$-functions, which corresponds to taking $\lambda_k(\vec\bsh_k, \pm 1) = 0$ (see Figure \ref{fig: illustrate stage k best case}). Recall that in the single-stage setting, the min-max risk criterion \ref{eq:worst_case_risk_single_stage} corresponds to maximizing the weighted value \eqref{eqn: middle ground of value interval} with weights set to be $1/2$. This criterion can be seamlessly generalized to the current multiple-stage setting by setting $\lambda_k(\vec\bsh_k, \pm 1)  = 1/2$ (see Figure \ref{fig: illustrate stage k min max}). Other choices of weighting functions can be made to incorporate domain knowledge and user preference, as suggested by \cite{cui2021machine}.


\subsection{Estimating IV-Optimal DTRs}
\label{subsec: estimate IV-optimal DTR}


We discuss how to estimate an IV-optimal DTR given i.i.d.~samples from the law of the random trajectory $(\bsX_k^\obs, A_k^\obs, R_k^\obs)_{k = 1}^K$ and a time-varying instrument variable $\{Z_k\}_{k=1}^K$. The IV-constrained sets $\{\calP_{\vec\bsh_k, a_k}\}$ are specified via partial identification intervals for the $Q$-functions. Specifically, let 
\begin{align}
  \calP_{\vec\bsh_k, a_k}
  & = 
  \bigg\{
    p_{a_k}(\cdot, \cdot|\vec\bsh_k)
    :
    \forall 
    f \in \calF_k,~
    \exists~
    \underline{Q}(\vec\bsh_k, a_k; f),
    \overline{Q}(\vec\bsh_k, a_k; f) \nonumber \\
  & \qquad
    \textnormal{ s.t. }
    \underline{Q}(\vec\bsh_k, a_k; f) 
    \leq 
    \int f(r_k , \bsx_{k+1}) p_{a_k}(d r_k, d\bsx_{k+1}|\vec\bsh_k)
    \leq 
    \overline{Q}(\vec\bsh_k, a_k; f) \nonumber \\
\label{eq:IV_contrained_set_general}
  & \qquad
    \textnormal{ and }
    R_k\in[\underline{C}_k, \overline{C}_k] 
    \textnormal{ almost surely for }
    (R_k, \bsX_{k+1}) \sim p_{a_k}(\cdot, \cdot|\vec\bsh_k)
  \bigg\}.
\end{align}

The boundedness of $R_k$ ensures that partial identification intervals (e.g, Manski-Pepper bounds) have finite width. This is a plausible assumption in many applications \citep{swanson2018partial}.


At the final stage $K$, we take $\calF_K = \{\textnormal{Id}\}$, a singleton containing the identity function that sends $r_K$ to itself (note that $\bsX_{x+1}$ is a null quantity and thus disregarded). 
Then the lower and upper bounds, namely $\underline{Q}(\vec\bsh_k, a_K; \textnormal{Id}) = \underline{Q}_K(\vec\bsh_k, a_K)$ and $\overline{Q}(\vec\bsh_k, a_K; \textnormal{Id}) = \underline{Q}_K(\vec\bsh_k, a_K)$, are precisely the endpoints of the partial identification intervals (constructed using the IV $Z_K$) of the expected final-stage reward $R_K$, where the expectation is taken over the potential outcome distribution $R_K\sim p^\star_{a_K}(\cdot|\vec\bsh_k)$. 
By construction, we have $p^\star_{a_K}(\cdot|\vec\bsh_k) \in \calP_{\vec\bsh_k, a_K}$.

For $k\leq K-1$, we will take $\calF_k$ to be a class of properly defined functions (the specific forms to be specified later) of the current reward $r_k$ and the next-stage covariate $\bsx_{k+1}$. 
The lower and upper bounds $\underline{Q}(\vec\bsh_k, a_k; f), \overline{Q}(\vec\bsh_k, a_k; f)$ constitute the partial identification intervals (constructed using the IV $Z_k$) of the expected value of $f(R_k, \bsX_{k+1})$, where the expectation is taken over the potential outcome distribution $(R_k, \bsX_{k+1})\sim p^\star_{a_k}(\cdot, \cdot|\bsh_k)$. By construction, we have $p^\star_{a_k}(\cdot, \cdot|\vec\bsh_k) \in \calP_{\vec\bsh_k, a_k}$.


Let $\{(\bsx_{k, i}, a_{k, i}, r_{k, i})_{k=1}^K \}_{i=1}^n$ denote the observed data.
Let $\bsh_{1, i}= \bsx_{1, i}$ be the baseline covariates for the $i$-th sample. 
For $2 \leq k \leq K$, let $\vec\bsh_{k, i}= (\vec a_{k-1, i}, \vec r_{k-1, i}, \vec\bsx_{k, i})$ be the historical information up to stage $k$ for the $i$-th sample. In addition, let $\{z_{k, i}\}_{i=1}^n$ denote the stage-$k$ IV data. 


\paragraph{Estimating the contrasts by $Q$-learning.}
Corollary \ref{cor:struc_of_iv-opt_dtr} shows that to estimate the IV-optimal DTR $\pi^\star$, it suffices to estimate the contrast functions $\{\contr^{\pi^\star}_{\vec\lambda, k}(\vec\bsh_k)\}$ defined in \eqref{eq:struc_of_iv-opt_dtr}, which in turn calls for estimating the weighted $Q$-functions $\{Q^{\pi^\star}_{\vec\lambda, k}(\vec\bsh_k, \pm 1)\}$.
We use a $Q$-learning approach to estimate the weighted $Q$-functions \citep{watkins1992q,schulte2014q}. 

For ease of exposition, we specify \eqref{eq:IV_contrained_set_general} via Manski-Pepper bounds \citep{manski1998monotone}. Generalization to other types of partial identification intervals is immediate. For Manski-Pepper bounds to hold, 
we make the 
\emph{mean exchangeability} assumption (\citealp{manski1990nonparametric, hernan2006instruments, swanson2018partial}) or a relaxed \emph{monotone instrumental variable} (MIV) assumption (\citealp{manski1998monotone}). The mean exchangeability assumption is automatically satisfied in a sequential randomized controlled trial with noncompliance.

At Stage $K$, define
\begin{align*}
    & \psi_{K}(\vec{\bsh}_K, a_k, z_k; C)  = C \cdot \bbP(A_K^\obs = -a_k \mid Z_K = z_k, \vec{\bsH}_K^\obs = \vec\bsh_k) \nonumber \\
    \label{eq:mp_bound_psi}
    & ~~ + \mathbb{E}[R_K^\obs \mid \vec{\bsH}_K^\obs = \vec\bsh_k, Z_K = z_k, A_K^\obs = a_k]
    \cdot \bbP(A_K^\obs = a_k \mid \vec{\bsH}_K^\obs = \vec\bsh_k, Z_K = z).\numberthis
\end{align*}
Manski-Pepper bounds state that if $R_K \in [\underline{C}_K, \overline{C}_K]$ almost surely (with respect to $p^\star_{a_K}(\cdot|\vec\bsh_k)$), then its conditional mean potential outcome $\bbE_{R_K\sim p^\star_{a_K}(\cdot|\vec\bsh_k)}[R_K]$ is lower and upper bounded by
\begin{align*}
    & \bbP(Z_K = -1 |\vec{\bsH}_K^\obs= \vec\bsh_k) \cdot \psi_{K}(\vec{\bsh}_K, a_K, -1; \underline{C}_K) \\
    \label{eq:mp_lower_bound}
    & \qquad + \bbP(Z_K = +1 |\vec{\bsH}_K^\obs = \vec\bsh_k)\cdot 
    \big[\psi_{K}(\vec{\bsh}_K, a_K, -1; \underline{C}_K) \lor 
    \psi_{K}(\vec{\bsh}_K, a_K, +1; \underline{C}_K) \big]\numberthis\\
    \text{and}& \\
    & \bbP(Z_K = -1 |\vec{\bsH}_K^\obs = \vec\bsh_k) 
    \cdot 
    \big [\psi_{K}(\vec{\bsh}_K, a_K, -1; \overline{C}_K )
    \land 
    \psi_{K}(\vec{\bsh}_K, a_K, +1; \overline{C}_K )\big]\\
    \label{eq:mp_upper_bound}
    & \qquad + \bbP(Z_K = +1|\vec{\bsH}_K^\obs = \vec\bsh_k)\cdot \psi_{K}(R_K, a_K, +1 ; \vec{\bsh}_K; \overline{C}_K ),\numberthis
\end{align*}
where $\land$ and $\lor$ are shorthands for $\min$ and $\max$, and both bounds are tight. Therefore, as long as we take $\calF_K \supseteq \{\text{Id}\}$ when defining $\calP_{\vec\bsh_k, a_K}$ in \eqref{eq:IV_contrained_set_general}, 
the worst-case and best-case $Q$ functions at stage $K$ defined in \eqref{eq:worst Q stage K}--\eqref{eq:best Q stage K} can be set to \eqref{eq:mp_lower_bound} and \eqref{eq:mp_upper_bound}, respectively. Along with $\lambda_K(\vec\bsh_k, a_K)$ specifications, the construction of the stage-$K$ weighted $Q$-function $Q_{\vec\lambda, K}(\vec\bsh_k, a_K)$ is concluded.

Since both \eqref{eq:mp_lower_bound} and \eqref{eq:mp_upper_bound} are functionals of the observed data distribution, $Q_{\vec\lambda, K}(\vec\bsh_k, a_K)$ can be estimated from the data by fitting parametric models (e.g., linear models) or flexible machine learning models (e.g., regression trees and random forests), and then invoke the plug-in principle.

Given an estimate $\hat Q_{\vec\lambda, K}^\star$ of $Q_{\vec\lambda, K}$, 
we can estimate the contrast function $\contr_{\vec\lambda, K}(\vec\bsh_k)$ by $\hat \contr^\star_{\vec\lambda, K}(\vec\bsh_k) = \hat Q_{\vec\lambda, K}^\star(\vec\bsh_k, +1) - \hat Q_{\vec\lambda, K}^\star(\vec\bsh_k, -1)$. 
In view of Corollary \ref{cor:struc_of_iv-opt_dtr}, the IV-optimal DTR at stage $K$, $\pi^\star_K(\vec\bsh_k)$, can be estimated by $\hpiq_K(\vec\bsh_k) = \sgn(\hat\contr_{\vec\lambda, K}(\vec\bsh_k))$. 
Moreover, the weighted value function of $\pi^\star$ at stage $K$, $V^{\pi^\star}_{\vec\lambda, K}(\vec\bsh_k)$, can be estimated by $\hat V^{\star}_{\vec\lambda, K} = \hat Q^\star_{\vec \lambda, K}(\vec\bsh_k, \hpiq_K(\vec\bsh_k))$.

Now, assume for any stage $t\geq k+1$, the specification of the IV-constrained sets \eqref{eq:IV_contrained_set_general} has been made, 
and the weighted value function at stage $t$, $V^{\pi^\star}_{\vec\lambda, t}(\vec\bsh_t)$, has been estimated by $\hat V^\star_{\vec\lambda, t}(\vec\bsh_t)$. 
In addition, assume that $R_t \in [\underline{C}_t, \overline{C}_t]$ almost surely (with respect to $p^\star_{a_k}(\cdot, \cdot|\vec\bsh_k)$).
At stage $k$, define the pseudo-outcome $PO_k(r_k, \bsx_{k+1}|\vec\bsh_k, a_k) = r_k + V^{\pi^\star}_{\vec\lambda, k+1}(\vec\bsh_k, a_k, r_k, \bsx_{k+1})$.
By construction, we have $PO_k(R_k, \bsX_{k+1}|\vec\bsh_k, a_k) \in [\sum_{t\geq k}\underline{C}_t, \sum_{t\geq k} \overline{C}_t]$ almost surely for $(R_k, \bsX_{k+1}) \sim p^\star_{a_k}(\cdot, \cdot|\vec\bsh_k)$. 
Thus, we can apply Manski-Pepper bounds again to bound the expected value of $PO_k(R_k, \bsX_{k+1}|\vec\bsh_k, a_k)$, and obtain an estimate $\hat Q_{\vec\lambda, k}^\star$ of the weighted $Q$-function of $\pi^\star$ at stage $k$; see Supplementary Material \ref{subappend:alg_for_iv_opt_dtr} for detailed expressions. One nuance is that the pseudo-outcome $PO_k(r_k, \bsx_{k+1}|\vec\bsh_k, a_k)$ depends on the unknown quantity $V^{\pi^\star}_{\vec\lambda, k+1}$. At stage $k$, we have already obtained an estimate $\hat V^{\star}_{\vec\lambda, k+1}$ of $V^{\pi^\star}_{\vec\lambda, k+1}$. 
Thus, the pseudo-outcome can be estimated by $\hat{PO}_k(r_k, \bsx_{k+1}|\vec\bsh_k, a_k) = r_k + \hat V^{\star}_{\vec\lambda, k+1}(\vec\bsh_k, a_k ,r_k, \bsx_{k+1})$.

Finally, the contrast function $\contr_{\vec\lambda, k}^{\pi^\star}(\vec\bsh_k)$ is estimated by $\hat \contr^\star_{\vec\lambda, k}(\vec\bsh_k) = \hat Q_{\vec\lambda, k}^\star(\vec\bsh_k, +1) - \hat Q_{\vec\lambda, k}^\star(\vec\bsh_k, -1)$,
the IV-optimal DTR at stage $k$ is estimated by $\hpiq_k(\vec\bsh_k) = \sgn(\hat\contr_{\vec\lambda, k}(\vec\bsh_k))$,
and the weighted value function of $\pi^\star$ at stage $k$ is estimated by $\hat V^{\star}_{\vec\lambda, k} = \hat Q_{\vec \lambda, k}(\vec\bsh_k, \hpiq_k(\vec\bsh_k))$. In this way, we recursively estimate all contrast functions and obtain an estimated IV-optimal DTR $\hpiq$.

\begin{algorithm}[t!] 
\SetAlgoLined
\caption{Estimation of the IV-Optimal DTR} 
\vspace*{0.12 cm}
\KwIn{Trajectories and instrument variables $\{(\bsx_{k, i}, a_{k, i}, r_{k, i}, z_{k, i}): k \in[K], i\in[n]\}$, weighting functions $\{\lambda(\vec\bsh_k, a_k)\}$, policy class $\Pi$, forms of partial identification intervals.}
\KwOut{Estimated IV-optimal DTR $\hat\pi^\star$.}
\texttt{{\# Step \RN{1}: Q-learning}} \\
Obtain an estimate $\hat Q^\star_{\vec\lambda, K}$ of $Q_{\vec\lambda, K}$ using $(\vec\bsh_{K, i}, r_{K, i}, z_{K, i})_{i=1}^n$\;
Estimate $\contr_{\vec\lambda, K}$ by $\hat \contr^\star_{\vec\lambda, K}(\vec\bsh_K) = \hat Q_{\vec\lambda, K}^\star(\vec\bsh_K, +1) - \hat Q_{\vec\lambda, K}^\star(\vec\bsh_K, -1)$\;
Set $\hpiq_K(\vec\bsh_K) = \sgn(\hat \contr^\star_{\vec\lambda, K}(\vec\bsh_K))$ and $\hat V^\star_{\vec\lambda, k}(\vec\bsh_K) = \hat Q_{\vec\lambda, K}^\star(\vec\bsh_K, \hpiq_K(\vec\bsh_K))$\;
\For{$k = K-1, \hdots, 1$}{
  For each $i\in[n]$, construct $\hat{PO}_{k, i} = r_{k, i} + \hat V^{\star}_{\vec\lambda, k+1}(\vec\bsh_{k+1, i})$\;
  Obtain an estimate $\hat Q^\star_{\vec\lambda, k}$ of $Q_{\vec\lambda, k}^{\pi^\star}$ using using $(\vec\bsh_{k, i}, \hat{PO}_{k, i}, z_{k, i})_{i=1}^n$\;
  Estimate $\contr_{\vec\lambda, k}^{\pi^\star}$ by $\hat \contr^\star_{\vec\lambda, k}(\vec\bsh_k) = \hat Q_{\vec\lambda, k}^\star(\vec\bsh_k, +1) - \hat Q_{\vec\lambda, K}^\star(\vec\bsh_k, -1)$\;
  Set $\hpiq_k(\vec\bsh_k) = \sgn(\hat \contr^\star_{\vec\lambda, k}(\vec\bsh_k))$ and $\hat V^\star_{\vec\lambda, k}(\vec\bsh_k) = \hat Q_{\vec\lambda, k}^\star(\vec\bsh_k, \hpiq_k(\vec\bsh_k))$\;
  
}
\texttt{{\# Step \RN{2}: Weighted classification}} \\
\For{$k = K, \hdots, 1$}{
  Solve the weighted classification problem in \eqref{eq:wted_classification_dtr} to obtain $\hat\pi^\star_k$\;
}
\Return{$\hat \pi^\star$}
\label{alg:IV-optimal_DTR}
\end{algorithm}

\paragraph{Obtaining parsimonious policies via weighted classification.}
In many applications, it is desirable to impose additional constrains on the estimated DTR. For example, one may require the DTR to be parsimonious and thus more interpretable. Such constraints are usually encoded by a restricted function class $\Pi = \Pi_1 \times \cdots \times \Pi_K$, such that any $\pi\notin \Pi$ will not be considered. 

We next introduce a strategy that ``projects'' $\hpiq$, the DTR obtained by $Q$-learning and is thus not necessarily inside $\Pi$, onto the function class $\Pi$. Recall the formula for the IV-optimal DTR $\pi^\star$ given in \eqref{eq:struc_of_iv-opt_dtr}. With some algebra, one readily checks that $\pi^\star$ is a solution to the following weighted classification problem:
\begin{equation}
  \label{eq:wted_classification_dtr}
  \pi^\star_k(\vec\bsh_k) \in \argmin_{a_k\in\{\pm 1\}} |\contr^{\pi^\star}_{\vec\lambda, k}(\vec\bsh_{k})| \cdot \indc{\sgn(\contr^{\pi^\star}_{\vec\lambda, k}(\vec\bsh_{k})) \neq a_k}.
\end{equation}

The above representation 
illuminates a rich class of strategies to search for a DTR $\pi$ within a desired, possibly parsimonious function class $\Pi$, via sequentially solving the following weighted classification problem:
\begin{equation*}
  \hat \pi^\star_k(\vec\bsh_k) \in \argmin_{\pi_k \in \Pi_k} \frac{1}{n}\sum_{i=1}^n |\hat\contr^\star_{\vec\lambda, k}(\vec\bsh_{k, i})| \cdot \indc{\sgn(\hat\contr^\star_{\vec\lambda, k}(\vec\bsh_{k, i})) \neq \pi_k(\vec\bsh_{k, i})},
\end{equation*}
where we recall that $\hat\contr^\star_{\vec\lambda, k}$ is an estimate of $\hat\contr^{\pi^\star}_{\vec\lambda, k}$ obtained via $Q$-learning, and $\{\vec\bsh_{k, i}\}_{i=1}^n$ is the historical information at stage $k$ in our dataset. This strategy bears similarities with the strategy proposed in \citet{zhao2015new}. The estimation procedure is summarized in Algorithm \ref{alg:IV-optimal_DTR}.

\section{Improving Dynamic Treatment Regimes with an Instrumental Variable}\label{sec:dtr_improve}

In this section, we show how the IV-optimality framework developed in Section \ref{sec: extension to DTR} can be modified to tackle the policy improvement problem under the multiple-stage setup.
Let $\pibase$ be a baseline DTR to be improved. 
Some most important baseline DTRs include the standard-of-care DTR (i.e., $\pibase_k = -1$ for any $k$) and the SRA-optimal DTR defined in Section \ref{subsec: DTR setup}. 
The goal of policy improvement, as discussed in Section \ref{sec: ITR IV improve}, is to obtain a DTR $\piimp$ so that $\pibase$ is no worse and potentially better than $\pibase$.

To achieve this goal, we leverage additional information encoded in the collection of IV-constrained sets $\{\calP_{\vec\bsh_k, a_k}\}$. 
We are to introduce the notion of \emph{IV-improved DTR}, generalizing the notion of IV-improved ITR in Definitions \ref{def:risk_based_improvement} and \ref{def:value_based_improvement}.
To start with, we define the following {relative $Q$-functions} and the corresponding {relative value functions} of a DTR $\pi$ with respect to the baseline DTR $\pibase$.

\begin{definition}[Relative $Q$-function and value function]
The {relative $Q$-function} of a DTR $\pi$ with respect to the baseline DTR $\pibase$ at stage $K$ is
\begin{align*}
  &\Qfunc^{\pi\rel\pibase}_K (\vec\bsh_K, a_K)
  = \inf_{ \substack{
  p_{a_K}\in\calP_{\vec\bsh_K, a_K}
  \\ p_{a_K'}\in\calP_{\vec\bsh_K, a_K'} }
  }
  \bigg\{\bbE_{\substack{R_K\sim p_{a_K} \\ R_K'\sim p_{a_K'}}} \left[R_K - R_K'\right] \bigg\}
  ,
\end{align*}
where $a_K' = \pibase_K(\vec\bsh_k)$ is the action taken according to the baseline DTR. The corresponding relative value function is defined as $\Vfunc^{\pi\rel\pibase}_K(\vec\bsh_K) = \Qfunc^{\pi\rel\pibase}_K \big(\vec\bsh_K, \pi_K(\vec\bsh_K)\big)$. 
Recursively, at stage $k = K - 1, K - 2, \cdots, 1$, the relative $Q$-function of $\pi$ with respect to $\pibase$ is defined as
\begin{align*}
  Q^{\pi\rel\pibase}_k(\vec\bsh_k, a_k)
  & = 
  \inf_{ \substack{
  p_{a_k}\in\calP_{\vec\bsh_k, a_k}
  \\ p_{a_k'}\in\calP_{\vec\bsh_k, a_k'} }
  }
  \bigg\{
  \bbE_{\substack{(R_k, \bsX_{k+1})\sim p_{a_k} \\ (R_k', \bsX'_{k+1})\sim p_{a_k'} }} 
  \big[R_k - R_k' + \Vfunc_{k+1}^{\pi\rel\pibase}(\vec\bsH_{k+1})\big]
  \bigg\},
\end{align*}
where $a_k' = \pibase_k(\vec\bsh_k)$ and $\vec\bsH_{k+1} = (\vec\bsh_k, a_k, R_k, \bsX_{k+1})$.
The corresponding relative value function is $\Vfunc^{\pi\rel\pibase}_k(\vec\bsh_k) = \Qfunc^{\pi\rel\pibase}_{k}\big(\vec\bsh_k, \pi_k(\vec\bsh_k)\big)$.
\end{definition}

The relative value function $V^{\pi\rel\pibase}_k$ at stage $k$ captures the cumulative worst-case (subject to the IV constraints) excess value of following the DTR $\pi$ over $\pibase$ from stage $k$ and onwards.
Note that the relative value function $V^{\pi\rel\pibase}_K$ at stage $K$ is analogous to the objective function in \eqref{eq:PIO_value}.
Maximizing the relative value functions would then deliver a DTR that follows the baseline regime $\pibase$ unless there is compelling evidence not to, similar to the IV-improved ITR whose explicit form is given in Propositions \ref{prop:itr_policy_imp_bayes_policy_classification} and \ref{prop:itr_policy_imp_bayes_policy_value}.

We are now ready to define the estimand of interest, generalizing Definitions \ref{def:risk_based_improvement} and \ref{def:value_based_improvement}. 

\begin{definition}[IV-improved DTR]\label{def:DTR imp operator}
Let $\pibase$ be a baseline DTR, $\Pi$ a policy class, and $\{\calP_{\vec\bsh_k, a_k}\}$ a collection of IV-constrained sets. 
A DTR $\piimp$ is said to be IV-improved if it satisfies
\begin{equation}
\label{eq:dtr_imp_optimization}
 \piimp \in \argmax_{\pi\in \Pi} V^{\pi\rel\pibase}_1(\bsx_1)
\end{equation}
for every fixed $\bsh_1 = \bsx_1 \in \calX$.
\end{definition}

Analogous to Theorem \ref{thm:iv_opt_dtr_dynamic_prog}, the following result solves the optimization problem \eqref{eq:dtr_imp_optimization} using a dynamic programming approach when $\Pi$ is the collection of all DTRs.
\begin{theorem}[Dynamic Programming for the IV-improved DTR]
  \label{thm:iv_imp_dtr_dynamic_prog}
Let $\piimp$ be recursively defined as follows:
\begin{align*}
  \piimp_K(\vec\bsh_K) & = \argmax_{a_K\in\{\pm 1\}} \Qfunc_{K}^{\piimp\rel\pibase}(\vec\bsh_K, a_K), 
  ~~\piimp_k(\vec\bsh_k) = \argmax_{a_k\in\{\pm 1\}} \Qfunc^{\piimp\rel\pibase}_{k}(\vec\bsh_k, a_k), ~k = K-1, \hdots, 1.
\end{align*}
Then the DTR $\piimp$ satisfies
\begin{equation}
  \label{eq:iv_imp_dtr_k}
  V^{\piimp\rel\pibase}_{k}(\vec\bsh_k) = \max_{\pi} V^{\pi\rel\pibase}_{k}(\vec\bsh_k)
\end{equation}
for any stage $k$ and any configuration of the historical information $\vec\bsh_k$, where the maximization is taken over all DTRs.
\end{theorem}

The following corollary is similar to Corollary \ref{cor:struc_of_iv-opt_dtr}, and gives an alternative characterization of the IV-improved DTR.

\begin{corollary}[Alternative characterization of the IV-improved DTR]
\label{cor:struc_of_iv-imp_dtr}
Recursively define $\piimp$ as
\begin{equation}
  \label{eq:struc_of_iv-imp_dtr}
  \piimp_k(\vec\bsh_k) = \pibase_k(\vec\bsh_k)
   \cdot 
   \sgn\big\{\contr^{\piimp\rel\pibase}_k(\vec\bsh_k) \big\}
   , ~~ k = K, \hdots, 1,
\end{equation}  
where for $a_{k}' = \pibase_k(\vec\bsh_k), \vec\bsH_{k+1} = (\vec\bsh_k, a_k, R_k, \bsX_{k+1})$ and $\vec\bsH_{k+1}' = (\vec\bsh_k, a_k', R_k', \bsX_{k+1}')$, the quantities
\begin{align}
  \contr^{\piimp\rel\pibase}_K(\vec\bsh_K) & =  
  \sup_{p_{a_K'}\in \calP_{\vec\bsh_K, a_K'}} 
  \bbE_{R_K'\sim p_{a_K'}}
  [R_K'] 
  - 
  \inf_{p_{-a_K'}\in \calP_{\vec\bsh_K, -a_K'}} 
  \bbE_{R_K\sim p_{-a_K'} }
  [ R_K  ], ~\text{and} 
  \nonumber \\
  \contr^{\piimp\rel\pibase}_k(\vec\bsh_k) & = 
  \sup_{p_{a_k'}\in \calP_{\vec\bsh_k, a_k'}} 
  \bbE_{(R_k', \bsX_{k+1}')\sim p_{a_k'}}
  [R_k'] 
  +
  \inf_{p_{a_k'}\in \calP_{\vec\bsh_k, a_k'}}
  \bbE_{(R_k', \bsX_{k+1}') \sim p_{a_k'}}[V^{\piimp\rel\pibase}_k(\vec\bsH_{k+1}')]
  \nonumber \\
  \label{eq:iv-imp_cate}
  & \qquad - 
  \inf_{p_{-a_k'}\in \calP_{\vec\bsh_k, -a_k'}} 
  \bbE_{(R_k, \bsX_{k+1})\sim p_{-a_k'} }
  [ R_k + V^{\piimp\rel\pibase}_{k+1}(\vec\bsH_{k+1}) ] 
  , ~ 1 \leq k \leq K-1,
\end{align}
depend on $\piimp$ only through $\piimp_{(k+1):K}$ and are therefore well-defined.
Then $\piimp$ satisfies \eqref{eq:iv_imp_dtr_k} for any stage $k$ and any configuration of the historical information $\vec\bsh_k$.
In particular, 
$\piimp$ is an IV-improved DTR in the sense of Definition \ref{def:DTR imp operator} when $\Pi$ is the set of all DTRs.
\end{corollary}


Expressions in \eqref{eq:struc_of_iv-imp_dtr} and \eqref{eq:iv-imp_cate} admit a rather intuitive explanation. The contrast function $\contr^{\piimp\rel\pibase}_k$ measures the cumulative worst-case gain of following $\piimp$ over $\pibase$, starting from stage $k$, and one would only flip the decision made by $\pibase_k$ if this worse-case gain is positive. In fact, with some algebra, one readily checks that $\piimp_K$ can be expressed in a similar form as \eqref{eq: bayes rule} and \eqref{eq: bayes rule II}. 

Given Corollary \ref{cor:struc_of_iv-imp_dtr} and the strategy in Section \ref{subsec: estimate IV-optimal DTR}, Algorithm \ref{alg:IV-optimal_DTR} can be modified \emph{mutatis mutandis} to yield an algorithm for estimating an IV-improved DTR. For brevity, we refer the readers to Supplementary Material \ref{subappend:alg_for_iv_imp_dtr} for details.



\section{Theoretical Properties}
\label{sec: theory}
In this section, we prove non-asymptotic bounds on the deviance between the estimated IV-optimal as well as the IV-improved DTRs to their population counterparts. To do this, we need several standard assumptions, which we detail below.

First of all, we need a proper control on the complexity of the policy class $\Pi = \Pi_1\times \cdots \Pi_K$.
Note that any $\pi_k \in \Pi_k$ is a Boolean functions that sends a specific configuration of the historical information $\vec\bsh_k$ to a binary decision $\pi_k(\vec\bsh_k)$. Let $\calH_k$ be the collection of all possible $\vec\bsh_k$s.
A canonical measure of complexity of Boolean functions is the Vapnik–Chervonenkis (VC) dimension \citep{vapnik1968uniform}. 
The VC-dimension of $\Pi_k$, denoted as $\vc(\Pi_k)$, is the largest positive integer $d$ such that there exists a set of $d$ points $\{\vec\bsh_{k}^{(1)}, \cdots, \vec\bsh_{k}^{(d)}\}\subseteq \calH_k$ \emph{shattered} by $\Pi_k$, in the sense that for any binary vector $\bsv \in \{\pm 1\}^d$, there exists $\pi_{k}^{(\bsv)} \in \Pi_k$ such that $\pi_{k}^{(\bsv)}(\vec\bsh_k^{(j)}) = v_j$ for $j\in[d]$. 
If no such $d$ exists, then $\vc(\Pi_k) = \infty$.

In practice, a DTR is most useful when it is parsimonious. For this reason, the class of linear decision rules and the class of decision trees with a fixed depth are popular in various application fields \citep{tao2018tree,speth2020assessment}.
It is well-known that when the domain is a subset of $\bbR^d$, the VC-dimension of the class of linear decision rules is at most $d+1$ (see, e.g., Example 4.21 of \citealt{wainwright2019high}) and the VC dimension of the class of decision tree with $L$ leaves is $\calO(L \log (Ld))$ \citep{leboeuf2020decision}.

Note that a DTR $\pi_{1:(k-1)}$, along with the distributions $\{p_{a_t}(\cdot, \cdot|\vec\bsh_t): a_t\in\{\pm 1\}, \vec\bsh_t \in \calH_t\}_{t=1}^{k-1}$, induces a probability distribution on $\calH_k$, the set of stage-$k$ historical information. Let $\scrH_k$ denote the set of all such probability distributions when we vary $\pi_{1:(k-1)} \in \Pi_{1:(k-1)}$ and $p_{a_t}(\cdot, \cdot|\vec\bsh_t) \in \calP_{\vec\bsh_t, a_t}$ for $t\in[k-1]$. An element in $\scrH_k$ will be denoted as $q_k$, and the law of $\vec\bsH_k^\obs$ will be denoted as $q_k^\obs$. The next assumption concerns how much information on $q_k^\obs$ can be generalized to the information on a specific $q_k \in \scrH_k$.

\begin{assumption}[Bounded concentration coefficients]
\label{assump:bdd_concentration_coef}
Suppose there exist positive constants $\{\sfc_k\}_{k=1}^K$ such that 
$$
	\sup_{q_k\in \scrH_{k}} \sup_{\vec\bsh_k \in \calH_k} \frac{d q_k}{d q^\obs_k}(\vec\bsh_k) \leq \sfc_k.
$$
\end{assumption}

Assumption \ref{assump:bdd_concentration_coef} is often made in the 
reinforcement learning literature (see, e.g., \citealt{munos2003error,szepesvari2005finite,chen2019information}) and is closely related to the ``overlap'' assumption in causal inference.
A sufficient condition for the above assumption to hold is that the probability of seeing any historical information (according to the observed data distribution or any distribution in $\scrH_k$) is strictly bounded away from $0$ and $1$.

Recall that our algorithms for estimating IV-optimal and IV-improved DTRs is a two-step procedure, where in Step \RN{1}, the contrast functions are estimated via $Q$-learning, and in Step \RN{2}, a parsimonious DTR is obtained via weighted classification.
As mentioned in Section \ref{subsec: estimate IV-optimal DTR}, in the first stage, there are a lot of flexibilities in choosing the specific models and algorithms for estimating the contrast functions. As the result, a fine-grained understanding of this stage would require a case-by-case analysis. In order to not over-complicate the discussion, we make the following assumption, which asserts the existence of a ``contrast estimation oracle''.

\begin{assumption}[Existence of the contrast estimation oracle]
\label{assump:contrast_estimation_oracle}
	Suppose in $Q$-learning, we use an algorithm that takes $n$ i.i.d.~samples and outputs $\{\hat \contr^{\star}_{\vec\lambda, k}\}_{k=1}^K$ and $\{\hat \contr^{\uparrow}_{k}\}_{k=1}^K$ that satisfy
	\begin{align}
		\label{eq:contr_est_err_policy_estimation}
		& \bbP\bigg(\bbE |\hat \contr^\star_{\vec\lambda, k}(\vec\bsH_k^\obs) - \contr^{\pi^\star}_{\vec\lambda, k}(\vec\bsH_{k}^\obs)| \geq \sfC_{k, \delta} \cdot  n^{-\alpha_{k}} \bigg) \leq \delta \\
		\label{eq:contr_est_err_policy_improvement}
		& \bbP\bigg(\bbE |\hat \contr^\uparrow_{k}(\vec\bsH_k^\obs) - \contr^{\piimp\rel\pibase}_{k}(\vec\bsH_{k}^\obs)| \geq \sfC_{k, \delta} \cdot  n^{-\alpha_{k}} \bigg) \leq \delta 
	\end{align}	
	for any $k \in[K], \delta > 0$.
	In the above display, the probability is taken over the randomness in the $n$ samples and the expectation is taken over the randomness in a fresh trajectory $\vec\bsH_k^\obs \sim q_k^\obs$.
\end{assumption}

We will use \eqref{eq:contr_est_err_policy_estimation} when we analyze the estimated IV-optimal DTR and we will invoke \eqref{eq:contr_est_err_policy_improvement} when we analyze the estimated IV-improved DTR.

Assumption \ref{assump:contrast_estimation_oracle} 
is usually a relatively mild assumption. Indeed, the ground truth contrast functions $C^{\pi^\star}_{\vec\lambda, k}$ are superimpositions of several conditional expectations of the observed data distribution (see \eqref{eq:mp_bound_psi}--\eqref{eq:mp_upper_bound}), and it is reasonable to assume that they can be estimated at a vanishing rate as the sample size tends to infinity.
While this assumption simplifies the analysis by allowing us to bypass the case-by-case analyses of $Q$-learning, the analysis of the weighted classification remains highly non-trivial.

To proceed further, we adopt a form of sample splitting procedure called cross-fitting (\citealp{chernozhukov2016double, athey2020policy}). Specifically, we split the $n$ samples into $m$ equally-sized batches: $[n] = \cup_{j\in[m]} B_j$, where $m\geq 2$ and $m\asymp 1$.
Each batch has size $|B_j| = n_j \asymp n/m \asymp n$. 
For each $i\in[n]$, let $j_i\in[m]$ be the index of its batch, so that $i\in B_{j_i}$. 
For the $i$-the sample, we apply the contrast estimation oracle in Assumption \ref{assump:contrast_estimation_oracle} on the out-of-batch samples $B_{-j_i} = [n]\setminus B_{j_i}$ to obtain either an estimate $\hat\contr^\star_{\vec\lambda, k}(\vec\bsh_k; B_{-j_i})$ of $\contr^{\pi^\star}_{\vec\lambda, k}(\vec\bsh_k)$ for policy estimation, or an estimate $\hat \contr^{\uparrow}_k(\vec\bsh_k; B_{-{j_i}})$ of $\contr^{\piimp \rel\pibase}_{\vec\lambda, k}$ for policy improvement. 
For the policy estimation problem, we solve the following optimization problem:
\begin{equation}
\label{eq:cross_fit}
    \min_{\pi_k \in \Pi_k} \frac{1}{n} \sum_{i=1}^n |\hat \contr^{\star}_{\vec\lambda, k}(\vec\bsh_{k, i}; B_{-j_i})| \cdot \Indc\bigg\{\sgn(\hat \contr^{\star}_{\vec\lambda, k}(\vec\bsh_{k, i}; B_{-j_i}))\neq \pi_k(\vec\bsh_{k, i}) \bigg\}.
\end{equation}	
The corresponding optimization problem for the policy improvement problem is given by 
\begin{equation}
	\label{eq:cross_fit_improvement}
		\min_{\pi_k \in \Pi_k} \frac{1}{n} \sum_{i=1}^n |\hat \contr^{\uparrow}_{k}(\vec\bsh_{k, i}; B_{-j_i})| \cdot \Indc\bigg\{\pibase(\vec\bsh_{k, i})\cdot \sgn(\hat \contr^{\uparrow}_{k}(\vec\bsh_{k, i};B_{-j_i}))\neq \pi_k(\vec\bsh_{k, i}) \bigg\}.
\end{equation}

{We emphasize that cross-fitting is mostly for theoretical convenience. The performance of our algorithm with and without cross-fitting is similar; see simulations in Section \ref{sec: simulation}.}

In practice, given a limited computational budget, we may only solve the optimization problem \eqref{eq:cross_fit} up to a certain precision. 
Our analysis will be conducted for the approximate minimizers $\hat \pi^\star_k, \hpiimp_k\in \Pi_k$ that satisfy
\begin{align}
	\label{eq:cross_fit_approx_min}
	& \frac{1}{n} \sum_{i=1}^n |\hat \contr^{\star}_{\vec\lambda, k}(\vec\bsh_{k, i}; B_{-j_i})| \cdot \Indc\bigg\{\sgn(\hat \contr^{\star}_{\vec\lambda, k}(\vec\bsh_{k, i}; B_{-j_i}))\neq \hat \pi^\star_k(\vec\bsh_{k, i}) \bigg\} \leq \textnormal{\eqref{eq:cross_fit}} + \ep^{\opt}_k,\\
	\label{eq:cross_fit_approx_min_improvement}
	&\frac{1}{n} \sum_{i=1}^n |\hat \contr^{\uparrow}_{k}(\vec\bsh_{k, i}; B_{-j_i})| \cdot \Indc\bigg\{\pibase(\vec\bsh_{k, i})\cdot \sgn(\hat \contr^{\uparrow}_{k}(\vec\bsh_{k, i};B_{-j_i}))\neq \hpiimp_k(\vec\bsh_{k, i}) \bigg\}
	\leq \textnormal{\eqref{eq:cross_fit_improvement}} + \ep^{\opt}_k,
\end{align}	
respectively, where $\ep^\opt_k$ is the optimization error when solving \eqref{eq:cross_fit} and \eqref{eq:cross_fit_improvement}. 

To quantify the loss of information from restriction to the policy class $\Pi$, we define 
\begin{align*}
    & \tilde \pi^\star_k \in \argmin_{\pi_k \in \Pi_k} \bbE \bigg[|\contr^{\pi^\star}_{\vec\lambda, k}(\vec\bsH_k^\obs)| \cdot \Indc\bigg\{\sgn(\contr^{\pi^\star}_{\vec\lambda, k}(\vec\bsH_{k}^\obs)\neq \pi_k(\vec\bsH_{k}^\obs) \bigg\}\bigg],\\
    & \tpiimp_k \in \argmin_{\pi_k\in\Pi_k} \bbE\bigg[ |\Cimp_k(\vec\bsH^\obs_k)| 
	\cdot \Indc\bigg\{\pibase_k(\vec\bsH_{k}^\obs)\cdot \sgn(\Cimp_{k}(\vec\bsH^\obs_{k}))\neq \pi_k(\vec\bsH_{k}^\obs) \bigg\}
	\bigg].
\end{align*}	
Note that according to Corollaries \ref{cor:struc_of_iv-opt_dtr} and \ref{cor:struc_of_iv-imp_dtr}, when $\Pi_k$ is the set of all Boolean functions, we have $\tilde \pi_k^\star = \pi^\star_k$ and $\tpiimp_k = \piimp_k$.
Otherwise, the difference between $\tilde\pi^\star_k$ and $\pi^\star_k$ (resp. $\tpiimp_k$ and $\piimp_k$) measures the approximation error when we restrict ourselves to $\Pi_k$ instead of the set of all Boolean functions in the policy estimation (resp. policy improvement) problem.

We are now ready to present the main result of this section.
\begin{theorem}[Performance of the estimated IV-optimal and IV-improved DTRs]
\label{thm: DTR performance}
Let Assumptions \ref{assump:bdd_concentration_coef} and \ref{assump:contrast_estimation_oracle} hold. Fix any $\delta\in(0, 1)$. Let the optimization error be $\calE_\opt = \sum_{k=1}^K \sfc_k \cdot \ep_k^\opt$. Let the approximation errors for the policy estimation problem and the policy improvement problem be 
\begin{align*}
\calE_{\prox} & = \sum_{k=1}^K \sfc_k \cdot \bbE[V^{\pi^\star_{k:K}}_{\vec\lambda, k}(\vec\bsH^\obs_{k}) - V^{\tilde\pi^\star_k \pi^{\star}_{(k+1):K}}_{\vec\lambda, k}(\vec\bsH^\obs_k)],\\
\calE_{\prox}' & = \sum_{k=1}^K \sfc_k \cdot \bbE[V^{\piimp_{k:K}\rel\pibase}_{k}(\vec\bsH_{k}) - V^{\tpiimp_k \piimp_{(k+1):K} \rel\pibase}_{k}(\vec\bsH_k^\obs)]
\end{align*}
respectively, where $\pi_k\pi'_{(k+1):K}$ denotes the DTR that acts according to $\pi_k$ at stage $k$ and follows $\pi'$ from stage $k+1$ to $K$. Finally, define the generalization error 
\begin{equation*}
	\calE_\gen = \sum_{k=1}^K \sfc_k \cdot \bigg(\sfC_{k, \frac{\delta}{4mK}} \cdot n^{-\alpha_k} + (K-k+1)\cdot \sqrt{\frac{\vc(\Pi_k) + \log(K/\delta)}{n}}\bigg) 	
\end{equation*}
Then, there exists an absolute constant $C>0$ such that the following two statements hold:
\begin{enumerate}
	\item For the policy estimation problem, with probability at least $1-\delta$, we have
	\begin{equation}
		\label{eq:policy_estimation_bound}
		\bbE [V^{\pi^\star}_{\vec\lambda, 1}(\bsX_1^\obs) - V^{\hat \pi^\star}_{\vec\lambda, 1}(\bsX_1^\obs)] \leq \calE_{\opt} + \calE_{\prox} + C \cdot \calE_{\gen};
	\end{equation}
	\item For the policy improvement problem, with probability $1 -\delta$, we have
	\begin{equation}
		\label{eq:policy_improvement_bound}
		\bbE [V^{\pi^\star\rel\pibase}_1(\bsX_1^\obs) - V^{\hpiimp\rel\pibase}_1(\bsX_1^\obs)] \leq \calE_{\opt} + \calE_{\prox}' + C \cdot \calE_{\gen}.
	\end{equation}
\end{enumerate}
The expectations in \eqref{eq:policy_estimation_bound} and \eqref{eq:policy_improvement_bound} are taken over a fresh sample of the observed first stage covariates $\bsX_1^\obs$.
\end{theorem}

Theorem \ref{thm: DTR performance} shows that the weighted value function of $\hat\pi^\star$ (resp. relative value function of $\hpiimp$) at the first stage converges to that of the IV-optimal DTR (resp. IV-improved) up to three sources of errors: the optimization error that stems from only approximately solving the weighted classification problem, the approximation error that results from the restriction to a parsimonious policy class $\Pi$, and a vanishing generalization error. We defer the proof to Supplementary Material \ref{append:proof_estimators}.

\section{Simulation studies}
\label{sec: simulation}
\subsection{Goal, data-generating process and simulation structure}
We verify that the IV-optimal DTRs indeed have superior performance compared to the baseline DTRs and investigate the performance of the IV-optimal DTRs for assorted choices of $\lambda_k(\Vec{\bsh}_k, a_k)$ via simulations. 
We consider a data-generating process with two time-independent covariates $X_1$, $X_2 \sim \text{Unif}[-1, 1]$. At stage one, there exists an unmeasured confounder $U_1 \sim \textnormal{Bernoulli}(1/2)$. 
The instrumental variable $Z_1$ is independent of $U_1$ and follows $\textnormal{Rademacher}(1/2)$, 
the observed $A_1$ is Rademacher with head probability $\text{expit}\{C_1\cdot (Z_1+1) - \xi U_1 - 2\}$, and the reward $R_1$ is Bernoulli with head probability $\text{expit}\{0.5(\sgn\{X_1 - 1\} - \lambda U_1 + 0.2)\cdot (A_1 + 1)$, where $C_1, \xi, \lambda$ are constants to be specified later, and $\textnormal{expit}\{x\} = 1/(1+e^{-x})$.
At stage two, there exists a second unmeasured confounder $U_2 \sim \textnormal{Bernoulli}(1/2)$. The instrumental variable is again independent of $Z_2$ and follows $\textnormal{Rademacher}(1/2)$, the action $A_2$ is Rademacher with head probability $\text{expit}\{C_1\cdot (Z_1+1) + X_1 - 7(R_1 - 0.5) - \xi(1 + X_1)(2U_2 - 1)\}$, and the reward $R_2$ is Bernoulli with head probability $\text{expit}\{0.1(A_1+1) + 0.4[1 - X_1 + R_1 - \lambda (2U_2 - 1)]\cdot (A_2 + 1)\}$.

According to this data-generating process, $Z_1$ and $Z_2$ are valid instrumental variables. There are two unmeasured confounders $U_1$ and $U_2$, one at each stage, and both unmeasured confounders are effect modifiers. We will interpret the treatment option $A_1, A_2 = -1$ as the standard-of-care, e.g., low-level NICU, and $+1$ as the prospective treatment, e.g., high-level NICU. One may check that for large $\xi$, the prospective treatment has a negative effect on $R_1$; however, the unmeasured confounding $U_1$ may create a spurious, positive treatment effect in an analysis that omits $U_1$. On the other hand, the prospective treatment at the second stage may have a positive or negative treatment effect on $R_2$ depending on the baseline covariate $X_1$ and the first stage outcome $R_1$.

Our simulation can be compactly summarized as a $3\times 3 \times 3 \times 2 \times 2$ factorial design with the following five factors:
\begin{description}
\item[Factor 1:] instrumental variable strength $C_1 = 3$, $4$, and $5$;
\item[Factor 2:]level of unmeasured confounding $\lambda = 1$, $2$, and $3$.
\item[Factor 3:] baseline policies $\pibase$. We consider three baseline policies: the standard-of-care regime $\pibase_{\std}$ that assigns $\pibase_{\std, 1}=\pibase_{\std, 2}= -1$ to everyone, a prospective treatment regime $\pibase_{\prosp}$ that assigns $\pibase_{\std, 1}=\pibase_{\std, 2} = +1$ to everyone, and $\pibase_{\sra}$ that is ignorant of the unmeasured confounding and is optimal under the sequential randomization assumption. 
\item[Factor 4:] training samples $n_{\train} = 500$ or $1000$.
\item[Factor 5:] procedures used to estimate the relevant conditional expectations in the partial identification intervals. We consider using either simple parametric models (linear, logistic, and multinomial regression) or random forests (\citealp{breiman2001random}). 
\end{description}

The observed data consist of $\mathcal{D}_{\text{obs}} = \{(X_1, X_2, A_1, Z_1, R_1, A_2, Z_2, R_2),~i = 1, \cdots, n_{\train}\}$. An SRA-optimal DTR $\pibase_\sra$ does not leverage the IV data ($Z_1$ and $Z_2$) while IV-improved DTR does and uses this information to improve upon baseline rules. We estimated the conditional average treatment effect involved in estimating $\pibase_{\sra}$ using a robust augmented inverse probability weighting estimator (AIPW), and then applied a weighted classification routine. This is known as C-learning in the literature (\citealp{zhang2018c}), and can also be considered a variant of the backward outcome weighted learning (BOWL) procedure (\citealp{zhao2015new}). All classification problems involved in estimating the $\pibase_{\sra}$ and the IV-improved DTRs were implemented using a classification tree with a maximum depth of $2$, which is meant to replicate the real data application where parsimonious rules are more useful and can deliver most insight (\citealp{laber2015tree, speth2020assessment}). Three IV-improved DTRs are denoted by $\piimp_{\std}$, $\piimp_{\prosp}$, and $\piimp_{\sra}$, respectively. Finally, for each data-generating process, we further estimated three IV-optimal DTRs with $\lambda_k(\Vec{\bsh}_k, \pm 1) = 0$, $1/2$, and $1$ for $k = 1, 2$, and these three IV-optimal DTRs are referred to as $\pi_{\texttt{IV}, 0}$, $\pi_{\texttt{IV}, 1/2}$, and $\pi_{\texttt{IV}, 1}$, respectively. Therefore, we have a total of nine regimes (three baseline regimes, three IV-improved regimes, and three IV-optimal regimes) under consideration. We evaluated each estimated regime $\widehat{\pi}$ by calculating its value function using $1,000,000$ fresh samples $(X_1, X_2) \sim \text{Unif}[-1, 1]$ and integrating out the binary unmeasured confounders $U_1, U_2$ using Monte Carlo.

\subsection{Simulation results}

\begin{table}[t]
\caption{Simulation results: all relevant conditional probabilities were estimated using random forests implemented in the \texttt{R} package \texttt{randomForest} with node size equal to $5$ and $n_{\train} = 1000$.}
\label{tbl: simulation rf n=1000}
\centering
\resizebox{\textwidth}{!}{\begin{tabular}{cccccccccc}
  \hline
$n_{\train} = 1000$ & $\pibase_{\std}$ & $\piimp_{\std}$  & $\pibase_{\prosp}$ & $\piimp_{\prosp}$  & $\pibase_{\sra}$ & $\piimp_{\sra}$ & $\pi_{\texttt{IV},1}$ & $\pi_{\texttt{IV},0}$ & $\pi_{\texttt{IV},1/2}$ \\ \\
   
   \multicolumn{10}{c}{$\xi = 1$} \\
   
 \multirow{2}{*}{$C_1 = 3$} & 1.00 & 1.14 & 1.03 & 1.23 & 1.13 & 1.22 & 1.13 & 1.18 & 1.18 \\  
   & [1.00,1.00] & [1.13,1.15] & [1.03,1.03] & [1.20,1.29] & [1.10,1.17] & [1.19,1.26] & [1.09,1.15] & [1.14,1.23] & [1.14,1.23] \\ 
  \multirow{2}{*}{$C_1 = 4$}  & 1.00 & 1.15 & 1.03 & 1.24 & 1.13 & 1.22 & 1.13 & 1.18 & 1.18 \\ 
   & [1.00,1.00] & [1.14,1.15] & [1.03,1.03] & [1.20,1.29] & [1.09,1.16] & [1.18,1.26] & [1.09,1.15] & [1.14,1.23] & [1.14,1.23] \\ 
  \multirow{2}{*}{$C_1 = 5$}  & 1.00 & 1.15 & 1.03 & 1.23 & 1.13 & 1.22 & 1.12 & 1.18 & 1.17 \\ 
   & [1.00,1.00] & [1.14,1.15] & [1.03,1.03] & [1.20,1.29] & [1.10,1.17] & [1.18,1.26] & [1.09,1.15] & [1.14,1.23] & [1.14,1.23] \\   \\
   
   \multicolumn{10}{c}{$\xi = 2$} \\
   
  \multirow{2}{*}{$C_1 = 3$} & 1.00 & 1.10 & 0.94 & 1.14 & 1.10 & 1.14 & 1.10 & 1.14 & 1.14 \\ 
   & [1.00,1.00] & [1.10,1.11] & [0.94,0.94] & [1.13,1.14] & [1.08,1.12] & [1.12,1.16] & [1.08,1.11] & [1.12,1.17] & [1.12,1.17] \\ 
  \multirow{2}{*}{$C_1 = 4$}  & 1.00 & 1.11 & 0.94 & 1.14 & 1.10 & 1.14 & 1.10 & 1.14 & 1.14 \\ 
   & [1.00,1.00] & [1.10,1.12] & [0.94,0.94] & [1.13,1.17] & [1.08,1.12] & [1.12,1.16] & [1.08,1.11] & [1.12,1.17] & [1.12,1.17] \\ 
  \multirow{2}{*}{$C_1 = 5$}  & 1.00 & 1.11 & 0.94 & 1.14 & 1.09 & 1.14 & 1.10 & 1.14 & 1.14 \\ 
   & [1.00,1.00] & [1.10,1.12] & [0.94,0.94] & [1.13,1.17] & [1.07,1.12] & [1.12,1.16] & [1.08,1.11] & [1.12,1.17] & [1.12,1.17] \\     \\
  
   \multicolumn{10}{c}{$\xi = 3$} \\
   
  \multirow{2}{*}{$C_1 = 3$} & 1.00 & 1.06 & 0.88 & 1.10 & 1.06 & 1.10 & 1.07 & 1.10 & 1.11 \\ 
   & [1.00,1.00] & [1.05,1.07] & [0.88,0.88] & [1.10,1.11] & [1.05,1.09] & [1.09,1.10] & [1.07,1.07] & [1.10,1.11] & [1.10,1.11] \\ 
  \multirow{2}{*}{$C_1 = 4$} & 1.00 & 1.07 & 0.88 & 1.10 & 1.07 & 1.10 & 1.07 & 1.10 & 1.10 \\ 
   & [1.00,1.00] & [1.06,1.08] & [0.88,0.88] & [1.10,1.11] & [1.05,1.09] & [1.10,1.11] & [1.07,1.07] & [1.10,1.11] & [1.10,1.11] \\ 
  \multirow{2}{*}{$C_1 = 5$} & 1.00 & 1.07 & 0.88 & 1.10 & 1.06 & 1.10 & 1.06 & 1.10 & 1.10 \\ 
   & [1.00,1.00] & [1.06,1.08] & [0.88,0.88] & [1.10,1.11] & [1.05,1.09] & [1.10,1.11] & [1.06,1.07] & [1.10,1.11] & [1.10,1.11] \\ 
   \hline
\end{tabular}}
\end{table}

\begin{figure}[t]
    \centering
    \includegraphics[width=\textwidth]{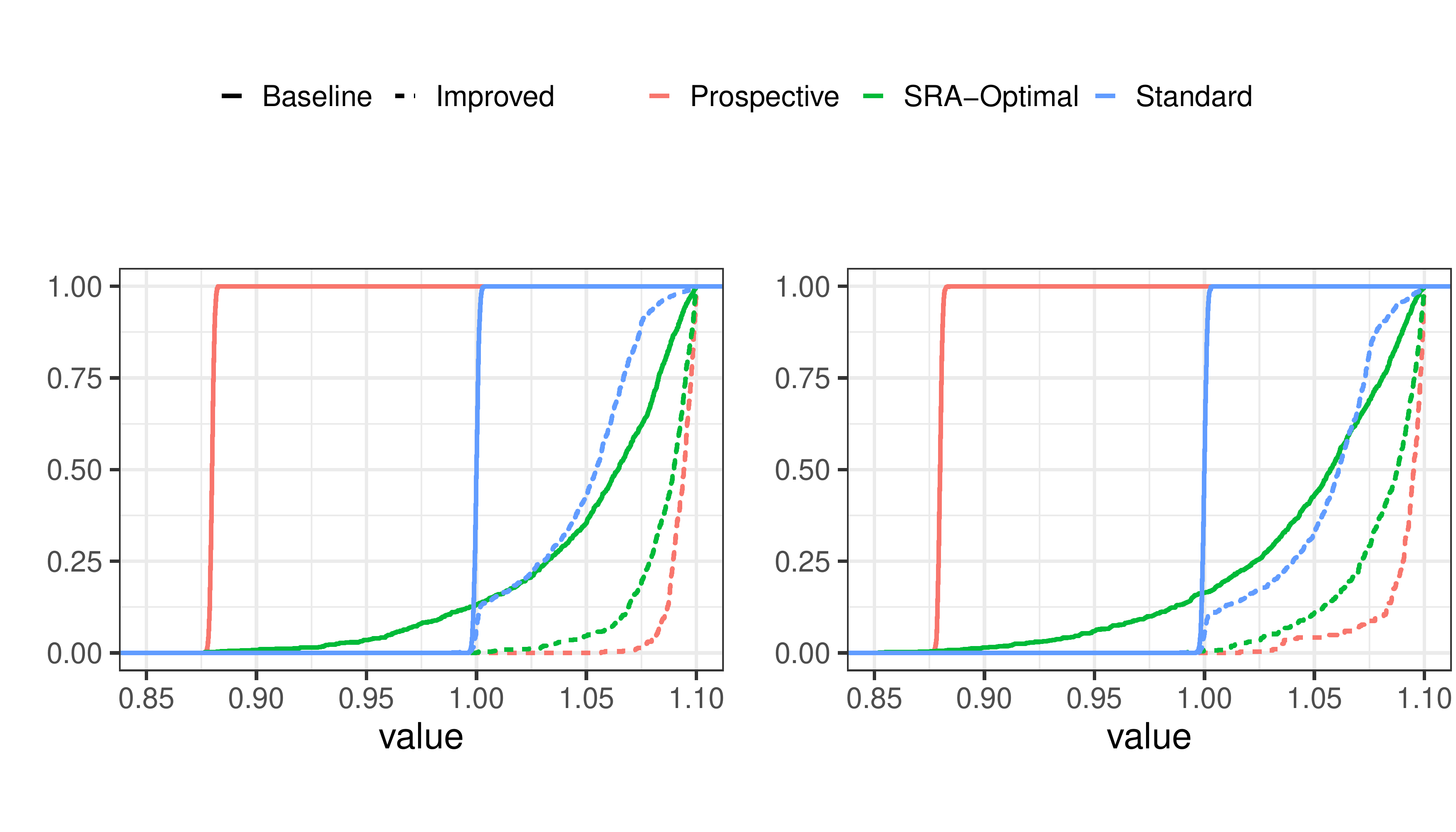}
    \caption{\small Cumulative distribution functions of value functions of three baseline policies $\pibase_{\std}$, $\pibase_{\prosp}$ and $\pibase_{\texttt{SRA}}$, and their respective IV-improved policies $\piimp_{\std}$, $\piimp_{\prosp}$ and $\piimp_{\texttt{SRA}}$. Left panel: $\xi = 3$ and $C_1 = 3$. Right panel: $\xi = 3$ and $C_1 = 5$.}
    \label{fig: CDFs two IV strengths}
\end{figure}

Table \ref{tbl: simulation rf n=1000} summarizes the estimated mean and interquartile range of the value functions for baseline DTRs, their corresponding IV-improved DTRs, and three different IV-optimal DTRs, when $n_{\train} = 1000$, all relevant conditional expectations estimated via random forests (\citealp{breiman2001random}) as implemented in the \texttt{R} package \texttt{randomForest}, and for various $(\xi, C_1)$ combinations. 


There are three trends consistent with our theory and intuition upon examining the simulation results. First and foremost, we observed that the IV-improved DTRs indeed had superior performance compared to their corresponding baseline DTRs including the SRA-optimal DTR. Although the extent of improvement depends on the specifics of data-generating processes, the improvement was uniform across all data-generating processes. Figure \ref{fig: CDFs two IV strengths} plots the cumulative distribution functions (CDFs) of three baseline DTRs and their IV-improved DTRs across $1,000$ simulations in two data-generating processes. It is evident that in either data-generating process and for any of the three baseline DTRs, the value functions of IV-improved DTRs always stochastically dominate those of corresponding baseline DTRs. We observed the same stochastic dominance phenomenon in each of the $108$ data-generating processes considered in the simulation studies. Second, when comparing three IV-optimal DTRs corresponding to different choices of weighting functions $\lambda_k(\Vec{\bsh}_k, a_k)$, we observed that the choices of $1/2$ and $1$, corresponding to the min-max and the worst-case perspectives, had better performance compared to the best-case DTR and the SRA-optimal DTR. We further plot CDFs of the value functions of each IV-optimal DTR and the SRA-optimal DTR (Figures \ref{fig: C=4 xi = 1 IV-optimal rules} and \ref{fig: C=4 xi = 3 IV-optimal rules} in Supplementary Material \ref{sec: appendix simu results}) and observed that the min-max and worst-case DTRs stochastically dominated the SRA-optimal DTRs in all sampling situations considered in the simulation studies. Lastly, we found that estimating relevant conditional expectations using simple parametric models (Tables \ref{tbl: simulation n=500} and \ref{tbl: simulation n=1000} in Supplementary Material \ref{sec: appendix simu results}) and the cross-fitting version of the algorithm (Tables \ref{tbl: simulation cross fit n=500} and \ref{tbl: simulation cross fit n=1000} in Supplementary Material \ref{sec: appendix simu results}) yielded slightly inferior, but qualitatively similar results.


\section{Application}
\label{sec: application}

\begin{figure}[t!]
\centering
\begin{subfigure}[b]{\textwidth}
  \centering
  \begin{tikzpicture}[
every text node part/.style={align=center},
squarednode/.style={rectangle, draw=black, fill=white, very thick, minimum size=5mm}
]

\node at (0, 0) {First birth:};
\node[squarednode, draw = red] at (3,0)  (1_1_root) {All High-Level NICU};
\node at (7.8, 0) {Second birth:};
\node[squarednode, draw = red] at (11,0)  (1_2_root) {All High-Level NICU};
\end{tikzpicture}
  \caption[]
  {{\small No penalty attached to attending a high-level NICU. All mothers are assigned to high-level NICUs.}}
\end{subfigure}
\\ \vspace{0.5 cm}

\begin{subfigure}[b]{\textwidth}
  \centering
  \begin{tikzpicture}[
every text node part/.style={align=center},
squarednode/.style={rectangle, draw=black, fill=white, very thick, minimum size=5mm}
]
\node at (0, 0) {First birth:};
\node[squarednode, draw = black, fill = white] at (3,0)  (2_1_root) {White = 0};
\node[align=left] at (2.2,-0.5) {\small $\boldsymbol y$};
\node[align=left] at (3.8,-0.5) {\small $\boldsymbol n$};
\node[squarednode, draw = red] at (1.5, -1.5) (2_1_root_yes) {High-Level\\ NICU};
\node[squarednode, draw = black, fill = white] at (4.5, -1.5) (2_1_root_no) {Age $\leq 25$};
\node[squarednode, draw = black, fill = white] at (3, -3) (2_1_age_yes) {GA $\leq 37$};
\node[squarednode, draw = red] at (6, -3) (2_1_age_no) {High-Level \\ NICU};
\node[squarednode, draw = red] at (1.5, -4.5) (2_1_ga_yes) {High-Level \\ NICU};
\node[squarednode, draw = red] at (4.5, -4.5) (2_1_ga_no) {Low-Level \\ NICU};

\node at (8, 0) {Second birth:};
\node[squarednode, draw = black, fill = white] at (11,0)  (2_2_root) {White = 0};
\node[align=left] at (10.2,-0.5) {\small $\boldsymbol y$};
\node[align=left] at (11.8,-0.5) {\small $\boldsymbol n$};
\node[squarednode, draw = red] at (9.5, -1.5) (2_2_root_yes) {High-Level\\ NICU};
\node[squarednode, draw = black, fill = white] at (12.5, -1.5) (2_2_root_no) {Age $\leq 30$};
\node[squarednode, draw = black, fill = white] at (11, -3) (2_2_age_yes) {GA $\leq 35$};
\node[squarednode, draw = red] at (14, -3) (2_2_age_no) {High-Level \\ NICU};
\node[squarednode, draw = red] at (9.5, -4.5) (2_2_ga_yes) {High-Level \\ NICU};
\node[squarednode, draw = red] at (12.5, -4.5) (2_2_ga_no) {Low-Level \\ NICU};

\draw[ultra thick, ->] (2_1_root) -- (2_1_root_yes);
\draw[ultra thick, ->] (2_1_root) -- (2_1_root_no);
\draw[ultra thick, ->] (2_1_root_no) -- (2_1_age_yes);
\draw[ultra thick, ->] (2_1_root_no) -- (2_1_age_no);
\draw[ultra thick, ->] (2_1_age_yes) -- (2_1_ga_yes);
\draw[ultra thick, ->] (2_1_age_yes) -- (2_1_ga_no);

\draw[ultra thick, ->] (2_2_root) -- (2_2_root_yes);
\draw[ultra thick, ->] (2_2_root) -- (2_2_root_no);
\draw[ultra thick, ->] (2_2_root_no) -- (2_2_age_yes);
\draw[ultra thick, ->] (2_2_root_no) -- (2_2_age_no);
\draw[ultra thick, ->] (2_2_age_yes) -- (2_2_ga_yes);
\draw[ultra thick, ->] (2_2_age_yes) -- (2_2_ga_no);
\end{tikzpicture}%
\caption[]
  {{\small Moderate penalty attached to attending a high-level NICU. $68.3\%$ of all mothers are assigned to high-level NICUs for their first deliveries and $59.9\%$ of all mothers are assigned to high-level NICUs for their second deliveries.}}
  \label{fig: real data moderate penalty}
\end{subfigure}\\
\vspace{0.5 cm}

\begin{subfigure}[b]{\textwidth}
  \centering
  \begin{tikzpicture}[
every text node part/.style={align=center},
squarednode/.style={rectangle, draw=black, fill=white, very thick, minimum size=5mm}
]
\node at (0, 0) {First birth:};
\node[squarednode, draw = black, fill = white] at (3,0)  (3_1_root) {White = 0};
\node[align=left] at (2.2,-0.5) {\small $\boldsymbol y$};
\node[align=left] at (3.8,-0.5) {\small $\boldsymbol n$};
\node[squarednode, draw = black, fill = white] at (1.5, -1.5) (3_1_root_yes) {GA $\leq 37$};
\node[squarednode, draw = red] at (4.5, -1.5) (3_1_root_no) {Low-Level \\ NICU};
\node[squarednode, draw = red] at (0, -3) (3_1_ga_yes) {High-Level \\NICU};
\node[squarednode, draw = red] at (3, -3) (3_1_ga_no) {Low-Level \\ NICU};

\node at (8, 0) {Second birth:};
\node[squarednode, draw = black, fill = white] at (12.5,0)  (3_2_root) {White = 0};
\node[align=left] at (11.7,-0.5) {\small $\boldsymbol y$};
\node[align=left] at (13.3,-0.5) {\small $\boldsymbol n$};
\node[squarednode, draw = black, fill = white] at (11, -1.5) (3_2_root_yes) {GA $\leq 36$};
\node[squarednode, draw = red] at (14, -1.5) (3_2_root_no) {Low-Level \\ NICU};
\node[squarednode, draw = black, fill = white] at (9.5, -3) (3_2_ga_yes) {Age $\leq 32$};
\node[squarednode, draw = red] at (12.5, -3) (3_2_ga_no) {Low-Level \\ NICU};
\node[squarednode, draw = red] at (8, -4.5) (3_2_age_yes) {Low-Level \\ NICU};
\node[squarednode, draw = red] at (11, -4.5) (3_2_age_no) {High-Level \\ NICU};

\draw[ultra thick, ->] (3_1_root) -- (3_1_root_yes);
\draw[ultra thick, ->] (3_1_root) -- (3_1_root_no);
\draw[ultra thick, ->] (3_1_root_yes) -- (3_1_ga_yes);
\draw[ultra thick, ->] (3_1_root_yes) -- (3_1_ga_no);

\draw[ultra thick, ->] (3_2_root) -- (3_2_root_yes);
\draw[ultra thick, ->] (3_2_root) -- (3_2_root_no);
\draw[ultra thick, ->] (3_2_root_yes) -- (3_2_ga_yes);
\draw[ultra thick, ->] (3_2_root_yes) -- (3_2_ga_no);
\draw[ultra thick, ->] (3_2_ga_yes) -- (3_2_age_yes);
\draw[ultra thick, ->] (3_2_ga_yes) -- (3_2_age_no);
\end{tikzpicture}%
\caption[]
  {{\small Large penalty attached to attending a high-level NICU. $4.53\%$ of all mothers are assigned to high-level NICUs for their first deliveries and $8.56\%$ of all mothers are assigned to high-level NICUs for their second deliveries.}}
  \label{fig: real data large penalty}
\end{subfigure}
\caption{\small IV-optimal dynamic treatment regimes estimated using the NICU data. Left arrow corresponds to `yes' and right `no'. GA stands for gestational age.}
\label{fig: real data DTR trees}
\end{figure}

We considered a total of $183,487$ mothers who delivered exactly two births during $1995$ and $2009$ in the Commonwealth of Pennsylvania, and relocated at their second deliveries so that their ``excess-travel-time'' IVs at two deliveries were different. We considered covariates that measured mothers' neighborhood circumstances including poverty rate, median income, etc, mothers' demographic information including race (white or not), age, years of education, etc, and variables related to delivery including gestational age in weeks and length of prenatal care in months, and eight congenital diseases. The ``excess-travel-time" IVs in both stages were then dichotomized: $1$ if above the median and $0$ otherwise. Mothers' treatment choice and their babies' mortality status at the first delivery were included as covariates for studying the second delivery. We used the \emph{multiple imputation by chained equations} method (\citealp{buuren2010mice}) implemented in the \textsf{R} package \textsf{MICE} to impute missing covariate data, and repeat analysis on $5$ imputed datasets.

We assume that high-level NICUs do no harm compared to low-level NICUs; therefore, all partial identification intervals in this application were estimated under the \emph{monotone treatment response} (MTR) assumption (\citealp[Chp.8]{manski2003partial}). We considered estimating an IV-optimal DTR minimizing the maximum risk at each delivery; see Section \ref{subsec:alternative characterization}, and explored the trade-off between minimizing the maximum risk and the cost/capacity constraint by adding a generic penalty to the value function. We performed weighted classification using a classification tree with maximum depth equal to $3$ so that the resulting DTR is interpretable. Figure \ref{fig: real data DTR trees} plots three estimated DTRs corresponding to no penalty attached, a moderate penalty, and a large penalty attached to attending a high-level NICU. When there is no penalty attached, all mothers are assigned to high-level NICUs. As we increase the penalty, fewer mothers (albeit mothers who benefit most from attending a high-level NICU) are assigned to high-level NICUs. For instance, Figure \ref{fig: real data moderate penalty} corresponds to sending $68\%$ mothers to a high-level NICU at their first deliveries and $59.9\%$ at their second deliveries. Mothers who are assigned to high-level NICUs according to this DTR \emph{either} belong to racial and ethnic minority groups \emph{or} are older and have premature gestational age. Similarly, Figure \ref{fig: real data large penalty} plots a regime where less than $10\%$ of mothers are assigned to a high-level NICU. Mothers who are assigned to high-level NICUs according to this DTR belong to racial and ethnicity minority groups and have premature births. Our analysis here seems to suggest that in general race/ethnicity, age, and gestational age are the most significant effect modifiers. Gestational age has long been hypothesized as an effect modifier; see \citet{lorch2012differential, yang2014estimation, michael2020instrumental}; more recently, \citet{yannekis2020differential} found a differential effect between different race/ethnic groups. On the other hand, mother's age appears to be a new discovery that worth looking into. Overall, our method both complemented previous published results and generated new insights.

\section{Discussion}
\label{sec: discussion}
We systematically study the problem of estimating an dynamic treatment regime from retrospective observational data using a time-varying instrumental variable. We formulate the problem under a generic partial identification framework, derive a counterpart of the classical $Q$-learning and Bellman equation under partial identification, and use it as the basis for generalizing a notion of IV-optimality to the dynamic treatment regimes. One important variant of the developed framework is a strategy to improve upon a baseline dynamic treatment regime. As demonstrated via extensive simulations, IV-improved DTRs indeed have favorable performance compared to the baseline DTRs, including baseline DTRs that are optimal under the no unmeasured confounding assumption.

With the increasing availability of administrative databases that keep track of clinical data, it is tempting to estimate some useful, real-world-evidence-based dynamic treatment regimes from such retrospective data. To make any causal/treatment effect statements from non-RCT data, an instrumental variable analysis is often better-received by clinicians. Fortunately, many reasonably good instrumental variables are available, e.g., daily precipitation, geographic distances, service providers' preference, etc. Many of these IVs are intrinsically time-varying and could be leveraged to estimate a dynamic treatment regime using the framework proposed in this article.

In practice, to deliver a most useful policy intervention, it is important to take into account various practical constraints, e.g., those arising from limited facility capacity or increased cost. Our framework can be readily extended to incorporating various constraints.

We conclude this article by mentioning a few open problems. 
First, our analysis depends on the assumption of bounded concentration coefficients (Assumption \ref{assump:bdd_concentration_coef}). 
It will be interesting to see if this assumption can be relaxed by imposing additional structural assumptions. There are some recent process in the reinforcement learning literature (see, e..g, \citealt{jiang2017contextual,sun2019model,du2021bilinear}), but whether those structural assumptions can be adapted to the current setting remains a question. 
Meanwhile, our proofs bypasses case-by-case analyses of $Q$-learning by assuming the existence of the contrast estimating oracle (Assumption \ref{assump:contrast_estimation_oracle}). 
It is an interesting future direction to conduct more fine-grained analyses and characterize the optimal rate of convergence of those contrast functions. 
Finally, the current article considers estimating DTRs from a historical dataset. It might be of interest to extend our framework to the online interactive setup such as the one considered in \cite{liao2021instrumental}.

\small{
\setlength{\bibsep}{0.2pt plus 0.3ex}
\bibliographystyle{apalike}
\bibliography{paper-ref}
}
\newpage
\appendixtitleon
\appendixtitletocon
\addtocontents{toc}{\protect\setcounter{tocdepth}{2}}
\begin{appendices}
\tableofcontents

\clearpage
\section{Proofs Regarding the Estimands}

\subsection{Proof of Propositions \ref{prop:itr_policy_imp_bayes_policy_classification} and \ref{prop:itr_policy_imp_bayes_policy_value}}
\begin{proof}[Proof of Proposition \ref{prop:itr_policy_imp_bayes_policy_classification}]
If $L(\bsX) > 0$, then 
\begin{align*}
  & \sup_{\contr(\bsX)\in [L(\bsX), U(\bsX)]}|\contr(\bsX)|
  \cdot \bigg(\indc{ \pi(\bsX) \neq \sgn(\contr(\bsX)) } - \indc{ \pibase(\bsX) \neq \sgn(\contr(\bsX)) }\bigg)\\
  & =\sup_{\contr(\bsX)\in [L(\bsX), U(\bsX)]} |\contr(\bsX)| \cdot \bigg(
  \indc{ \pi(\bsX) = -1} - \indc{ \pibase(\bsX)  = -1 }
  \bigg) \\
  & = \indc{\pi(\bsX) = -1, \pibase(\bsX) = +1} \cdot |U(\bsX)| - \indc{\pi(\bsX) = +1, \pibase(\bsX) = -1}\cdot |L(\bsX)|  \\
  & = \indc{\pi(\bsX) = -1} \bigg(\indc{\pibase(\bsX) = +1}\cdot|U(\bsX)| + \indc{\pibase(\bsX) = -1}\cdot|L(\bsX)|\bigg) - \indc{\pibase(\bsX) = -1} \cdot |L(\bsX),
\end{align*}
where the third line is by $|U(\bsX)| \geq |L(\bsX)|$ when $L(\bsX) > 0$.
The above display is minimized when $\pi(\bsX) = +1$.
Meanwhile, if $U(\bsX) < 0$, then
\begin{align*}
  & \sup_{\contr(\bsX)\in [L(\bsX), U(\bsX)]}|\contr(\bsX)|
  \cdot \bigg(\indc{ \pi(\bsX) \neq \sgn(\contr(\bsX)) } - \indc{ \pibase(\bsX) \neq \sgn(\contr(\bsX)) }\bigg)\\
  & =\sup_{\contr(\bsX)\in [L(\bsX), U(\bsX)]} |\contr(\bsX)| \cdot \bigg(
  \indc{ \pi(\bsX) = +1} - \indc{ \pibase(\bsX)  = +1 }
  \bigg)\\
  & = \indc{\pi(\bsX) = + 1, \pibase(\bsX) = -1}\cdot |L(\bsX)| - \indc{\pi(\bsX) = -1, \pibase(\bsX) = + 1}\cdot |U(\bsX)| \\
  & = \indc{\pi(\bsX) = -1} \bigg( -\indc{\pibase(\bsX) = -1}\cdot |L(\bsX)| - \indc{\pibase(\bsX) = +1}\cdot |U(\bsX)| \bigg) + \indc{\pibase(\bsX) = -1} \cdot |L(\bsX)|,
\end{align*}
where the third line is by $|L(\bsX)| \geq |U(\bsX)|$ when $U(\bsX) < 0$. The above display is minimized when $\pi(\bsX) = -1$.
Finally, if $L(\bsX)\leq 0 \leq U(\bsX)$, then taking $\contr(\bsX) = 0$ gives
\begin{align*}
  & \sup_{\contr(\bsX)\in [L(\bsX), U(\bsX)]}|\contr(\bsX)|
  \cdot \bigg(\indc{ \pi(\bsX) \neq \sgn(\contr(\bsX)) } - \indc{ \pibase(\bsX) \neq \sgn(\contr(\bsX)) }\bigg) \geq 0.
\end{align*}
Moreover, the above lower bound can be attained by taking $\pi(\bsX) = \pibase(\bsX)$, meaning that the left-hand side above is minimized when $\pi(\bsX) = \pibase(\bsX)$. 
%
Combining the three cases above concludes the proof.
\end{proof}

\begin{proof}[Proof of Proposition \ref{prop:itr_policy_imp_bayes_policy_value}]
We have
\begin{align*}
  & \inf_{ \substack{
  p_{\pi(\bsX)}\in\calP_{\bsX, \pi(\bsX)}
  \\ p_{\pibase(\bsX)}\in\calP_{\bsX, \pibase(\bsX)} }
  }
  \bigg\{\bbE_{\substack{Y\sim p_{\pi(\bsX)} , Y'\sim p_{\pibase(\bsX)}}} \left[Y - Y'\right] \bigg\} \\
  & = \indc{\pi(\bsX) = \pibase(\bsX)} \cdot 0 
  + \indc{\pi(\bsX) \neq \pibase(\bsX)} \cdot \bigg( \inf_{p_{\pi(\bsX)} \in \calP_{\bsX, \pi(\bsX)}} \bbE_{Y\sim p_{\pi(\bsX)}}[Y] - \sup_{p_{\pibase(\bsX)} \in \calP_{\bsX, \pibase(\bsX)}} \bbE_{Y' \sim p_{\pibase(\bsX)}}[Y'] \bigg)\\
  & = \indc{\pi(\bsX)  = - \pibase(\bsX)} \cdot ( \underline{Q}(\bsX, -\pibase(\bsX)) - \overline{Q}(\bsX, \pibase(\bsX)) ) \\
  & = \indc{\pibase(\bsX) = +1, \pi(\bsX)  = - 1} \cdot  ( \underline{Q}(\bsX, -1) - \overline{Q}(\bsX, +1))  \\
  & \qquad + 
  \indc{\pibase(\bsX) = -1, \pi(\bsX)  = + 1} \cdot  ( \underline{Q}(\bsX, +1) - \overline{Q}(\bsX, -1))  \\
  & = - \indc{\pibase(\bsX) = +1, \pi(\bsX)  = - 1} \cdot \Uval(\bsX) + 
  \indc{\pibase(\bsX) = -1, \pi(\bsX)  = + 1} \cdot \Lval(\bsX)
\end{align*}
where the third line is by the assumption that $\underline{p}_{a}(\cdot|\bsx), \overline{p}_a(\cdot|\bsx) \in \calP_{\bsx, a}$. 
If $\pibase(\bsX) = +1$, then the above display becomes $-\indc{\pi(\bsX) = -1} \cdot \Uval(\bsX)$, which is maximized by taking $\pi(\bsX) = -1$ when $\Uval(\bsX) < 0$ and $\pi(\bsX) = +1$ when $\Uval(\bsX) \geq 0$. 
On the other hand, if $\pibase(\bsX) = -1$, then the above display becomes $\indc{\pi(\bsX) = +1} \cdot \Lval(\bsX)$, which is maximized by taking $\pi(\bsX) = +1$ when $\Lval(\bsX) > 0$ and $\pi(\bsX) = -1$ when $\Lval(\bsX) \leq 0$. 

Now, note that if $\Lval(\bsX) > 0$, then $\Uval(\bsX) > 0$ by construction, and thus the optimal $\pi(\bsX) = +1$ regardless of which action $\pibase$ takes. Similarly, if $\Uval(\bsX) < 0$, then $\Lval(\bsX) < 0$, and hence the optimal $\pi(\bsX) = -1$. Finally, if $\Lval(\bsX) \leq 0 \leq \Uval(\bsX)$, one readily checks that the optimal $\pi(\bsX)$ is $+1$ when $\pibase(\bsX) = +1$ and is $-1$ when $\pibase(\bsX) = -1$. The proof is concluded.
\end{proof}

\subsection{Proof of Theorem \ref{thm:iv_opt_dtr_dynamic_prog}}
  We proceed by induction. At stage $K$, \eqref{eq:iv_opt_dtr_k} holds by construction. Suppose \eqref{eq:iv_opt_dtr_k} holds for any stage $t\geq k+1$ and any $\vec\bsh_t$. At stage $k$, we have
  \begin{align*}
    \Vfunc^{\pi^\star}_{\scrP,k}(\vec\bsh_k) 
    & = \Qfunc_{\scrP,k}^{\pi^\star}\big(\vec\bsh_k, \pi^\star_{k}(\vec\bsh_k)\big) \\
    & = \max_{a_k\in\{\pm 1\}} \Qfunc_{\scrP,k}^{\pi^\star}\big(\vec\bsh_k, a_k\big)  \\
    & = \max_{a_k\in\{\pm 1\}}  \bbE_{p_{a_k}(\cdot,\cdot|\vec\bsh_k) \sim \scrP_{\vec\bsh_k, a_k}}\bbE_{(R_k, \bsX_{k+1})\sim p_{a_k}(\cdot, \cdot|\vec\bsh_k)}[R_k + \Vfunc^{\pi^\star}_{\scrP, k+1}(\vec\bsH_{k+1})] \\
    & \overset{(*)}{=} \max_{a_k\in\{\pm 1\}}  \bbE_{p_{a_k}(\cdot,\cdot|\vec\bsh_k) \sim \scrP_{\vec\bsh_k, a_k}}\bbE_{(R_k, \bsX_{k+1})\sim p_{a_k}(\cdot, \cdot|\vec\bsh_k)}[R_k + \max_{\pi}\Vfunc^{\pi}_{\scrP, k+1}(\vec\bsH_{k+1})] \\
    & \overset{(**)}{\geq} \max_{\substack{a_k\in\{\pm 1\} \\ \pi_{(k+1):K}}} 
    \bbE_{p_{a_k}(\cdot,\cdot|\vec\bsh_k) \sim \scrP_{\vec\bsh_k, a_k}}\bbE_{(R_k, \bsX_{k+1})\sim p_{a_k}(\cdot, \cdot|\vec\bsH_k)}[R_k + \Vfunc^{\pi}_{\scrP, k+1}(\vec\bsH_{k+1})] \\
    & = \max_{\substack{a_k\in\{\pm 1\} \\ \pi_{(k+1):K}}} \Qfunc^{\pi}_{\scrP, k} (\vec\bsh_{k}, a_k)\\
    & = \max_{\pi_{k:K}} \Vfunc_{\scrP,k}^\pi(\vec\bsh_k)
  \end{align*}
  where $(*)$ is by our inductive hypothesis and $(**)$ is by Jensen's inequality. Note that the reverse inequality $\Vfunc^{\pi^\star}_{\scrP,k}(\vec\bsH_k)  \leq \max_{\pi} \Vfunc_{\scrP,k}^\pi(\vec\bsH_k)$ holds by construction. Thus, the desired result follows.

\subsection{Proof of Proposition \ref{prop:struc_of_iv-opt_dtr}}
\begin{proof}[Proof of Proposition \ref{prop:struc_of_iv-opt_dtr}]
Part 1 of this proposition is proved by recursively constructing $\scrP_{\vec\bsh_k, a_k}$ be the two-point prior:
$$
  \scrP_{\vec\bsh_k, a_k} = \lambda(\vec\bsh_k, a_k) \cdot \delta_{\underline p_{a_K}(\cdot | \vec\bsh_k)} + \big(1 - \lambda(\vec\bsh_k, a_k)\big) \cdot \delta_{\overline p_{a_K}(\cdot | \vec\bsh_k)},
$$
where $\underline p_{a_K}(\cdot | \vec\bsh_k)$ and $\overline p_{a_K}(\cdot | \vec\bsh_k)$ are the distributions that attain the infimum and the supremum of \eqref{eq:worst_Q_func_k} and \eqref{eq:best_Q_func_k}, respectively. 
Part 2 holds trivially since $\underline{Q}_{\vec\lambda, k} \leq Q^\pi_{\scrP, k} \leq \overline{Q}_{\vec\lambda, k}$ for any $k\in[K]$.
\end{proof}

\subsection{Proof of Theorem \ref{thm:iv_imp_dtr_dynamic_prog}}
  The proof is similar to the proof of Theorem \ref{thm:iv_opt_dtr_dynamic_prog}.
  We proceed by induction. At stage $K$, the desired result holds by construction. Suppose \eqref{eq:iv_imp_dtr_k} holds for any stage $t\geq k+1$ and any $\vec\bsh_t$. 
  At stage $k$, for $a_k' = \pibase_k(\vec\bsh_k)$ and $\vec\bsH_{k+1} = (\vec\bsh_k, a_k, R_k, \bsX_{k+1})$ we have
  \begin{align*}
    \Vfunc^{\piimp\rel\pibase}_{k}(\vec\bsh_k) 
    & = \Qfunc_{k}^{\piimp\rel\pibase}\big(\vec\bsh_k, \piimp_{k}(\vec\bsh_k)\big) \\
    & = \max_{a_k\in\{\pm 1\}} \Qfunc_{k}^{\piimp\rel\pibase}\big(\vec\bsh_k, a_k\big)  \\
    & = \max_{a_k\in\{\pm 1\}}  
    \inf_{ \substack{
    p_{a_k}(\cdot|\vec\bsh_k)\in\calP_{\vec\bsh_k, a_k}
    \\ p_{a_k'}(\cdot|\vec\bsh_k)\in\calP_{\vec\bsh_k, a_k'} }
    }
    \bigg\{
    \bbE_{\substack{(R_k, \bsX_{k+1})\sim p_{a_k}(\cdot, \cdot|\vec\bsh_k)\\ (R_k', \bsX'_{k+1})\sim p_{a_k'}(\cdot, \cdot|\vec\bsh_k)}} 
    [R_k - R_k' + \Vfunc_{k+1}^{\piimp\rel\pibase}(\vec\bsH_{k+1})] 
    \bigg\}
    \\
    & \overset{(*)}{=} 
    \max_{a_k\in\{\pm 1\}}  
    \inf_{ \substack{
    p_{a_k}(\cdot|\vec\bsh_k)\in\calP_{\vec\bsh_k, a_k}
    \\ p_{a_k'}(\cdot|\vec\bsh_k)\in\calP_{\vec\bsh_k, a_k'} }
    }
    \bigg\{
    \bbE_{\substack{(R_k, \bsX_{k+1})\sim p_{a_k}(\cdot, \cdot|\vec\bsh_k)\\ (R_k', \bsX'_{k+1})\sim p_{a_k'}(\cdot, \cdot|\vec\bsh_k)}}
    [R_k - R_k' +  \max_{\pi} \Vfunc_{k+1}^{\pi\rel\pibase}(\vec\bsH_{k+1})]
    \bigg\}
    \\
    & \overset{(**)}{\geq} 
    \max_{\substack{a_k\in\{\pm 1\} \\ \pi_{(k+1):K}}}
    \inf_{ \substack{
    p_{a_k}(\cdot|\vec\bsh_k)\in\calP_{\vec\bsh_k, a_k}
    \\ p_{a_k'}(\cdot|\vec\bsh_k)\in\calP_{\vec\bsh_k, a_k'} }
    }
    \bigg\{
    \bbE_{\substack{(R_k, \bsX_{k+1})\sim p_{a_k}(\cdot, \cdot|\vec\bsh_k)\\ (R_k', \bsX'_{k+1})\sim p_{a_k'}(\cdot, \cdot|\vec\bsh_k)}}
    [R_k  - R_k'+  \Vfunc_{k+1}^{\pi\rel\pibase}(\vec\bsH_{k+1})]
    \bigg\} \\
    & = \max_{\substack{a_k\in\{\pm 1\} \\ \pi_{k+1}, \hdots, \pi_K}} \Qfunc^{\pi\rel\pibase}_{k} (\vec\bsh_{k}, a_k)\\
    & = \max_{\pi_k, \hdots, \pi_K} \Vfunc_{k}^{\pi\rel\pibase}(\vec\bsh_k),
  \end{align*}
  where $(*)$ is by our inductive hypothesis and $(**)$ is by Jensen's inequality followed by max-min inequality. Since the reverse inequality $V^{\piimp\rel\pibase}_{k}(\vec\bsh_k) = \max_{\pi} V^{\pi\rel\pibase}_{k}(\vec\bsh_k)$ is trivial, the proof is concluded.

\subsection{Proof of Corollary \ref{cor:struc_of_iv-imp_dtr}} 

By Theorem \ref{thm:iv_imp_dtr_dynamic_prog}, an IV-improved DTR satisfies (with $V^{\pi\rel\pibase}_{K+1} = 0$)
\begin{align*}
  \piimp_k(\vec\bsh_k)
  & \in \argmax_{a_k\in\{\pm 1\}} Q^{\piimp\rel\pibase}_k(\vec\bsh_k, a_k) \\
  & = \argmax_{a_k \in \{\pm 1\}]} 
  \inf_{ \substack{
  p_{a_k}(\cdot|\vec\bsh_k)\in\calP_{\vec\bsh_k, a_k}
  \\ p_{a_k'}(\cdot|\vec\bsh_k)\in\calP_{\vec\bsh_k, a_k'} }
  }
  \bigg\{
  \bbE_{\substack{(R_k, \bsX_{k+1})\sim p_{a_k}(\cdot, \cdot|\vec\bsh_k) \\ (R_k', \bsX'_{k+1})\sim p_{a_k'}(\cdot, \cdot|\vec\bsh_k) }} 
  [R_k - R_k' + \Vfunc_{k+1}^{\piimp\rel\pibase}(\vec\bsH_{k+1})]
  \bigg\} \\
  & = \argmax_{a_k \in \{\pm 1\}}
  \bigg\{
  \indc{a_k = - \pibase_k(\vec\bsh_k)}
  \times 
  \big(- \contr^{\piimp\rel\pibase}_{k}(\vec\bsh_{k})\big)
  \bigg\}\\
  & = \argmin_{a_k\in\{\pm 1\}}
  \bigg\{
  \indc{a_k = - \pibase_k(\vec\bsh_k)}
  \times 
  \contr^{\piimp\rel\pibase}_{k}(\vec\bsh_{k})
  \bigg\},
\end{align*}
from which \eqref{eq:struc_of_iv-imp_dtr} follows.

\section{Proofs Regarding the Estimators} \label{append:proof_estimators}
\subsection{Proof for Part 1 of Theorem \ref{thm: DTR performance}: Performance of the Estimated IV-Optimal DTR}

We start by presenting several useful preliminary results.

\begin{lemma}[Performance difference lemma for policy estimation]
\label{lemma:perf_diff}
Let Assumption \ref{assump:bdd_concentration_coef} hold. For two DTRs $\pi$ and $\pi'$, we have
\begin{equation}
  \label{eq:perf_diff}
  \bbE[V^{\pi}_{\vec\lambda, 1}(\bsX_1^\obs) - V^{\pi'}_{\vec\lambda, 1}(\bsX_1^\obs)] \leq  \sum_{k  = 1}^K  \sfc_k \cdot \bbE[V_{\vec{\lambda}, k}^{\pi_{k:K}}(\vec\bsH_k^\obs) - V^{\pi'_k, \pi_{(k+1):K}}_{\vec{\lambda}, k}(\vec\bsH_k^\obs)].
\end{equation}  
\end{lemma}
\begin{proof}
  See Appendix \ref{prf:lemma:perf_diff}.
\end{proof}

\begin{lemma}[Duality between risk and value for policy estimation]
  \label{lemma:risk_and_regret}
  For any DTR $\pi$, any $k\in[K]$, and any historical information $\vec\bsh_{k}$, we have 
  $$
    V^{\pi}_{\vec{\lambda}, k}(\vec\bsh_k) =  \max_{a_k\in\{\pm 1\}} Q^{\pi}_{\vec{\lambda}, k}(\vec\bsh_k, a_k) - |\contr^{\pi}_{\vec{\lambda}, k}(\vec\bsh_k)| \cdot \Indc\bigg\{ \sgn(\contr^{\pi}_{\vec{\lambda}, k}(\vec\bsh_k)) \neq \pi_k(\vec\bsh_k) \bigg\}.
  $$
\end{lemma}
\begin{proof}
  See Appendix \ref{prf:lemma:risk_and_regret}.
\end{proof}

\begin{proposition}[Risk bound for the estimated IV-optimal DTR]
\label{prop:iv_opt_risk_bound}
Fix $k\in[K]$ and $\delta\in(0, 1)$. Let Assumption \ref{assump:contrast_estimation_oracle} hold. 
Then, with probability at least $1-\delta$, we have
$$
  \calR_k(\hat\pi^\star_k) - \inf_{\pi_k \in \Pi_k} \calR_k(\pi_k) \leq \ep^\opt_k + C \cdot \bigg(\sfC_{k, \frac{\delta}{4m}} \cdot n^{-\alpha_k} + (K-k+1) \cdot \sqrt{\frac{\vc(\Pi_k) + \log(1/\delta)}{n}}\bigg),
$$
where $C>0$ is an absolute constant and
\begin{equation}
  \label{eq:pop_risk_iv_opt_dtr}
\calR_k(\pi_k) = \bbE\bigg[ |\contr^{\pi^\star}_{\vec{\lambda}, k}(\vec\bsH_k^\obs)| \cdot \Indc\bigg\{\sgn(\contr^{\pi^\star}_{\vec{\lambda}, k}(\vec\bsH_k^\obs)) \neq \pi_k(\vec\bsH_k^\obs)\bigg\} \bigg].
\end{equation}
\end{proposition}
\begin{proof}
  See Appendix \ref{prf:prop:iv_opt_risk_bound}.
\end{proof}

With the above results, we are ready to present the proof for Part 1 of Theorem \ref{thm: DTR performance}.
  Applying Lemma \ref{lemma:perf_diff}, we have
  \begin{align*}
    \bbE[V^{\pi^\star}_1(\bsX_1^\obs) - V^{\hat \pi^\star}_1(\bsX_1^\obs)] 
    & \leq \sum_{k  = 1}^K  \sfc_k \cdot \bbE[V_{\vec{\lambda}, k}^{\pi^\star_{k:K}}(\vec\bsH_k^\obs) - V^{\hat\pi^\star_k, \pi^\star_{(k+1):K}}_{\vec{\lambda}, k}(\vec\bsH_k^\obs)].\\
    & =  \calE_\prox
    + \sum_{k=1}^K \sfc_k \cdot \bbE[V_{\vec{\lambda}, k}^{\tilde \pi^\star_{k} \pi^\star_{(k+1):K}}(\vec\bsH_k^\obs) - V^{\hat\pi^\star_k, \pi^\star_{(k+1):K}}_{\vec{\lambda}, k}(\vec\bsH_k^\obs)],
  \end{align*}
  where we recall that $\tilde \pi_k$ is defined as the minimizer of $R(\pi_k)$ over $\Pi_k$.
  By Lemma \ref{lemma:risk_and_regret}, we have
  \begin{align*}
    & \bbE[V_{\vec{\lambda}, k}^{\tilde\pi^\star_{k}\pi^\star_{(k+1):K}}(\vec\bsH_k^\obs) - V^{\hat\pi^\star_k, \pi^\star_{(k+1):K}}_{\vec{\lambda}, k}(\vec\bsH_k^\obs)] \\
    & = \bbE \bigg[
    \max_{a_k\in\{\pm 1\}} Q^{\pi^\star_{(k+1):K}}_{\vec{\lambda}, k}(\vec\bsH_k^\obs, a_k) 
    - 
    |\contr^{\pi^\star_{(k+1):K}}_{\vec{\lambda}, k}(\vec\bsH_k^\obs)| \cdot \Indc\bigg\{ \sgn(\contr^{\pi^\star_{(k+1):K}}_{\vec{\lambda}, k}(\vec\bsH_{k}^\obs)) \neq \tilde\pi^\star_k(\vec\bsH_k^\obs) \bigg\} \\
    & \qquad - \max_{a_k\in\{\pm 1\}} Q^{\pi^\star_{(k+1):K}}_{\vec{\lambda}, k}(\vec\bsH_k^\obs, a_k) 
    +
    |\contr^{\pi^\star_{(k+1):K}}_{\vec{\lambda}, k}(\vec\bsH_k^\obs)| \cdot \Indc\bigg\{ \sgn(\contr^{\pi^\star_{(k+1):K}}_{\vec{\lambda}, k}(\vec\bsH_{k}^\obs)) \neq \hat\pi^\star_k(\vec\bsH_k^\obs) \bigg\}\bigg]\\
    & = \calR_k(\hat \pi^\star_k) - \calR_k(\tilde\pi^\star_k).
  \end{align*}
  Hence, we have
  $$
    \bbE[V^{\pi^\star}_1(\bsX_1^\obs) - V^{\hat \pi^\star}_1(\bsX_1^\obs)] \leq \calE_\prox +\sum_{k=1}^K \sfc_k \cdot [\calR_k(\hat \pi^\star_k) - \calR_k(\tilde\pi^\star_k)]
  $$
  The desired result follows from Proposition \ref{prop:iv_opt_risk_bound} and a union bound over $[K]$.

\subsubsection{Proof of Lemma \ref{lemma:perf_diff}} \label{prf:lemma:perf_diff}

We start by proving a useful lemma.
\begin{lemma}
  \label{lemma:peeling}
  For any three DTRs $\pi, \pi', \pi''$, any $k\in[K]$, and any historical information $\vec\bsh_k$, we have
  \begin{align*}
    & V^{\pi_k \pi'_{(k+1):K}}_{\vec{\lambda}, k}(\vec\bsh_k) - V^{\pi_k \pi''_{(k+1):K}}_{\vec{\lambda}, k}(\vec\bsh_k)\\
    & \leq \sup_{p_{\pi_k(\vec\bsh_k)} \in \calP_{\vec\bsh_k, \pi_k(\vec\bsh_k)}} \bbE_{(R_k, \bsX_{k+1})\sim p_{\pi_k(\vec\bsh_k)}} [V^{\pi'_{(k+1):K}}_{\vec{\lambda}, k+1}(\vec\bsH_{k+1}) - V^{\pi''_{(k+1):K}}_{\vec{\lambda}, k+1}(\vec\bsh_k)],
  \end{align*}  
  where $\vec\bsH_{k+1} = (\vec\bsh_k, a_{k}, R_k, \bsX_{k+1})$.
\end{lemma}
\begin{proof}
  By definition, we have
  \begin{align*}
    & V^{\pi_k \pi'_{(k+1):K}}_{\vec{\lambda}, k}(\vec\bsh_k) - V^{\pi_k \pi''_{(k+1):K}}_{\vec{\lambda}, k}(\vec\bsh_k) \\
    & = \lambda_k(\vec\bsh_k, \pi_k(\vec\bsh_k))  \cdot \bigg( \underline{Q}^{\pi'_{(k+1):K}}_{\vec{\lambda}, k}(\vec\bsh_k, \pi_k(\vec\bsh_k)) - \underline{Q}^{\pi''_{(k+1):K}}_{\vec{\lambda}, k}(\vec\bsh_k, \pi_k(\vec\bsh_k)) \bigg) \\
    & \qquad + 
    (1-\lambda_k(\vec\bsh_k, \pi_k(\vec\bsh_k))) \cdot \bigg( \overline{Q}^{\pi'_{(k+1):K}}_{\vec{\lambda}, k}(\vec\bsh_k, \pi_k(\vec\bsh_k)) - \overline{Q}^{\pi''_{(k+1):K}}_{\vec{\lambda}, k}(\vec\bsh_k, \pi_k(\vec\bsh_k)) \bigg) \\
    & = \lambda_k(\vec\bsh_k, \pi_k(\vec\bsh_k)) \cdot 
    \bigg(
      \inf_{p_{\pi_k(\vec\bsh_k)}}\bbE_{(R_k, \bsX_{k+1})\sim p_{\pi_k(\vec\bsh_k)}}[R_k + V^{\pi'_{(k+1):K}}_{\vec{\lambda}, k+1}(\vec\bsH_{k+1})] \\
      & \qquad\qquad\qquad\qquad\qquad\qquad - 
      \inf_{p_{\pi_k(\vec\bsh_k)}}\bbE_{(R_k, \bsX_{k+1})\sim p_{\pi_k(\vec\bsh_k)}}[R_k + V^{\pi''_{(k+1):K}}_{\vec{\lambda}, k+1}(\vec\bsH_{k+1})] 
    \bigg)\\
    & \qquad + (1-\lambda_k(\vec\bsh_k, \pi_k(\vec\bsh_k))) \cdot
    \bigg(
      \sup_{p_{\pi_k(\vec\bsh_k)}}\bbE_{(R_k, \bsX_{k+1})\sim p_{\pi_k(\vec\bsh_k)}}[R_k + V^{\pi'_{(k+1):K}}_{\vec{\lambda}, k+1}(\vec\bsH_{k+1})] \\
      & \qquad\qquad\qquad\qquad\qquad\qquad\qquad\qquad - 
      \sup_{p_{\pi_k(\vec\bsh_k)}}\bbE_{(R_k, \bsX_{k+1})\sim p_{\pi_k(\vec\bsh_k)}}[R_k + V^{\pi''_{(k+1):K}}_{\vec{\lambda}, k+1}(\vec\bsH_{k+1})] 
    \bigg)\\
    & \leq \lambda_k(\vec\bsh_k, \pi_k(\vec\bsh_k)) \cdot 
    \bigg(
      \inf_{p_{\pi_k(\vec\bsh_k)}}\bbE_{(R_k, \bsX_{k+1})\sim p_{\pi_k(\vec\bsh_k)}}[R_k + V^{\pi''_{(k+1):K}}_{\vec{\lambda}, k+1}(\vec\bsH_{k+1})] \\
      & \qquad\qquad\qquad\qquad\qquad\qquad + 
      \sup_{p_{\pi_k(\vec\bsh_k)}}\bbE_{(R_k, \bsX_{k+1})\sim p_{\pi_k(\vec\bsh_k)}}[V^{\pi'_{(k+1):K}}_{\vec{\lambda}, k+1}(\vec\bsH_{k+1})- V^{\pi''_{(k+1):K}}_{\vec{\lambda}, k+1}(\vec\bsH_{k+1})] \\
      & \qquad\qquad\qquad\qquad\qquad\qquad - 
      \inf_{p_{\pi_k(\vec\bsh_k)}}\bbE_{(R_k, \bsX_{k+1})\sim p_{\pi_k(\vec\bsh_k)}}[R_k + V^{\pi''_{(k+1):K}}_{\vec{\lambda}, k+1}(\vec\bsH_{k+1})] 
    \bigg)\\
    & \qquad + (1-\lambda_k(\vec\bsh_k, \pi_k(\vec\bsh_k))) \cdot
    \bigg(
      \sup_{p_{\pi_k(\vec\bsh_k)}}\bbE_{(R_k, \bsX_{k+1})\sim p_{\pi_k(\vec\bsh_k)}}[R_k + V^{\pi'_{(k+1):K}}_{\vec{\lambda}, k+1}(\vec\bsH_{k+1})] \\
      & \qquad\qquad\qquad\qquad\qquad\qquad\qquad - 
      \sup_{p_{\pi_k(\vec\bsh_k)}}\bbE_{(R_k, \bsX_{k+1})\sim p_{\pi_k(\vec\bsh_k)}}[R_k + V^{\pi'_{(k+1):K}}_{\vec{\lambda}, k+1}(\vec\bsH_{k+1})]  \\
      & \qquad\qquad\qquad\qquad\qquad\qquad\qquad - 
      \inf_{p_{\pi_k(\vec\bsh_k)}}\bbE_{(R_k, \bsX_{k+1})\sim p_{\pi_k(\vec\bsh_k)}}[V^{\pi''_{(k+1):K}}_{\vec{\lambda}, k+1}(\vec\bsH_{k+1}) - V^{\pi'_{(k+1):K}}_{\vec{\lambda}, k+1}(\vec\bsH_{k+1})] 
    \bigg)\\
    & = \sup_{p_{\pi_k(\vec\bsh_k)}}\bbE_{(R_k, \bsX_{k+1})\sim p_{\pi_k(\vec\bsh_k)}}[V^{\pi'_{(k+1):K}}_{\vec{\lambda}, k+1}(\vec\bsH_{k+1})- V^{\pi''_{(k+1):K}}_{\vec{\lambda}, k+1}(\vec\bsH_{k+1})],
  \end{align*}
  which is the desired result.
\end{proof}

To prove Lemma \ref{lemma:perf_diff}, we start by decomposing the left-hand side of \eqref{eq:perf_diff} into telescoping sums:
  \begin{align*}
  \bbE[V^{\pi}_{\lambda, 1}(\bsX_1) - V^{\pi'}_{\lambda, 1}(\bsX_1)] & = \bbE\bigg[\sum_{k = 1}^K V_{\vec{\lambda}, 1}^{\pi'_{1:(k-1)}, \pi_{k:K}}(\bsX_1) - V_{\vec{\lambda}, 1}^{\pi'_{1:k}, \pi_{(k+1):K}}(\bsX_1) \bigg].
  \end{align*}
  Thus, it suffices to show 
  \begin{equation*}
    \bbE[V_{\vec{\lambda}, 1}^{\pi'_{1:(k-1)}, \pi_{k:K}}(\bsX_1) - V_{\vec{\lambda}, 1}^{\pi'_{1:k}, \pi_{(k+1):K}}(\bsX_1)] 
    \leq \sfc_k \cdot \bbE_{\vec\bsH_k \sim q^{\obs}_k}[V_{\vec{\lambda}, k}^{\pi_{k:K}}(\vec\bsH_k) - V^{\pi'_k, \pi_{(k+1):K}}_{\vec{\lambda}, k}(\vec\bsH_k)]
  \end{equation*}
  for any $k\in[K]$.
  For $k = 1$, the above result is immediate as $\sfc_1 = 1$. For $k \geq 2$, iteratively invoking Lemma \ref{lemma:peeling} gives
  \begin{align*}
    & \bbE[V_{\vec{\lambda}, 1}^{\pi'_{1:(k-1)}, \pi_{k:K}}(\bsX_1) - V_{\vec{\lambda}, 1}^{\pi'_{1:k}, \pi_{(k+1):K}}(\bsX_1)] \\
    & \leq 
    \bbE_{\bsX_1}\sup_{p_{\pi'_1(\bsX_1)}} \bbE_{(R_1, \bsX_{2})\sim p_{\pi'_1(\bsX_1)}} [V^{\pi'_{2:(k-1)} \pi_{k:K} }_{\vec{\lambda}, 2}(\vec\bsH_{2}) - V^{\pi'_{2:k} \pi_{(k+1):K}}_{\vec{\lambda}, 2}(\vec\bsH_2)] \\
    & \leq 
    \bbE_{\bsX_1}\sup_{p_{\pi'_1(\bsX_1)}} \bbE_{(R_1, \bsX_{2})\sim p_{\pi'_1(\bsX_1)}} \sup_{p_{\pi_2'(\vec\bsH_2)}} \bbE_{(R_2, \bsX_3)\sim p_{\pi'_2(\vec\bsH_2)}} [V^{\pi'_{3:(k-1)} \pi_{k:K} }_{\vec{\lambda}, 3}(\vec\bsH_{3}) - V^{\pi'_{3:k} \pi_{(k+1):K}}_{\vec{\lambda}, 3}(\vec\bsH_3)] \\
    & \leq \cdots\\
    & \leq \bbE_{\bsX_1}\sup_{p_{\pi'_1(\bsX_1)}} \bbE_{(R_1, \bsX_{2})\sim p_{\pi'_1(\bsX_1)}} \cdots \sup_{p_{\pi_{k-1}'(\vec\bsH_{k-1})}} \bbE_{(R_{k-1}, \bsX_k)\sim p_{\pi'_{k-1}(\vec\bsH_{k-1})}} [V^{\pi_{k:K} }_{\vec{\lambda}, k}(\vec\bsH_{k}) - V^{\pi'_{k} \pi_{(k+1):K}}_{\vec{\lambda}, k}(\vec\bsH_k)]\\
    & \leq \sfc_k \cdot \bbE_{\vec\bsH_k\sim q^\obs_k} \bbE_{\vec\bsH_k \sim q^{\obs}_k}[V_{\vec{\lambda}, k}^{\pi_{k:K}}(\vec\bsH_k) - V^{\pi'_k, \pi_{(k+1):K}}_{\vec{\lambda}, k}(\vec\bsH_k)],
  \end{align*}
  where the last inequality is by Assumption \ref{assump:bdd_concentration_coef}.
  And the proof is concluded.

\subsubsection{Proof of Lemma \ref{lemma:risk_and_regret}} \label{prf:lemma:risk_and_regret}
  By definition, we have  
  \begin{align*}
    V^{\pi}_{\vec{\lambda}, k}(\vec\bsh_k) & = Q^{\pi}_{\vec{\lambda}, k}(\vec\bsh_k, +1) \cdot \indc{\pi_k(\vec\bsh_k) = +1} + Q^\pi_{\vec{\lambda}, k}(\vec\bsh_k , -1) \cdot \indc{\pi_k(\vec\bsh_k) = -1} \\
    & = Q^\pi_{\vec{\lambda}, k}(\vec\bsh_k , +1) - \indc{\pi_k(\vec\bsh_k) = -1} \cdot \contr^{\pi}_{\vec{\lambda}, k}(\vec\bsh_{k}) \\
    & = Q^\pi_{\vec{\lambda}, k}(\vec\bsh_k , +1) 
    - \bigg(
      \Indc\Big\{\pi_k(\vec\bsh_k) = -1, \sgn(\contr^{\pi}_{\vec{\lambda}, k}(\vec\bsh_k)) = +1\Big\} \cdot  |\contr^{\pi}_{\vec{\lambda}, k}(\vec\bsh_k) |\\
      & \qquad \qquad \qquad \qquad \qquad  - \Indc\Big\{\pi_k(\vec\bsh_k) = -1, \sgn(\contr^{\pi}_{\vec{\lambda}, k}(\vec\bsh_k)) = -1\Big\} \cdot  |\contr^{\pi}_{\vec{\lambda}, k}(\vec\bsh_k) | 
    \bigg)\\
    & = Q^\pi_{\vec{\lambda}, k}(\vec\bsh_k , +1) 
    - \bigg(
      \Indc\Big\{\pi_k(\vec\bsh_k) = -1, \sgn(\contr^{\pi}_{\vec{\lambda}, k}(\vec\bsh_k)) = +1\Big\} \cdot  |\contr^{\pi}_{\vec{\lambda}, k}(\vec\bsh_k) |\\
      & \qquad \qquad \qquad \qquad \qquad  - \Indc\Big\{\sgn(\contr^{\pi}_{\vec{\lambda}, k}(\vec\bsh_k)) = -1\Big\} \cdot  |\contr^{\pi}_{\vec{\lambda}, k}(\vec\bsh_k) | \\
      & \qquad \qquad \qquad \qquad \qquad  + \Indc\Big\{\pi_k(\vec\bsh_k) = +1, \sgn(\contr^{\pi}_{\vec{\lambda}, k}(\vec\bsh_k)) = -1\Big\} \cdot  |\contr^{\pi}_{\vec{\lambda}, k}(\vec\bsh_k) |  
    \bigg)\\
    & = Q^\pi_{\vec{\lambda}, k}(\vec\bsh_k , +1)  
    + \Indc\Big\{\sgn(\contr^{\pi}_{\vec{\lambda}, k}(\vec\bsh_k)) = -1\Big\} \cdot  |\contr^{\pi}_{\vec{\lambda}, k}(\vec\bsh_k) | 
    - |\contr^{\pi}_{\vec{\lambda}, k}(\vec\bsh_k)| \cdot \Indc\Big\{ \sgn(\contr^{\pi}_{\vec{\lambda}, k}(\vec\bsh_{k})) \neq \pi_k(\vec\bsh_k) \Big\}\\
    & = \max_{a_k\in\{\pm 1\}} Q^{\pi}_{\vec{\lambda}, k}(\vec\bsh_k, a_k) - |\contr^{\pi}_{\vec{\lambda}, k}| \cdot \Indc\bigg\{ \sgn(\contr^{\pi}_{\vec{\lambda}, k}(\vec\bsh_{k})) \neq \pi_k(\vec\bsh_k) \bigg\},
  \end{align*}
  and the proof is concluded.

\subsubsection{Proof of Proposition \ref{prop:iv_opt_risk_bound}} \label{prf:prop:iv_opt_risk_bound}

We begin by proving two useful lemmas.

\begin{lemma}[Uniform concentration of the risk]
\label{lemma:unif_concentration}
  Fix $k\in[K]$ and $\delta\in(0, 1)$. With probability at least $1-\delta$, we have
  \begin{align*}
    \sup_{\pi_k \in \Pi_k} \bigg| \hat \calR_k(\pi_k) -  \calR_k(\pi_k) \bigg| \lesssim (K - k + 1)\cdot \sqrt{\frac{\vc(\Pi_k) + \log(1/\delta)}{n}},
  \end{align*}
  where 
  $$
    \hat \calR_k(\pi_k) =  \frac{1}{n} \sum_{i =1}^n  |\contr^{\pi^\star}_{\vec{\lambda}, k}(\vec\bsh_{k, i})| \cdot \Indc\bigg\{\sgn(\contr^{\pi^\star}_{\vec{\lambda}, k}(\vec\bsh_{k, i})) \neq \pi_k(\vec\bsh_{k, i})\bigg\} 
  $$
  and $\calR_k(\pi_k)$ is defined in \eqref{eq:pop_risk_iv_opt_dtr}.
\end{lemma}

\begin{proof}
Let $f^{\pi_k}(\{\vec\bsh_{k,i}\}):= \hat \calR_k(\pi_k) -  \calR_k(\pi_k)$. For an arbitrary index $j\in[n]$, introduce $\{\vec\bsh_{k, i}'\}_{i=1}^n$, where $\vec\bsh_{k,i}' = \vec\bsh_{k, i}$ for any $i\neq j$, and $\vec\bsh_{k, j}'$ is an independent copy of $\vec\bsh_{k, j}$. We then have
\begin{align*}
  & |f^{\pi_k}(\{\vec\bsh_{k, i}\})| - \sup_{\pi_k'\in\Pi_k}
  |f^{\pi_k'}(\{\vec\bsh_{k, i}'\})| \\
  & \leq |f^{\pi_k}(\{\vec\bsh_{k,i}\})| -|f^{\pi_k}(\{\vec\bsh_{k, i}'\})|\\
  & \leq \frac{1}{n} \bigg| |\contr^{\pi^\star}_{\vec{\lambda}, k} (\vec\bsh_{k, j})| \cdot \Indc\bigg\{ \sgn(\contr_{\vec{\lambda}, k}^{\pi^\star}(\vec\bsh_{k,j})) \neq \pi_k(\vec\bsh_{k, j}) \bigg\}
  -
  |\contr^{\pi^\star}_{\vec{\lambda}, k} (\vec\bsh_{k, j}')| \cdot \Indc\bigg\{ \sgn(\contr_{\vec{\lambda}, k}^{\pi^\star}(\vec\bsh_{k,j}')) \neq \pi_k(\vec\bsh_{k, j}') \bigg\} 
  \bigg| \\
  & \lesssim \frac{(K-k+1)}{n},
\end{align*}  
where the last step is by our assumption that the reward at each step is bounded. Thus, we have 
$$
  \sup_{\pi_k\in\Pi_k}|f^{\pi_k}(\{\vec\bsh_{k, i}\})| - \sup_{\pi_k'\in\Pi_k}|f^{\pi_k'}(\{\vec\bsh_{k, i}'\})|  \lesssim \frac{(K-k+1)}{n}.
$$
A symmetric argument further gives 
$$
\Big|\sup_{\pi_k\in\Pi_k} |f^{\pi_k}(\{\vec\bsh_{k, i}\})| - \sup_{\pi_k'\in\Pi_k} |f^{\pi_k'}(\{\vec\bsh_{k, i}'\})|  \Big| \lesssim \frac{(K-k+1)}{n},
$$
which allows us to invoke McDiarmid's inequality to conclude that
\begin{equation}
  \label{eq:unif_concentration_around_expectation}
  \sup_{\pi_k\in\Pi_k}|f^{\pi_k}(\{\vec\bsh_{k, i}\})| - \bbE \sup_{\pi_k\in\Pi_k} |f^{\pi_k}(\{\vec\bsh_{k, i}\})|\lesssim \frac{(K-k+1)\sqrt{\log(1/\delta)}}{\sqrt{n}}
\end{equation}  
with probability at least $1-\delta$. We then focus on bounding the expectation term. Using a standard symmetrization argument (see, e.g., Lemma 11.4 of \citealt{boucheron2013concentration}), we get
\begin{equation}
  \label{eq:unif_concentration_symmetrization}
  \bbE \sup_{\pi_k\in\Pi_k} |f^{\pi_k}(\{\vec\bsh_{k, i}\})| \leq \frac{2}{\sqrt{n}} \bbE_{\{\vec\bsh_{k, i}\}} \bbE_{\{\xi_i\}}\Big[ \sup_{\pi_k\in\Pi_k} |M_{\pi_k}| \Big] 
\end{equation}  
where
\begin{align*}
  M_{\pi_k} & = \frac{1}{\sqrt{n}}\sum_{i=1}^n \xi_i \bigg( |\contr^{\pi^\star}_{\vec{\lambda}, k}(\vec\bsh_{k, i})| \indc{\sgn(\contr^{\pi^\star}_{\vec{\lambda}, k}(\vec\bsh_{k, i}))\neq \pi_k(\vec\bsh_{k,i})}  - \bbE\Big[|\contr^{\pi^\star}_{\vec{\lambda}, k}(\vec\bsh_{k, i})| \indc{\sgn(\contr^{\pi^\star}_{\vec{\lambda}, k}(\vec\bsh_{k, i}))\neq \pi_k(\vec\bsh_{k,i})}\Big]\bigg)
\end{align*}  
is the empirical process indexed by $\pi_k \in \Pi_k$ and $\{\xi_i\}$ are i.i.d.~Rademacher random variables. Since $M_{\pi_k^\star} = 0$, we can further bound the right-hand side of  \eqref{eq:unif_concentration_symmetrization} by
$$
  \frac{2}{\sqrt{n}} \bbE_{\{\vec\bsh_{k, i}\}} \bbE_{\{\xi_i\}}\Big[ \sup_{\pi_k\in\Pi_k\cup\{\pi_k^\star\}} |M_{\pi_k} - M_{\pi_k^\star}| \Big].
$$
Conditional on $\{\xi_i\}$, $M_{\pi_k}$ is a Rademacher process satisfying
$
  \|M_{\pi_k}, M_{\pi_k}'\|_{\psi_2} \lesssim \|f^{\pi_k} - f^{\pi_k'}\|_{L^2(\mu_n)},
$
where $\|\cdot\|_{\psi_2}$ is the sub-Gaussian norm (see, e.g., Definition 2.5.6 in \citealt{vershynin2018high}) and $\mu_n$ is the empirical measure on $\{\vec\bsh_{k, i}\}_{i=1}^n$.
Using Dudley’s integral inequality (see, e.g., Theorem 8.1.3 in \citealt{vershynin2018high}), we get 
\begin{equation}
  \label{eq:unif_concentration_dudley}
  \bbE_{\{\xi_i\}}\Big[ \sup_{\pi_k\in\Pi_k\cup\{\pi_k^\star\}} |M_{\pi_k} - M_{\pi_k^\star}| \Big] \lesssim \int_0^{D} \sqrt{\log \calN(\calF, L^2(\mu_n), \ep)} d \ep,
\end{equation}  
where 
$
\calF = \big\{ f^{\pi_k}: \pi_k\in\Pi_k\cup\{\pi_k^\star\} \big\}
$
is the function class under consideration, $\calN(\calF, L^2(\mu_n), \ep)$ is the $\ep$-covering number, defined as the minimum number of $L^2(\mu_n)$ balls with radius $\ep$ required to cover $\calF$, and $D$ is the radius of $\calF$ with respect to $L^2(\mu_n)$ metric.
We then relate the covering number of $\calF$ to the VC dimension of $\Pi_k$. To do this, note that
\begin{align*}
  & \Indc\Big\{ \sgn(\hat \contr^\star_{\vec{\lambda}, k} (\vec\bsh_{k, i}; B_{-j_i})) \neq \pi_k(\vec\bsh_{k,i}) \Big\}
  -
  \Indc\Big\{ \sgn(\hat \contr^\star_{\vec{\lambda}, k} (\vec\bsh_{k, i}; B_{-j_i})) \neq \pi_k'(\vec\bsh_{k,i}) \Big\} \\
  & = \frac{1}{4} \bigg[ \bigg(  \sgn(\hat \contr^\star_{\vec{\lambda}, k} (\vec\bsh_{k, i}; B_{-j_i}))  - \pi_k(\vec\bsh_{k, i})\bigg)^2 - \bigg( \sgn(\hat \contr^\star_{\vec{\lambda}, k} (\vec\bsh_{k, i}; B_{-j_i}))  - \pi'_k(\vec\bsh_{k, i})\bigg)^2 \bigg]\\
  & = \frac{1}{4} \bigg(2 \cdot \sgn(\contr^{\pi^\star}_{\vec{\lambda}, k}(\vec\bsh_{k, i})) - \pi_k(\vec\bsh_{k, i}) - \pi_k'(\vec\bsh_{k, i})\bigg) \bigg( \pi_k'(\vec\bsh_{k, i}) - \pi_k(\vec\bsh_{k, i})\bigg).
\end{align*}
With the above display, we have
\begin{align*}
  & \|f^{\pi_k} - f^{\pi_k'}\|_{L^2(\mu_n)} \\
  & = \frac{1}{n} \sum_{i=1}^n 
  \bigg(
  |\contr^{\pi^\star}_{\vec{\lambda}, k} (\vec\bsh_{k, i})| \cdot \Indc\Big\{ \sgn(\hat \contr^\star_{\vec{\lambda}, k} (\vec\bsh_{k, i}; B_{-j_i})) \neq \pi_k(\vec\bsh_{k,i}) \Big\} 
  - 
  \bbE\Big[|\contr^{\pi^\star}_{\vec{\lambda}, k} (\vec\bsh_{k, i})| \cdot \Indc\Big\{ \sgn(\hat \contr^\star_{\vec{\lambda}, k} (\vec\bsh_{k, i}; B_{-j_i})) \neq \pi_k(\vec\bsh_{k,i}) \Big\}\Big] \\
  & \qquad - 
  |\contr^{\pi^\star}_{\vec{\lambda}, k} (\vec\bsh_{k, i})| \cdot \Indc\Big\{ \sgn(\hat \contr^\star_{\vec{\lambda}, k} (\vec\bsh_{k, i}; B_{-j_i})) \neq \pi_k'(\vec\bsh_{k,i}) \Big\} 
  +
  \bbE\Big[|\contr^{\pi^\star}_{\vec{\lambda}, k} (\vec\bsh_{k, i})| \cdot \Indc\Big\{ \sgn(\hat \contr^\star_{\vec{\lambda}, k} (\vec\bsh_{k, i}; B_{-j_i})) \neq \pi_k'(\vec\bsh_{k,i}) \Big\}\Big]
  \bigg)^2\\
  & = \frac{1}{n} \sum_{i=1}^n \bigg( \frac{|\contr^{\pi^\star}_{\vec{\lambda}, k}(\vec\bsh_{k, i})|}{4} \cdot \Big( 2 \cdot \sgn(\contr^{\pi^\star}_{\vec{\lambda}, k}(\vec\bsh_{k, i})) - \pi_k(\vec\bsh_{k, i}) - \pi_k'(\vec\bsh_{k, i}) \Big)\Big( \pi_k'(\vec\bsh_{k, i}) - \pi_k(\vec\bsh_{k, i})\Big)\\
  & \qquad 
  - 
  \bbE\Big[\frac{|\contr^{\pi^\star}_{\vec{\lambda}, k}(\vec\bsh_{k, i})|}{4} \cdot \Big( 2 \cdot \sgn(\contr^{\pi^\star}_{\vec{\lambda}, k}(\vec\bsh_{k, i})) - \pi_k(\vec\bsh_{k, i}) - \pi_k'(\vec\bsh_{k, i}) \Big)\Big( \pi_k'(\vec\bsh_{k, i}) - \pi_k(\vec\bsh_{k, i})\Big)\Big]
  \bigg)^2\\
  & \lesssim \frac{(K-k+1)^2 }{n} \sum_{i=1}^n \bigg(\Big( \pi_k'(\vec\bsh_{k, i}) - \pi_k(\vec\bsh_{k, i}) \Big)^2 +  \Big( \bbE[|\pi_k'(\vec\bsh_{k, i}) - \pi_k(\vec\bsh_{k, i})|] \Big)^2\bigg)\\
  & \leq \frac{(K-k+1)^2 }{n} \sum_{i=1}^n \Big( \pi_k'(\vec\bsh_{k, i}) - \pi_k(\vec\bsh_{k, i}) \Big)^2 
  +
  \bbE\bigg[\frac{(K-k+1)^2 }{n} \sum_{i=1}^n \Big(\pi_k'(\vec\bsh_{k, i}) - \pi_k(\vec\bsh_{k, i})\Big)^2\bigg],
\end{align*}
where the last line follows from Jensen's inequality. Recalling the definition of the covering number, we have established that
$$
  \log \calN(\calF, L^2(\mu_n), \ep) \leq \log \calN\bigg(\Pi_k\cup\{\pi^\star_k\}, L^2(\mu_n), \frac{c\cdot \ep}{(K-k+1)}\bigg),
$$
where $c>0$ is an absolute constant. Since $\vc(\Pi_{k}\cup\{\pi^\star_k\})\asymp \vc(\Pi_k)$, we can invoke Theorem 8.3.18 in \citet{vershynin2018high} to get
$$
\log \calN(\calF, L^2(\mu_n), \ep) \lesssim \vc(\Pi_k) \cdot \log\bigg(\frac{2 (K-k+1)}{c\cdot\ep}\bigg).
$$
Plugging the above display to \eqref{eq:unif_concentration_dudley} and noting that $D \lesssim {K-k+1}$ gives that ($c'$ is another absolute constant)
\begin{align*}
  \bbE_{\{\xi_i\}}\Big[ \sup_{\pi_k\in\Pi_k\cup\{\pi_k^\star\}} |M_{\pi_k} - M_{\pi_k^\star}| \Big] 
  & \lesssim \sqrt{\vc(\Pi_k)} \int_0^{c'\cdot({K-k+1})} \sqrt{ \log\bigg(\frac{2 (K-k+1)}{c\cdot\ep}\bigg)} d\ep \\
  & = \sqrt{\vc(\Pi_k)}\frac{2(K-k+1)}{c} \int_0^{\frac{cc'}{2}} \sqrt{\log(1/u)}du\\
  & \lesssim (K-k+1) \sqrt{\vc(\Pi_k)}.
\end{align*}  
Plugging the above inequality to \eqref{eq:unif_concentration_symmetrization} and recalling \eqref{eq:unif_concentration_around_expectation}, the proof is concluded.
\end{proof}

\begin{lemma}[Crossing fitting]
\label{lemma:coupling}
Fix $k\in[K]$ and $\delta\in(0, 1)$. Let Assumption \ref{assump:contrast_estimation_oracle} hold. 
With probability at least $1-\delta$, we have
\begin{align*}  
  & \sup_{\pi_k\in\Pi_k} \Bigg\{\frac{1}{n} \sum_{i=1}^n |\contr^{\pi^\star}_{\vec{\lambda}, k} (\vec\bsh_{k, i})| 
  \cdot 
  \Indc\Big\{ \sgn(\contr^{\pi^\star}_{\vec{\lambda}, k} (\vec\bsh_{k, i})) \neq \hat \pi_k^\star(\vec\bsh_{k, i}) \Big\} \\
  & \qquad \qquad - \frac{1}{n} \sum_{i=1}^n |\hat \contr^\star_{\vec{\lambda}, k} (\vec\bsh_{k, i}; B_{-j_i})| \cdot \Indc\Big\{ \sgn(\hat \contr^\star_{\vec{\lambda}, k} (\vec\bsh_{k, i}; B_{-j_i})) \neq \hat \pi^\star_k(\vec\bsh_{k, i}) \Big\} \Bigg\} \\
  \label{eq:coupling}
  & \lesssim  \sfC_{k, \frac{\delta}{2m}} \cdot n^{-\alpha_k} + (K-k+1)\sqrt{\frac{\log(1/\delta)}{n}} .\numberthis
\end{align*}
\end{lemma}
\begin{proof}
For each $i\in[n]$, we have
\begin{align*}
  & |\contr^{\pi^\star}_{\vec{\lambda}, k} (\vec\bsh_{k, i})| \cdot \Indc\Big\{ \sgn(\contr^{\pi^\star}_{\vec{\lambda}, k} (\vec\bsh_{k, i})) \neq \hat \pi_k^\star(\vec\bsh_{k, i}) \Big\} 
  - 
  |\hat \contr^\star_{\vec{\lambda}, k} (\vec\bsh_{k, i}; B_{-j_i})| \cdot \Indc\Big\{ \sgn(\hat \contr^\star_{\vec{\lambda}, k} (\vec\bsh_{k, i}; B_{-j_i})) \neq \hat \pi^\star_k(\vec\bsh_{k, i}) \Big\} \\
  & = 
  \frac{1}{4} 
  \bigg[
    |\contr^{\pi^\star}_{\vec{\lambda}, k}(\vec\bsh_{k, i})| \cdot \Big( \sgn(\contr^{\pi^\star}_{\vec{\lambda}, k}(\vec\bsh_{k, i})) - \hat\pi^\star(\vec\bsh_{k, i})\Big)^2
    -
    |\hat \contr^\star_{\vec{\lambda}, k} (\vec\bsh_{k, i}; B_{-j_i})| \cdot \Big(\sgn(\contr^{\pi^\star}_{\vec{\lambda}, k}(\vec\bsh_{k, i})) - \hat\pi^\star_k(\vec\bsh_{k, i})\Big)^2\\
    & ~~ + 
    |\hat \contr^\star_{\vec{\lambda}, k} (\vec\bsh_{k, i}; B_{-j_i})| \cdot \Big(\sgn(\contr^{\pi^\star}_{\vec{\lambda}, k}(\vec\bsh_{k, i})) - \hat\pi^\star_k(\vec\bsh_{k, i})\Big)^2
    - 
    |\hat \contr^\star_{\vec{\lambda}, k} (\vec\bsh_{k, i}; B_{-j_i})| \cdot 
    \Big(\sgn(\hat \contr^\star_{\vec{\lambda}, k} (\vec\bsh_{k, i}; B_{-j_i})) - \hat \pi^\star_k(\vec\bsh_{k, i})  \Big)^2
  \bigg]\\
  & = \frac{1}{4} \bigg[ \underbrace{\Big(\contr^{\pi^\star}_{\vec{\lambda}, k}(\vec\bsh_{k, i}) - \hat \contr^\star_{\vec{\lambda}, k} (\vec\bsh_{k, i}; B_{-j_i})\Big) \Big( \sgn(\contr^{\pi^\star}_{\vec{\lambda}, k}(\vec\bsh_{k, i})) - \hat\pi^\star_k(\vec\bsh_{k, i})\Big)^2}_{\RN{1}} \\
  & ~~ + \underbrace{|\hat \contr^\star_{\vec{\lambda}, k} (\vec\bsh_{k, i}; B_{-j_i})| \Big( \sgn(\contr^{\pi^\star}_{\vec{\lambda}, k}(\vec\bsh_{k, i})) + \sgn(\hat \contr^\star_{\vec{\lambda}, k} (\vec\bsh_{k, i}; B_{-j_i}))
   - 2 \hat\pi^\star_k(\vec\bsh_{k, i}) \Big)\Big( \sgn(\contr^{\pi^\star}_{\vec{\lambda}, k}(\vec\bsh_{k, i})) - \sgn(\hat \contr^\star_{\vec{\lambda}, k} (\vec\bsh_{k, i}; B_{-j_i})) \Big)}_{\RN{2}}
  \bigg].
\end{align*}
It is clear that Term $\RN{1}\lesssim |\contr^{\pi^\star}_{\vec{\lambda}, k}(\vec\bsh_{k, i}) - \hat \contr^\star_{\vec{\lambda}, k} (\vec\bsh_{k, i}; B_{-j_i})|$. Note that Term \RN{2} can be non-zero only if the sign of $\contr^{\pi^\star}_{\vec{\lambda}, k}(\vec\bsh_{k, i})$ and $\hat \contr^\star_{\vec{\lambda}, k} (\vec\bsh_{k, i}; B_{-j_i})$ disagree, in which case we have $|\hat \contr^\star_{\vec{\lambda}, k} (\vec\bsh_{k, i}; B_{-j_i})| \leq |\contr^{\pi^\star}_{\vec{\lambda}, k}(\vec\bsh_{k, i}) - \hat \contr^\star_{\vec{\lambda}, k} (\vec\bsh_{k, i}; B_{-j_i})|$. Thus, Term $\RN{2}$ also satisfies $ \RN{2}\lesssim |\contr^{\pi^\star}_{\vec{\lambda}, k}(\vec\bsh_{k, i}) - \hat \contr^\star_{\vec{\lambda}, k} (\vec\bsh_{k, i}; B_{-j_i})|$. Thus, the right-hand side of the above display can be upper bounded by a constant multiple of $|\contr^{\pi^\star}_{\vec{\lambda}, k}(\vec\bsh_{k, i}) - \hat \contr^\star_{\vec{\lambda}, k} (\vec\bsh_{k, i}; B_{-j_i})|$, which further gives
\begin{align*}
  \textnormal{LHS of \eqref{eq:coupling}} & \lesssim \frac{1}{n} \sum_{i=1}^n |\contr^{\pi^\star}_{\vec{\lambda}, k}(\vec\bsh_{k, i}) - \hat \contr^\star_{\vec{\lambda}, k} (\vec\bsh_{k, i}; B_{-j_i})|\\
  & = \frac{1}{n} \sum_{j=1}^m n_j \cdot \frac{1}{n_j}\sum_{i\in B_j} |\contr^{\pi^\star}_{\vec{\lambda}, k}(\vec\bsh_{k, i}) - \hat \contr^\star_{\vec{\lambda}, k}(\vec\bsh_{k, i}, B_{-j})|
\end{align*}
Since each $|\contr^{\pi^\star}_{\vec{\lambda}, k}(\vec\bsh_{k, i}) - \hat \contr^\star_{\vec{\lambda}, k}(\vec\bsh_{k, i}, B_{-j})|$ is bounded in $[0, C(K-k+1)]$ for some absolute constant $C>0$ and $\hat \contr^\star_{\vec{\lambda}, k}(\vec\bsh_{k, i}, B_{-j})$ is independent of $\{\vec\bsh_{k, i}: i\in B_j\}$, invoking Hoeffding's inequality and recalling Assumption \ref{assump:contrast_estimation_oracle} give that for any $j\in[m]$, 
\begin{align*}
  \frac{1}{n_j}\sum_{i\in B_j} |\contr^{\pi^\star}_{\vec{\lambda}, k}(\vec\bsh_{k, i}) - \hat \contr^\star_{\vec{\lambda}, k}(\vec\bsh_{k, i}, B_{-j})|
  & \lesssim  \bbE |\contr^{\pi^\star}_{\vec{\lambda}, k}(\vec\bsH_{k}^\obs) - \hat \contr^{\star}_{\vec{\lambda}, k} (\vec\bsH_{k}^\obs; B_{-j})|  + \frac{(K-k+1)\sqrt{\log(2/\delta)}}{\sqrt{n_j}}\\
  & \leq \sfC_{k, \delta/2} \cdot  (n-n_j)^{-\alpha_{k}} + \frac{(K-k+1)\sqrt{\log(2/\delta)}}{\sqrt{n_j}}
\end{align*}  
with probability at least $1-\delta$. Recalling that $m\asymp 1$ and invoking a union bound over $j\in[m]$, we conclude that with probability at least $1-\delta$,
\begin{align*}
  \textnormal{LHS of \eqref{eq:coupling}} 
  & \lesssim  \frac{1}{n} \sum_{j=1}^n  n_j \cdot \bigg(\sfC_{k, \frac{\delta}{2m}} \cdot n_j^{-\alpha_k}  + \frac{(K-k+1)\sqrt{\log(2m/\delta)}}{\sqrt{n_j}}\bigg)\\
  & \lesssim \sfC_{k, \frac{\delta}{2m}} \cdot n^{-\alpha_j} + \frac{(K-k+1)\sqrt{\log(1/\delta)}}{\sqrt{n}}.
\end{align*}
The proof is concluded.
\end{proof}

We are now ready to present the proof of Proposition \ref{prop:iv_opt_risk_bound}.
Recall that $\tilde\pi^\star_k$ is the regime that exactly minimizes $R_k(\pi_k)$ over $\pi_k \in \Pi_k$:
\begin{align*}
  \tilde\pi^\star_k \in \argmin_{\pi_k \in \Pi_k} \bbE\bigg[ |\contr^{\pi^\star}_{\vec{\lambda}, k}(\vec\bsH^\obs_k)| \cdot \Indc\bigg\{\sgn(\contr^{\pi^\star}_{\vec{\lambda}, k}(\vec\bsH_k^\obs)) \neq \pi_k(\vec\bsH^\obs_k)\bigg\} \bigg].
\end{align*}
Then we have
\begin{align*}
  R_k(\hat \pi^\star_k) - R_k(\tilde\pi^\star_k) = \mathscr{T}_{1} + \mathscr{T}_{2} + \mathscr{T}_{3} + \mathscr{T}_{4} + \mathscr{T}_{5},
\end{align*}
where
\begin{align*}
  \mathscr{T}_1 
  & = \bbE\bigg[ |\contr^{\pi^\star}_{\vec{\lambda}, k}(\vec\bsH_k^\obs)| \cdot \Indc\bigg\{\sgn(\contr^{\pi^\star}_{\vec{\lambda}, k}(\vec\bsH_k^|obs)) \neq \hat \pi^\star_k(\vec\bsH_k^\obs)\bigg\} \bigg] \\
  & \qquad - \frac{1}{n} \sum_{i=1}^n |\contr^{\pi^\star}_{\vec{\lambda}, k} (\vec\bsh_{k, i})| \cdot \Indc\bigg\{ \sgn(\contr^{\pi^\star}_{\vec{\lambda}, k} (\vec\bsh_{k, i})) \neq \hat \pi_k^\star(\vec\bsh_{k, i}) \bigg\},  \\
  \mathscr{T}_2
  & = \frac{1}{n} \sum_{i=1}^n |\contr^{\pi^\star}_{\vec{\lambda}, k} (\vec\bsh_{k, i})| \cdot \Indc\bigg\{ \sgn(\contr^{\pi^\star}_{\vec{\lambda}, k} (\vec\bsh_{k, i})) \neq \hat \pi_k^\star(\vec\bsh_{k, i}) \bigg\} \\
  & \qquad - \frac{1}{n} \sum_{i=1}^n |\hat \contr^\star_{\vec{\lambda}, k} (\vec\bsh_{k, i}; B_{-j_i})| \cdot \Indc\bigg\{ \sgn(\hat \contr^\star_{\vec{\lambda}, k} (\vec\bsh_{k, i}; B_{-j_i})) \neq \hat \pi^\star_k(\vec\bsh_{k, i}) \bigg\},\\
  \mathscr{T}_3
  & = \frac{1}{n} \sum_{i=1}^n |\hat \contr^\star_{\vec{\lambda}, k} (\vec\bsh_{k, i}; B_{-j_i})| \cdot \Indc\bigg\{ \sgn(\hat \contr^\star_{\vec{\lambda}, k} (\vec\bsh_{k, i}; B_{-j_i})) \neq \hat \pi^\star_k(\vec\bsh_{k, i}) \bigg\} \\
  & \qquad - \frac{1}{n} \sum_{i=1}^n |\hat \contr^\star_{\vec{\lambda}, k} (\vec\bsh_{k, i}; B_{-j_i})| \cdot \Indc\bigg\{ \sgn(\hat \contr^\star_{\vec{\lambda}, k} (\vec\bsh_{k, i}; B_{-j_i})) \neq \tilde\pi^\star_k(\vec\bsh_{k, i}) \bigg\},\\
  \mathscr{T}_4
  & = \frac{1}{n} \sum_{i=1}^n |\hat \contr^\star_{\vec{\lambda}, k} (\vec\bsh_{k, i}; B_{-j_i})| \cdot \Indc\bigg\{ \sgn(\hat \contr^\star_{\vec{\lambda}, k} (\vec\bsh_{k, i}; B_{-j_i})) \neq \tilde\pi^\star_k(\vec\bsh_{k, i}) \bigg\} \\
  & \qquad - \frac{1}{n} \sum_{i=1}^n |\contr^{\pi^\star}_{\vec{\lambda}, k} (\vec\bsh_{k, i})| \cdot \Indc\bigg\{ \sgn(\contr^{\pi^\star}_{\vec{\lambda}, k} (\vec\bsh_{k, i})) \neq \tilde\pi^\star_k(\vec\bsh_{k, i}) \bigg\}, \\
  \mathscr{T}_5
  & = 
  \frac{1}{n} \sum_{i=1}^n |\contr^{\pi^\star}_{\vec{\lambda}, k} (\vec\bsh_{k, i})| \cdot \Indc\bigg\{ \sgn(\contr^{\pi^\star}_{\vec{\lambda}, k} (\vec\bsh_{k, i})) \neq \tilde\pi^\star_k(\vec\bsh_{k, i}) \bigg\}
  \\
  & \qquad - \bbE\bigg[ |\contr^{\pi^\star}_{\vec{\lambda}, k}(\vec\bsH_k^\obs)| \cdot \Indc\bigg\{\sgn(\contr^{\pi^\star}_{\vec{\lambda}, k}(\vec\bsH_k^\obs)) \neq \tilde\pi^\star_k(\vec\bsH_k^\obs)\bigg\} \bigg].
\end{align*}
Since $\hat\pi^\star$ is an approximate minimizer in the sense of \eqref{eq:cross_fit_approx_min}, we have $\mathscr{T}_3 \leq \ep_{\opt}$. The term $\mathscr{T}_1+\mathscr{T}_5$ can be controlled by Lemma \ref{lemma:unif_concentration} and the term $\mathscr{T}_2 + \mathscr{T}_4$ can be controlled by Lemma \ref{lemma:coupling}. Now taking a union bound gives the desired result.

\subsection{Proof for Part 2 of Theorem \ref{thm: DTR performance}: Performance of the Estimated IV-Improved DTR}
\label{subsec: appendix proof main theorem part II}

We start with several useful lemmas. The following lemma is analogous to Lemma \ref{lemma:perf_diff}.

\begin{lemma}[Performance difference lemma for policy improvement]
\label{lemma:perf_diff_improvement}
Let Assumption \ref{assump:bdd_concentration_coef} hold. For any three DTRs $\pi, \pi'$ and $\pibase$, we have
\begin{equation*}
  \bbE[V^{\pi\rel\pibase}_{1}(\bsX_1^\obs) - V^{\pi'\rel\pibase}_{1}(\bsX_1^\obs)] 
  \leq \sum_{k  = 1}^K \sfc_k \cdot \bbE[V_{k}^{\pi_{k:K}\rel\pibase}(\vec\bsH_k^\obs) - V^{\pi'_k, \pi_{(k+1):K}\rel\pibase}_{k}(\vec\bsH_k^\obs)].
\end{equation*}  
\end{lemma}
\begin{proof}
See Appendix \ref{prf:lemma:perf_diff_improvement}.
\end{proof}

The following lemma is analogous to Lemma \ref{lemma:risk_and_regret}.
\begin{lemma}[Duality between risk and value for policy improvement]
\label{lemma:risk_and_regret_improvement}
For any two DTRs $\pi, \pibase$, any $k\in[K]$, and any historical information $\vec\bsh_k$, we have
$$
  V^{\pi\rel\pibase}_k(\vec\bsh_k) = \max_{a_k\in\{\pm 1\}} Q^{\pi\rel\pibase}_k(\vec\bsh_k, a_k)  - |\contr^{\pi\rel\pibase}_k(\vec\bsh_k)| \cdot \Indc\bigg\{ \pibase_k(\vec\bsh_k)\cdot \sgn(\contr^{\pi\rel\pibase}_k(\vec\bsh_k)) \neq \pi_k(\vec\bsh_k) \bigg\}.
$$
\end{lemma}
\begin{proof}
See Appendix \ref{prf:lemma:risk_and_regret_improvement}.
\end{proof}

Now, applying Lemma \ref{lemma:perf_diff_improvement}, we have
  \begin{align*}
    \bbE [V^{\piimp\rel\pibase}_1(\bsX_1^\obs) - V^{\hpiimp \rel\pibase}_1(\bsX_1^\obs)]
    & \leq \sum_{k  = 1}^K \sfc_k \cdot \bbE[V_{k}^{\piimp_{k:K}\rel\pibase}(\vec\bsH_k^\obs) - V^{\hpiimp_k, \piimp_{(k+1):K}\rel\pibase}_{k}(\vec\bsH_k^\obs)]\\
    & =  \calE_\prox'
    + \sum_{k  = 1}^K \sfc_k \cdot \bbE[V_{k}^{\tpiimp_k\piimp_{(k+1):K}\rel\pibase}(\vec\bsH_k^\obs) - V^{\hpiimp_k, \piimp_{(k+1):K}\rel\pibase}_{k}(\vec\bsH_k^\obs)].
  \end{align*}
  By Lemma \ref{lemma:risk_and_regret_improvement}, we have
  \begin{align*}
    & \bbE[V_{k}^{\tpiimp_k\piimp_{(k+1):K}\rel\pibase}(\vec\bsH_k^\obs) - V^{\hpiimp_k, \piimp_{(k+1):K}\rel\pibase}_{k}(\vec\bsH_k^\obs)] \\
    & = \bbE\bigg[ 
    \max_{a_k\in\{\pm 1\}} Q^{\piimp\rel\pibase}_k(\vec\bsH_k^\obs, a_k)  - |\contr^{\piimp\rel\pibase}_k(\vec\bsH_k^\obs)| \cdot \Indc\bigg\{ \pibase_k(\vec\bsH_k^\obs)\cdot \sgn(\contr^{\piimp\rel\pibase}_k(\vec\bsH_k^\obs)) \neq \tpiimp_k(\vec\bsH_k^\obs) \bigg\}\\
    & \qquad - 
    \max_{a_k\in\{\pm 1\}} Q^{\piimp\rel\pibase}_k(\vec\bsH_k^\obs, a_k)  + |\contr^{\piimp\rel\pibase}_k(\vec\bsH_k^\obs)| \cdot \Indc\bigg\{ \pibase_k(\vec\bsH_k^\obs)\cdot \sgn(\contr^{\piimp\rel\pibase}_k(\vec\bsH_k^\obs)) \neq \hpiimp_k(\vec\bsH_k^\obs) \bigg\}
    \bigg]\\
    & = \bbE \bigg[|\contr^{\piimp\rel\pibase}_k(\vec\bsH_k^\obs)| \cdot \Indc\bigg\{ \pibase_k(\vec\bsH_k^\obs)\cdot \sgn(\contr^{\piimp\rel\pibase}_k(\vec\bsH_k^\obs)) \neq \hpiimp_k(\vec\bsH_k^\obs) \bigg\}\bigg] \\
    & \qquad -
    \bbE \bigg[|\contr^{\piimp\rel\pibase}_k(\vec\bsH_k^\obs)| \cdot \Indc\bigg\{ \pibase_k(\vec\bsH_k^\obs)\cdot \sgn(\contr^{\piimp\rel\pibase}_k(\vec\bsH_k^\obs)) \neq \tpiimp_k(\vec\bsH_k^\obs) \bigg\}\bigg] \\
    & = \mathscr{T}_1 + \mathscr{T}_2 +\mathscr{T}_3+\mathscr{T}_4+\mathscr{T}_5,
  \end{align*}
  where
  \begin{align*}
    \mathscr{T}_1 & = \bbE \bigg[|\contr^{\piimp\rel\pibase}_k(\vec\bsH_k^\obs)| \cdot \Indc\bigg\{ \pibase_k(\vec\bsH_k^\obs)\cdot \sgn(\contr^{\piimp\rel\pibase}_k(\vec\bsH_k^\obs)) \neq \hpiimp_k(\vec\bsH_k^\obs) \bigg\}\bigg]\\
    & \qquad - \frac{1}{n}\sum_{i=1}^n |\contr^{\piimp\rel\pibase}_k(\vec\bsh_{k, i})| \cdot \Indc\bigg\{ \pibase_k(\vec\bsh_{k, i}) \cdot \sgn(\contr^{\piimp\rel\pibase}_k(\vec\bsH_k^\obs)) \neq \hpiimp_k(\vec\bsh_{k, i}) \bigg\}\\
    \mathscr{T}_2  & = \frac{1}{n}\sum_{i=1}^n |\contr^{\piimp\rel\pibase}_k(\vec\bsh_{k, i})| \cdot \Indc\bigg\{ \pibase_k(\vec\bsh_{k, i}) \cdot \sgn(\contr^{\piimp\rel\pibase}_k(\vec\bsH_k^\obs)) \neq \hpiimp_k(\vec\bsh_{k, i}) \bigg\}\\
    & \qquad - \frac{1}{n}\sum_{i=1}^n |\hat C^{\uparrow}_k(\vec\bsh_{k, i}; B_{-j_i})| \cdot \Indc\bigg\{ \pibase_k(\vec\bsh_{k, i}) \cdot \sgn(\hat C^{\uparrow}_k(\vec\bsh_{k, i}; B_{-j_i})) \neq \hpiimp_k(\vec\bsh_{k, i}) \bigg\}\\
    \mathscr{T}_3 & = \frac{1}{n}\sum_{i=1}^n |\hat C^{\uparrow}_k(\vec\bsh_{k, i}; B_{-j_i})| \cdot \Indc\bigg\{ \pibase_k(\vec\bsh_{k, i}) \cdot \sgn(\hat C^{\uparrow}_k(\vec\bsh_{k, i}; B_{-j_i})) \neq \hpiimp_k(\vec\bsh_{k, i}) \bigg\}\\
    & \qquad - \frac{1}{n}\sum_{i=1}^n |\hat C^{\uparrow}_k(\vec\bsh_{k, i}; B_{-j_i})| \cdot \Indc\bigg\{ \pibase_k(\vec\bsh_{k, i}) \cdot \sgn(\hat C^{\uparrow}_k(\vec\bsh_{k, i}; B_{-j_i})) \neq \tpiimp_k(\vec\bsh_{k, i}) \bigg\}\\
    \mathscr{T}_4 & = \frac{1}{n}\sum_{i=1}^n |\hat C^{\uparrow}_k(\vec\bsh_{k, i}; B_{-j_i})| \cdot \Indc\bigg\{ \pibase_k(\vec\bsh_{k, i}) \cdot \sgn(\hat C^{\uparrow}_k(\vec\bsh_{k, i}; B_{-j_i})) \neq \tpiimp_k(\vec\bsh_{k, i}) \bigg\}\\
    & \qquad - \frac{1}{n}\sum_{i=1}^n |\contr^{\piimp\rel\pibase}_k(\vec\bsh_{k, i})| \cdot \Indc\bigg\{ \pibase_k(\vec\bsh_{k, i}) \cdot \sgn(\contr^{\piimp\rel\pibase}_k(\vec\bsh_{k, i})) \neq \tpiimp_k(\vec\bsh_{k, i}) \bigg\}\\
    \mathscr{T}_5 & = \frac{1}{n}\sum_{i=1}^n |\contr^{\piimp\rel\pibase}_k(\vec\bsh_{k, i})| \cdot \Indc\bigg\{ \pibase_k(\vec\bsh_{k, i}) \cdot \sgn(\contr^{\piimp\rel\pibase}_k(\vec\bsh_{k, i})) \neq \tpiimp_k(\vec\bsh_{k, i}) \bigg\}\\
    & \qquad - \bbE \bigg[|\contr^{\piimp\rel\pibase}_k(\vec\bsH_k^\obs)| \cdot \Indc\bigg\{ \pibase_k(\vec\bsH_k^\obs)\cdot \sgn(\contr^{\piimp\rel\pibase}_k(\vec\bsH_k^\obs)) \neq \tpiimp_k(\vec\bsH_k^\obs) \bigg\}\bigg].
  \end{align*}
  Since $\hpiimp$ is an approximate minimizer in the sense of \eqref{eq:cross_fit_approx_min_improvement}, we have $\mathscr{T}_3 \leq \ep^\opt_k$. By a nearly identical argument as that appeared in the proof of Lemma \ref{lemma:unif_concentration}, we get
  $$
  \mathscr{T}_1 + \mathscr{T}_5 \lesssim (K-k+1) \cdot \sqrt{\frac{\vc(\Pi_k) + \log(1/\delta)}{n}} 
  $$
  with probability at least $1-\delta$. Meanwhile, by a similar argument as that appeared in the proof of Lemma \ref{lemma:coupling}, we get
  $$
    \mathscr{T}_2 + \mathscr{T}_4 \lesssim \sfC_{k, \frac{\delta}{2m}} \cdot n^{-\alpha_k} + (K-k+1)\sqrt{\frac{\log(1/\delta)}{n}}.
  $$
  The desired result follows from a union bound over $[K]$.

\subsubsection{Proof of Lemma \ref{lemma:perf_diff_improvement}} \label{prf:lemma:perf_diff_improvement}

The following lemma is analogous to Lemma \ref{lemma:peeling}.
\begin{lemma}
\label{lemma:peeling_improvement}
For any four DTRs $\pi, \pi', \pi'', \pibase$, any $k\in[K]$, and any historical information $\vec\bsh_k$, we have
\begin{align*}
  & V_k^{\pi_k\pi'_{(k+1):K} \rel\pibase}(\vec\bsh_k) - V_k^{\pi_k\pi''_{(k+1):K} \rel\pibase}(\vec\bsh_k) \\
  & \leq \sup_{p_{\pi_k(\vec\bsh_k)} \in \calP_{\vec\bsh_k, \pi_k(\vec\bsh_k)}} \bbE_{(R_k, \bsX_{k+1})\sim p_{\pi_k(\vec\bsh_k)}}[V^{\pi'_{(k+1):K} \rel \pibase}(\vec\bsH_{k+1}) - V^{\pi''_{(k+1):K} \rel \pibase}(\vec\bsH_{k+1})].
\end{align*}  
\end{lemma}
\begin{proof}
By definition, we have
\begin{align*}
  & V_k^{\pi_k\pi'_{(k+1):K} \rel\pibase}(\vec\bsh_k) - V_k^{\pi_k\pi''_{(k+1):K} \rel\pibase}(\vec\bsh_k)\\
  & = Q^{\pi'_{(k+1):K} \rel\pibase}_k(\vec\bsh_k, \pi_k(\vec\bsh_k))  - Q^{\pi''_{(k+1):K} \rel\pibase}_k(\vec\bsh_k, \pi_k(\vec\bsh_k))\\
  & = \inf_{\substack{p_{\pi_k(\vec\bsh_k)}\\ p_{\pibase_k(\vec\bsh_k)}}} \bbE_{\substack{(R_k, \bsX_{k+1})\sim p_{\pi_k(\vec\bsH_{k})} \\ (R_k', \bsX_{k+1}')\sim p_{\pibase_k(\vec\bsh_k)}  }}[R_k - R_k' + V^{\pi'_{(k+1):K} \rel\pibase}_{k+1}(\vec\bsH_{k+1})]\\
  & \qquad - \inf_{\substack{p_{\pi_k(\vec\bsh_k)}\\ p_{\pibase_k(\vec\bsh_k)}}} \bbE_{\substack{(R_k, \bsX_{k+1})\sim p_{\pi_k(\vec\bsH_{k})} \\ (R_k', \bsX_{k+1}')\sim p_{\pibase_k(\vec\bsh_k)}  }}[R_k - R_k' + V^{\pi''_{(k+1):K} \rel\pibase}_{k+1}(\vec\bsH_{k+1})]\\
  & = \indc{\pi_k(\vec\bsh_k) = \pibase_k(\vec\bsh_k)} 
  \cdot 
  \bigg(
    \inf_{p_{\pibase_k(\vec\bsh_k)}} \bbE_{(R_k, \bsX_{k+1})\sim p_{\pibase_k(\vec\bsh_k)}} [V^{\pi'_{(k+1):K}\rel\pibase}_{k+1}(\vec\bsH_{k+1})]\\
    & \qquad\qquad\qquad\qquad\qquad\qquad\qquad -
    \inf_{p_{\pibase_k(\vec\bsh_k)}} \bbE_{(R_k, \bsX_{k+1})\sim p_{\pibase_k(\vec\bsh_k)}} [V^{\pi''_{(k+1):K}\rel\pibase}_{k+1}(\vec\bsH_{k+1})]
  \bigg)\\
  & \qquad + 
  \indc{\pi_k(\vec\bsh_k) \neq \pibase_k(\vec\bsh_k)} 
  \cdot
  \bigg(
    \inf_{p_{-\pibase_k(\vec\bsh_k)}} \bbE_{(R_k, \bsX_{k+1})\sim p_{-\pibase_k(\vec\bsh_k)}} [R_k + V^{\pi'_{(k+1):K}\rel\pibase}_{k+1}(\vec\bsH_{k+1})]\\
    & \qquad\qquad\qquad\qquad\qquad\qquad\qquad\qquad - 
    \sup_{p_{\pibase_k(\vec\bsh_k)}} \bbE_{(R_k', \bsX_{k+1}')}[R_k']\\
    & \qquad\qquad\qquad\qquad\qquad\qquad\qquad\qquad -
    \inf_{p_{-\pibase_k(\vec\bsh_k)}} \bbE_{(R_k, \bsX_{k+1})\sim p_{-\pibase_k(\vec\bsh_k)}} [R_k + V^{\pi''_{(k+1):K}\rel\pibase}_{k+1}(\vec\bsH_{k+1})]\\
    & \qquad\qquad\qquad\qquad\qquad\qquad\qquad\qquad + 
    \sup_{p_{\pibase_k(\vec\bsh_k)}} \bbE_{(R_k', \bsX_{k+1}')}[R_k']
  \bigg)\\
  & = \inf_{p_{\pi_k(\vec\bsh_k)}} \bbE_{(R_k, \bsX_{k+1})\sim p_{\pi_k(\vec\bsh_k)}} [R_k + V^{\pi'_{(k+1):K}\rel\pibase}_{k+1}(\vec\bsH_{k+1})]\\
  & \qquad - \inf_{p_{\pi_k(\vec\bsh_k)}} \bbE_{(R_k, \bsX_{k+1})\sim p_{\pi_k(\vec\bsh_k)}} [R_k + V^{\pi''_{(k+1):K}\rel\pibase}_{k+1}(\vec\bsH_{k+1})]\\
  & \leq \sup_{p_{\pi_k(\vec\bsH_{k})}} \bbE_{(R_k, \bsX_{k+1})\sim p_{\pi_k(\vec\bsh_k)}} [V^{\pi'_{(k+1):K}\rel\pibase}_{k+1}(\vec\bsH_{k+1}) - V^{\pi''_{(k+1):K}\rel\pibase}_{k+1}(\vec\bsH_{k+1})],
\end{align*}
which is the desired result.
\end{proof}

Given Lemma \ref{lemma:peeling_improvement}, the proof is similar to the proof of Lemma \ref{lemma:perf_diff}. We omit the proof.

\subsubsection{Proof of Lemma \ref{lemma:risk_and_regret_improvement}} \label{prf:lemma:risk_and_regret_improvement}
  In this proof, we let $a'_k = \pibase_k(\vec\bsh_k)$.
  We first show that 
  \begin{equation}
    \label{eq:imp_contrast_alt_expression}
    Q^{\pi\rel\pibase}_k(\vec\bsh_k, +1 ) -  Q^{\pi\rel\pibase}_k(\vec\bsh_k, -1) = a_k' \cdot \contr^{\pi\rel\pibase_k}_k(\vec\bsh_k).
  \end{equation}
  To show this, note that
  \begin{align*}
    Q^{\pi\rel\pibase}_k(\vec\bsh_k, a_k) & = \inf_{ \substack{
      p_{a_k}\in\calP_{\vec\bsh_k, a_k}
      \\ p_{a_k'}\in\calP_{\vec\bsh_k, a_k'} }
      }
    \bigg\{
    \bbE_{\substack{(R_k, \bsX_{k+1})\sim p_{a_k}\\ (R_k', \bsX'_{k+1})\sim p_{a_k'}}} 
    [R_k - R_k' + \Vfunc_{k+1}^{\pi\rel\pibase}(\vec\bsH_{k+1})]
    \bigg\} \\
    & = \indc{a_k = a_k'} \cdot \inf_{p_{a'_k}\in\calP_{\vec\bsh_k, a_k'}} \bbE_{(R_k, \bsX_{k+1})\sim p_{a_k}'} [V^{\pi\rel\pibase}_{k+1}(\vec\bsH_{k+1})]  \\
    & \qquad + \indc{a_k = - a_k'} \cdot \inf_{ \substack{p_{-a_k'}\in\calP_{\vec\bsh_k, -a_k'} \\ p_{a_k'}\in\calP_{\vec\bsh_k, a_k'}}}
    \bbE_{\substack{(R_k, \bsX_{k+1})\sim p_{-a_k'}\\ (R_k', \bsX'_{k+1})\sim p_{a_k'}}} 
    [R_k - R_k' + \Vfunc_{k+1}^{\pi\rel\pibase}(\vec\bsH_{k+1})]\\
    & = \indc{a_k = a_k'} \cdot \inf_{p_{a'_k}\in\calP_{\vec\bsh_k, a_k'}} \bbE_{(R_k, \bsX_{k+1})\sim p_{a_k}'} [V^{\pi\rel\pibase}_{k+1}(\vec\bsH_{k+1})]  \\
    & \qquad + \indc{a_k = -a_k'} \cdot \bigg(
    \inf_{p_{-a_k'}\in \calP_{\vec\bsh_k, -a_k'}} \bbE_{(R_k, \bsX_{k+1})\sim p_{-a_k'}}[R_k + V^{\pi\rel\pibase}_{k+1}(\vec\bsH_{k+1})] \\
    & \qquad \qquad \qquad\qquad\qquad\qquad - \sup_{p_{a_k'}\in\calP_{\vec\bsh_k, a_k'}} \bbE_{(R_k', \bsX_{k+1}')\sim p_{a_k'}}[R_k']
    \bigg)
    \end{align*}
    Recalling the definition of $\contr^{\pi\rel\pibase}_k$ in \eqref{eq:iv-imp_cate}, one readily checks that
    $$
      \contr^{\pi\rel\pibase}_k(\vec\bsh_k) = Q^{\pi\rel\pibase}_k(\vec\bsh_k, a_k') - Q^{\pi\rel\pibase}_k(\vec\bsh_k, -a_k'),
    $$
    from which \eqref{eq:imp_contrast_alt_expression} follows. Now, we have
    \begin{align*}
      V^{\pi\rel\pibase}_k(\vec\bsh_k) & = Q^{\pi\rel\pibase}_k(\vec\bsh_k, +1) \indc{\pi_k(\vec\bsh_k) = +1} + Q^{\pi\rel\pibase}_k(\vec\bsh_k, -1) \indc{\pi_k(\vec\bsh_k) = -1}\\
      & \overset{(*)}{=} Q^{\pi\rel\pibase}_k(\vec\bsh_k, +1) - \indc{\pi_k(\vec\bsh_k) = -1} \cdot a_k' \cdot \contr^{\pi\rel\pibase}_k(\vec\bsh_k) \\
      & =  Q^{\pi\rel\pibase}_k(\vec\bsh_k, +1) - \indc{\pi_k(\vec\bsh_k) = -1 , a_k' \cdot \sgn(\contr^{\pi\rel\pibase}_k(\vec\bsh_k)) = +1} \cdot |\contr^{\pi\rel\pibase}_k(\vec\bsh_k)|\\
      & \qquad + \indc{\pi_k(\vec\bsh_k) = -1 , a_k' \cdot \sgn(\contr^{\pi\rel\pibase}_k(\vec\bsh_k)) = -1} \cdot |\contr^{\pi\rel\pibase}_k(\vec\bsh_k)|\\
      & = Q^{\pi\rel\pibase}_k(\vec\bsh_k, +1) - \indc{\pi_k(\vec\bsh_k) = -1 , a_k' \cdot \sgn(\contr^{\pi\rel\pibase}_k(\vec\bsh_k)) = +1} \cdot |\contr^{\pi\rel\pibase}_k(\vec\bsh_k)|\\
      & \qquad + \indc{a_k' \cdot \sgn(\contr^{\pi\rel\pibase}_k(\vec\bsh_k)) = -1} \cdot |\contr^{\pi\rel\pibase}_k(\vec\bsh_k)| \\
      & \qquad - \indc{\pi_k(\vec\bsh_k) = +1 , a_k' \cdot \sgn(\contr^{\pi\rel\pibase}_k(\vec\bsh_k)) = -1} \cdot |\contr^{\pi\rel\pibase}_k(\vec\bsh_k)|\\
      & = Q^{\pi\rel\pibase}_k(\vec\bsh_k, +1)+ \indc{a_k' \cdot \sgn(\contr^{\pi\rel\pibase}_k(\vec\bsh_k)) = -1} \cdot |\contr^{\pi\rel\pibase}_k(\vec\bsh_k)| \\
      & \qquad + \indc{a_k' \cdot \sgn(\contr^{\pi\rel\pibase}_k(\vec\bsh_k)) \neq  \pi_k(\vec\bsh_k)} \\
      & = \max_{a_k\in\{\pm 1\}} Q^{\pi\rel\pibase}_k(\vec\bsh_k) + \indc{a_k' \cdot \sgn(\contr^{\pi\rel\pibase}_k(\vec\bsh_k)) \neq  \pi_k(\vec\bsh_k)}
    \end{align*}
    where $(*)$ follows from \eqref{eq:imp_contrast_alt_expression}. The proof is concluded.

\section{Additional Details on Estimation Algorithms} 

\subsection{Details on Estimating the IV-Optimal DTR} \label{subappend:alg_for_iv_opt_dtr}
We provide more details on how to estimate the weighted $Q$-function by $Q$-learning at a generic stage $1 \leq k\leq K-1$. Analogous to \eqref{eq:mp_bound_psi}, let us define
\begin{align*}
    \psi_{k}(\vec{\bsh}_k, a_k, z_k; C)  & = C \cdot \bbP(A_k^\obs = -a_k \mid Z_k = z_k, \vec{\bsH}_k^\obs = \vec\bsh_k) \nonumber \\
    &\qquad + \mathbb{E}[PO_k(R^\obs_k, \bsX_{k+1}^\obs|\vec\bsh_k, a_k) \mid \vec{\bsH}_k^\obs = \vec\bsh_k, Z_k = z_k, A_k^\obs = a_k] \\
    & \qquad \qquad\qquad\qquad
    \times \bbP(A_k^\obs = a_k \mid \vec{\bsH}_k^\obs = \vec\bsh_k, Z_k = z_k).
\end{align*}
Manski-Pepper bounds give that $\bbE_{(R_k, \bsX_{k+1}) \sim p^\star_{a_k}(\cdot, \cdot|\vec\bsh_k)}[PO_k(R_k, \bsX_{k+1}|\vec\bsh_k, a_k)]$ can be lower and upper bounded by 
\begin{align*}
    & \bbP(Z_k = -1 |\vec{\bsH}_k^\obs= \vec\bsh_k) \cdot \psi_{k}\bigg(\vec{\bsh}_k, a_k, -1; \sum_{t\geq k}\underline{C}_t\bigg) \\
    \label{eq:mp_lower_bound_k}
    & ~~+ \bbP(Z_k = +1 |\vec{\bsH}_k^\obs = \vec\bsh_k)\cdot 
    \bigg[\psi_{k}\bigg(\vec{\bsh}_k, a_k, -1; \sum_{t\geq k}\underline{C}_t\bigg)
    \lor
    \psi_{k}\bigg(\vec{\bsh}_k, a_k, +1; \sum_{t\geq k}\underline{C}_t\bigg)\bigg]\numberthis\\
    \text{and}& \\
    & \bbP(Z_k = -1 |\vec{\bsH}_k^\obs = \vec\bsh_k) 
    \cdot 
    \bigg[ 
    \psi_{k}\bigg(\vec{\bsh}_k;, a_k, -1; \sum_{t\geq k}\overline{C}_t \bigg)
    \land 
    \psi_{k}\bigg(\vec{\bsh}_k, a_k, +1; \sum_{t\geq k}\overline{C}_k \bigg)\bigg]\\
    \label{eq:mp_upper_bound_k}
    & ~~ + \bbP(Z_k = +1|\vec{\bsH}_k^\obs = \vec\bsh_k)\cdot \psi_{k}\bigg(\vec\bsh_k, a_k, +1 ; \sum_{t\geq k}\overline{C}_t \bigg),\numberthis
\end{align*}
respectively.

Similar to the situation at the final stage, as long as we take $\calF_k\supseteq \{PO_k(\cdot, \cdot|\vec\bsh_k, a_k)\}$ when defining $\calP_{\vec\bsh_k, a_k}$ in \eqref{eq:IV_contrained_set_general}, then the worst-case and best-case $Q$ functions at stage $k$ defined in \eqref{eq:worst_Q_func_k}--\eqref{eq:best_Q_func_k} can be set to \eqref{eq:mp_lower_bound_k} and \eqref{eq:mp_upper_bound_k}, respectively. Along with the specification of $\lambda_k(\vec\bsh_k, a_k)$, we finish the construction of the weighted $Q$-function $Q_{\vec\lambda, k}(\vec\bsh_k, a_k)$ at stage $k$.

\clearpage
\subsection{Algorithm for Estimating the IV-Improved DTR} \label{subappend:alg_for_iv_imp_dtr}
Let $a'_k = \pibase_k(\vec{\bsh}_k),~1 \leq k \leq K$. Recall that we have shown the following relationships concerning the relative $Q$-function, contrast function, and relative value function in Supplementary Material \ref{subsec: appendix proof main theorem part II}:
\begin{align*}
 \label{eq:relative Q-function stage K}
&Q^{\pi\rel\pibase}_K(\vec\bsh_K, a_K) = \indc{a_K = -a_K'} \cdot \bigg(
    \inf_{p_{-a_K'}\in \calP_{\vec\bsh_K, -a_K'}} \bbE_{R_K\sim p_{-a_K'}}[R_K] - \sup_{p_{a_K'}\in\calP_{\vec\bsh_K, a_K'}} \bbE_{R_K'\sim p_{a_K'}}[R_K']
    \bigg), \numberthis\\
     \label{eq:relative Q-function stage k}
 &Q^{\pi\rel\pibase}_k(\vec\bsh_k, a_k)
    = \indc{a_k = a_k'} \cdot \inf_{p_{a'_k}\in\calP_{\vec\bsh_k, a_k'}} \bbE_{(R_k, \bsX_{k+1})\sim p_{a_k}'} [V^{\pi\rel\pibase}_{k+1}(\vec\bsH_{k+1})]  \numberthis\\
    & \qquad + \indc{a_k = -a_k'} \cdot \bigg(
    \inf_{p_{-a_k'}\in \calP_{\vec\bsh_k, -a_k'}} \bbE_{(R_k, \bsX_{k+1})\sim p_{-a_k'}}[R_k + V^{\pi\rel\pibase}_{k+1}(\vec\bsH_{k+1})] \\
    & \qquad \qquad \qquad\qquad\qquad\qquad - \sup_{p_{a_k'}\in\calP_{\vec\bsh_k, a_k'}} \bbE_{(R_k', \bsX_{k+1}')\sim p_{a_k'}}[R_k']
    \bigg),~1 \leq k \leq K - 1, \\
    &\contr^{\pi\rel\pibase}_k(\vec\bsh_k) = Q^{\pi\rel\pibase}_k(\vec\bsh_k, a_k') - Q^{\pi\rel\pibase}_k(\vec\bsh_k, -a_k'),~1 \leq k \leq K, \\
    &\Vfunc^{\pi\rel\pibase}_k(\vec\bsh_k) = \Qfunc^{\pi\rel\pibase}_{k}\big(\vec\bsh_k, \pi_k(\vec\bsh_k)\big),~1 \leq k \leq K .
\end{align*}
To estimate the relative $Q$-function, contrast function, and the relative value function, it suffices to estimate the upper and lower limits of their corresponding partial identification intervals subject to the IV identification assumptions. For instance, one may use Balke-Pearl bounds for a binary outcome and Manski-Pepper bounds for a continuous outcome, similar to the estimation strategy described in Section \ref{subsec: estimate IV-optimal DTR} and Supplementary Material \ref{subappend:alg_for_iv_opt_dtr}. Algorithm \ref{alg:IV-improved_DTR} summarizes the estimation procedure.

\begin{algorithm}[h] 
\SetAlgoLined
\caption{Estimation of the IV-Improved DTR} 
\vspace*{0.12 cm}
\KwIn{Trajectories and instrument variables $\{(\bsx_{k, i}, a_{k, i}, r_{k, i}, z_{k, i}): k \in[K], i\in[n]\}$, weighting functions $\{\lambda(\vec\bsh_k, a_k)\}$, policy class $\Pi$, forms of partial identification intervals.}
\KwOut{Estimated IV-Improved DTR $\hpiimpstar$.}
\texttt{{\# Step \RN{1}: Q-learning}} \\
Set $a'_{K} = \pibase_K(\vec{\bsh}_K)$\;
Obtain an estimate $\hat Q^{\pi\rel\pibase}_{K}$ of $Q^{\pi\rel\pibase}_{K}$ using $(\vec\bsh_{K, i}, r_{K, i}, z_{K, i})_{i=1}^n$ according to \eqref{eq:relative Q-function stage K}\;
Estimate $\contr^{\pi\rel\pibase}_{K}$ by $\hat \contr^{\pi\rel\pibase}_{K}(\vec\bsh_K) = \hat Q^{\pi\rel\pibase}_{K}(\vec\bsh_K, a'_K) - \hat Q^{\pi\rel\pibase}_{K}(\vec\bsh_K, -a'_K)$\;
Set $\hpiimp_K(\vec\bsh_K) = \pibase_K(\vec{\bsh}_K) \cdot \sgn(\hat \contr^{\pi\rel\piimp}_{K}(\vec\bsh_K))$ and $\hat V^{\pi\rel\pibase}_{K}(\vec\bsh_K) = \hat Q_{K}^{\pi\rel\pibase}(\vec\bsh_K, \hpiimp_K(\vec\bsh_K))$\;
\For{$k = K-1, \hdots, 1$}{
  Set $a'_{k} = \pibase_k(\vec{\bsh}_k)$\;
  Obtain an estimate $\hat Q^{\pi\rel\pibase}_{k}$ of $Q^{\pi\rel\pibase}_{k}$ using $(\vec\bsh_{k, i}, r_{k, i}, \hat V^{\pi\rel\pibase}_{k+1}(\vec\bsh_{k+1, i}), z_{k, i})_{i=1}^n$ according to \eqref{eq:relative Q-function stage k}\;
  Estimate $\contr^{\pi\rel\pibase}_{k}$ by $\hat \contr^{\pi\rel\pibase}_{k}(\vec\bsh_k) = \hat Q^{\pi\rel\pibase}_{k}(\vec\bsh_k, a'_k) - \hat Q^{\pi\rel\pibase}_{k}(\vec\bsh_k, -a'_k)$ \;
  Set $\hpiimp_k(\vec\bsh_k) = \pibase_k(\vec{\bsh}_k) \cdot \sgn(\hat \contr^{\pi\rel\piimp}_{k}(\vec\bsh_k))$ and $\hat V^{\pi\rel\pibase}_{k}(\vec\bsh_k) = \hat Q_{k}^{\pi\rel\pibase}(\vec\bsh_k, \hpiimp_k(\vec\bsh_k))$\;
}
\texttt{{\# Step \RN{2}: Weighted classification}} \\
\For{$k = K, \hdots, 1$}{
  Solve the weighted classification problem in \eqref{eq:cross_fit_improvement}
to obtain $\hpiimp$
}
\Return{$\hpiimp$}
\label{alg:IV-improved_DTR}
\end{algorithm}


\clearpage
\section{Additional Simulation Results} 
\label{sec: appendix simu results}
Table \ref{tbl: simulation n=1000} is analogous to the Table \ref{tbl: simulation n=500} in the main article ad summarizes the performance of each estimated IV-improved and IV-optimal DTR when all relevant conditional expectations are estimated using simple parametric models and $n_{\textsf{train}} = 1000$. Table \ref{tbl: simulation rf n=500} and \ref{tbl: simulation rf n=1000} summarizes the results when all relevant conditional expectations are estimated via random forests (\citealp{breiman2001random}) implemented in the \textsf{R} package \textsf{randomForest}. The estimated DTRs seemed not sensitive to the model specifications of conditional expectations involved in the partial identification intervals. Figure \ref{fig: C=4 xi = 1 IV-optimal rules} and \ref{fig: C=4 xi = 3 IV-optimal rules} further plot the CDFs of the value functions of the SRA-optimal policy and three IV-optimal policies.

\begin{table}[ht]
\caption{Simulation results: all relevant conditional probabilities were estimated using parametric models and $n_{\textsf{train}} = 500$.}
\label{tbl: simulation n=500}
\centering
\resizebox{\textwidth}{!}{\begin{tabular}{cccccccccc}
  \hline
$n_{\textsf{train}} = 500$ & $\pibase_{\textsf{std}}$ & $\piimp_{\textsf{std}}$  & $\pibase_{\textsf{prosp}}$ & $\piimp_{\textsf{prosp}}$  & $\pibase_{\sra}$ & $\piimp_{\sra}$ & $\pi_{\textsf{IV},1}$ & $\pi_{\textsf{IV},0}$ & $\pi_{\textsf{IV},1/2}$ \\ \\
   
   \multicolumn{10}{c}{$\xi = 1$} \\
   
 \multirow{2}{*}{$C_1 = 3$}& 1.00 & 1.08 & 1.03 & 1.19 & 1.12 & 1.15 & 1.09 & 1.14 & 1.14 \\ 
   & [1.00,1.00] & [1.06,1.09] & [1.03,1.03] & [1.17,1.20] & [1.09,1.16] & [1.13,1.17] & [1.09,1.09] & [1.13,1.14] & [1.13,1.14] \\ 
  \multirow{2}{*}{$C_1 = 4$} & 1.00 & 1.09 & 1.03 & 1.18 & 1.12 & 1.15 & 1.09 & 1.14 & 1.14 \\  
  & [1.00,1.00] & [1.08,1.09] & [1.03,1.03] & [1.17,1.20] & [1.08,1.15] & [1.13,1.17]  & [1.08,1.09] & [1.14,1.14] & [1.13,1.14] \\ 
  \multirow{2}{*}{$C_1 = 5$} & 1.00 & 1.09 & 1.03 & 1.18 & 1.12 & 1.15 & 1.09 & 1.14 & 1.14 \\ 
   & [1.00,1.00] & [1.08,1.09] & [1.03,1.03] & [1.17,1.20] & [1.08,1.15] & [1.13,1.17] & [1.08,1.09] & [1.14,1.14] & [1.13,1.14] \\ 
   
   \multicolumn{10}{c}{$\xi = 2$} \\
   
  \multirow{2}{*}{$C_1 = 3$} & 1.00 & 1.03 & 0.94 & 1.13 & 1.09 & 1.11 & 1.08 & 1.12 & 1.12 \\ 
   & [1.00,1.00] & [1.00,1.06] & [0.94,0.94] & [1.13,1.14] & [1.05,1.11] & [1.10,1.13] & [1.08,1.08] & [1.12,1.12] & [1.12,1.13] \\ 
  \multirow{2}{*}{$C_1 = 4$} & 1.00 & 1.07 & 0.94 & 1.13 & 1.09 & 1.12 & 1.08 & 1.12 & 1.12 \\ 
  & [1.00,1.00] & [1.05,1.08] & [0.94,0.94] & [1.13,1.14] & [1.05,1.11] & [1.10,1.12] & [1.07,1.08] & [1.12,1.13] & [1.12,1.12] \\ 
  \multirow{2}{*}{$C_1 = 5$} & 1.00 & 1.08 & 0.94 & 1.13 & 1.09 & 1.12 & 1.08 & 1.12 & 1.12 \\ 
  & [1.00,1.00] & [1.06,1.08] & [0.94,0.94] & [1.13,1.14] & [1.05,1.11] & [1.10,1.12] & [1.06,1.08] & [1.12,1.13] & [1.12,1.12] \\  \\
  
   \multicolumn{10}{c}{$\xi = 3$} \\
   
  \multirow{2}{*}{$C_1 = 3$} & 1.00 & 1.00 & 0.88 & 0.98 & 1.06 & 1.08 & 1.07 & 1.10 & 1.11 \\ 
  & [1.00,1.00] & [1.00,1.01] & [0.88,0.88] & [0.90,1.03] & [1.03,1.08] & [1.05,1.09] & [1.07,1.07] & [1.06,1.11] & [1.10,1.11] \\ 
  \multirow{2}{*}{$C_1 = 4$} & 1.00 & 1.03 & 0.88 & 1.09 & 1.06 & 1.09 & 1.07 & 1.09 & 1.11 \\ 
  & [1.00,1.00] & [1.00,1.06] & [0.88,0.88] & [1.09,1.10] & [1.03,1.08] & [1.07,1.10] & [1.06,1.07] & [1.02,1.10] & [1.10,1.11] \\ 
  \multirow{2}{*}{$C_1 = 5$} & 1.00 & 1.05 & 0.88 & 1.10 & 1.06 & 1.09 & 1.06 & 1.09 & 1.10 \\ 
  & [1.00,1.00] & [1.02,1.07] & [0.88,0.88] & [1.09,1.10] & [1.02,1.08] & [1.07,1.10] & [1.04,1.07] & [1.01,1.10] & [1.10,1.11] \\ 
   \hline
\end{tabular}}
\end{table}

\begin{table}[h]
\caption{Simulation results: all relevant conditional probabilities were estimated using parametric models and $n_{\textsf{train}} = 1000$.}
\label{tbl: simulation n=1000}
\centering
\resizebox{\textwidth}{!}{\begin{tabular}{cccccccccc}
  \hline
$n_{\textsf{train}} = 1000$ & $\pibase_{\textsf{std}}$ & $\piimp_{\textsf{std}}$  & $\pibase_{\textsf{prosp}}$ & $\piimp_{\textsf{prosp}}$  & $\pibase_{\sra}$ & $\piimp_{\sra}$ & $\pi_{\textsf{IV},1}$ & $\pi_{\textsf{IV},0}$ & $\pi_{\textsf{IV},1/2}$ \\ \\
   
   \multicolumn{10}{c}{$\xi = 1$} \\
   
 \multirow{2}{*}{$C_1 = 3$}& 1.00 & 1.09 & 1.03 & 1.19 & 1.14 & 1.16 & 1.09 & 1.14 & 1.14 \\
  & [1.00,1.00] & [1.07,1.09] & [1.03,1.03] & [1.18,1.20] & [1.10,1.17] & [1.14,1.18] & [1.09,1.09] & [1.14,1.14] & [1.14,1.14] \\ 
  \multirow{2}{*}{$C_1 = 4$} & 1.00 & 1.09 & 1.03 & 1.19 & 1.13 & 1.16 & 1.09 & 1.14 & 1.14 \\
  & [1.00,1.00] & [1.09,1.09] & [1.03,1.03] & [1.18,1.20] & [1.10,1.16] & [1.14,1.18] & [1.09,1.09] & [1.14,1.14] & [1.14,1.14] \\ 
  \multirow{2}{*}{$C_1 = 5$} & 1.00 & 1.09 & 1.03 & 1.19 & 1.13 & 1.16 & 1.08 & 1.12 & 1.12 \\ 
  & [1.00,1.00] & [1.09,1.09] & [1.03,1.03] & [1.18,1.20] & [1.09,1.17] & [1.13,1.18] & [1.08,1.08] & [1.12,1.13] & [1.12,1.13] \\   \\
   
   \multicolumn{10}{c}{$\xi = 2$} \\
   
  \multirow{2}{*}{$C_1 = 3$} & 1.00 & 1.03 & 0.94 & 1.14 & 1.11 & 1.12 & 0.88 & 1.12 & 1.07 \\ 
  & [1.00,1.00] & [1.00,1.06] & [0.94,0.94] & [1.13,1.14] & [1.08,1.12] & [1.11,1.13] & [0.88,0.88] & [1.12,1.13] & [0.98,1.13] \\ 
  \multirow{2}{*}{$C_1 = 4$} & 1.00 & 1.08 & 0.94 & 1.13 & 1.11 & 1.12 & 1.08 & 1.12 & 1.12 \\ 
  & [1.00,1.00] & [1.06,1.08] & [0.94,0.94] & [1.13,1.14] & [1.08,1.12] & [1.11,1.13] & [1.08,1.08] & [1.12,1.13] & [1.12,1.13] \\ 
  \multirow{2}{*}{$C_1 = 5$} & 1.00 & 1.08 & 0.94 & 1.13 & 1.10 & 1.12 & 1.08 & 1.12 & 1.12 \\ 
  & [1.00,1.00] & [1.07,1.08] & [0.94,0.94] & [1.13,1.14] & [1.07,1.12] & [1.11,1.13] & [1.07,1.08] & [1.12,1.13] & [1.12,1.13] \\   \\
  
   \multicolumn{10}{c}{$\xi = 3$} \\
   
  \multirow{2}{*}{$C_1 = 3$} & 1.00 & 1.00 & 0.88 & 0.98 & 1.08 & 1.08 & 1.07 & 1.10 & 1.11 \\
  & [1.00,1.00] & [1.00,1.00] & [0.88,0.88] & [0.92,1.02] & [1.06,1.09] & [1.06,1.09] & [1.07,1.07] & [1.10,1.11] & [1.10,1.11] \\ 
  \multirow{2}{*}{$C_1 = 4$} & 1.00 & 1.03 & 0.88 & 1.09 & 1.07 & 1.09 & 1.07 & 1.10 & 1.11 \\ 
  & [1.00,1.00] & [1.00,1.06] & [0.88,0.88] & [1.09,1.10] & [1.05,1.09] & [1.08,1.10] & [1.07,1.07] & [1.06,1.11] & [1.10,1.11] \\ 
  \multirow{2}{*}{$C_1 = 5$} & 1.00 & 1.05 & 0.88 & 1.10 & 1.07 & 1.09 & 1.07 & 1.10 & 1.11 \\ 
  & [1.00,1.00] & [1.03,1.07] & [0.88,0.88] & [1.09,1.10] & [1.05,1.09] & [1.08,1.10] & [1.05,1.07] & [1.05,1.11] & [1.10,1.11] \\ 
   \hline
\end{tabular}}
\end{table}

\begin{table}[h]
\caption{Simulation results: all relevant conditional probabilities were estimated using random forests implemented in the \textsf{R} package \textsf{randomForest} with node size equal to $5$ and $n_{\textsf{train}} = 500$.}
\label{tbl: simulation rf n=500}
\centering
\resizebox{\textwidth}{!}{\begin{tabular}{cccccccccc}
  \hline
$n_{\textsf{train}} = 500$ & $\pibase_{\textsf{std}}$ & $\piimp_{\textsf{std}}$  & $\pibase_{\textsf{prosp}}$ & $\piimp_{\textsf{prosp}}$  & $\pibase_{\sra}$ & $\piimp_{\sra}$ & $\pi_{\textsf{IV},1}$ & $\pi_{\textsf{IV},0}$ & $\pi_{\textsf{IV},1/2}$ \\ \\
   
   \multicolumn{10}{c}{$\xi = 1$} \\
  \multirow{2}{*}{$C_1 = 3$} & 1.00 & 1.12 & 1.03 & 1.21 & 1.11 & 1.19 & 1.11 & 1.16 & 1.16 \\ 
   & [1.00,1.00] & [1.09,1.15] & [1.03,1.03] & [1.19,1.21] & [1.08,1.15] & [1.15,1.22] & [1.09,1.15] & [1.13,1.22] & [1.13,1.22] \\ 
  \multirow{2}{*}{$C_1 = 4$} & 1.00 & 1.13 & 1.03 & 1.21 & 1.11 & 1.19 & 1.11 & 1.16 & 1.16 \\ 
   & [1.00,1.00] & [1.09,1.15] & [1.03,1.03] & [1.19,1.24] & [1.08,1.15] & [1.15,1.22] & [1.09,1.14] & [1.14,1.22] & [1.13,1.22] \\ 
  \multirow{2}{*}{$C_1 = 5$} & 1.00 & 1.13 & 1.03 & 1.21 & 1.11 & 1.19 & 1.11 & 1.16 & 1.16 \\  
   & [1.00,1.00] & [1.09,1.15] & [1.03,1.03] & [1.19,1.21] & [1.08,1.15] & [1.15,1.21] & [1.09,1.14] & [1.14,1.21] & [1.13,1.21] \\ \\

   \multicolumn{10}{c}{$\xi = 2$} \\
   
  \multirow{2}{*}{$C_1 = 3$} & 1.00 & 1.08 & 0.94 & 1.14 & 1.08 & 1.12 & 1.09 & 1.13 & 1.13 \\ 
   & [1.00,1.00] & [1.07,1.11] & [0.94,0.94] & [1.13,1.14] & [1.05,1.11] & [1.11,1.13] & [1.08,1.11] & [1.12,1.15] & [1.12,1.16] \\ 
  \multirow{2}{*}{$C_1 = 4$} & 1.00 & 1.09 & 0.94 & 1.14 & 1.07 & 1.12 & 1.08 & 1.13 & 1.13 \\ 
   & [1.00,1.00] & [1.07,1.11] & [0.94,0.94] & [1.12,1.14] & [1.04,1.11] & [1.11,1.14] & [1.07,1.10] & [1.12,1.13] & [1.12,1.16] \\ 
  \multirow{2}{*}{$C_1 = 5$} & 1.00 & 1.09 & 0.94 & 1.14 & 1.07 & 1.12 & 1.08 & 1.13 & 1.13 \\ 
   & [1.00,1.00] & [1.08,1.11] & [0.94,0.94] & [1.13,1.14] & [1.04,1.11] & [1.11,1.14] & [1.07,1.10] & [1.12,1.13] & [1.12,1.15] \\    \\
  
   \multicolumn{10}{c}{$\xi = 3$} \\
   
  \multirow{2}{*}{$C_1 = 3$} & 1.00 & 1.05 & 0.88 & 1.10 & 1.05 & 1.09 & 1.07 & 1.09 & 1.10 \\  
   & [1.00,1.00] & [1.03,1.07] & [0.88,0.88] & [1.09,1.10] & [1.03,1.08] & [1.08,1.10] & [1.07,1.07] & [1.09,1.10] & [1.10,1.11] \\ 
  \multirow{2}{*}{$C_1 = 4$} & 1.00 & 1.05 & 0.88 & 1.10 & 1.05 & 1.09 & 1.06 & 1.09 & 1.10 \\  
   & [1.00,1.00] & [1.04,1.07] & [0.88,0.88] & [1.10,1.11] & [1.02,1.08] & [1.08,1.10] & [1.06,1.07] & [1.09,1.10] & [1.10,1.11] \\ 
  \multirow{2}{*}{$C_1 = 5$} & 1.00 & 1.06 & 0.88 & 1.10 & 1.05 & 1.09 & 1.05 & 1.09 & 1.10 \\ 
   & [1.00,1.00] & [1.04,1.07] & [0.88,0.88] & [1.10,1.11] & [1.03,1.08] & [1.08,1.10] & [1.05,1.07] & [1.09,1.10] & [1.10,1.11] \\ 
   \hline
\end{tabular}}
\end{table}

\begin{table}[h]
\caption{Simulation results for the cross-fitting version of the algorithm. All relevant conditional probabilities were estimated using random forests implemented in the \textsf{R} package \textsf{randomForest} with node size equal to $5$ and $n_{\textsf{train}} = 500$.}
\label{tbl: simulation cross fit n=500}
\centering
\resizebox{\textwidth}{!}{\begin{tabular}{cccccccccc}
  \hline
$n_{\textsf{train}} = 500$ & $\pibase_{\textsf{std}}$ & $\piimp_{\textsf{std}}$  & $\pibase_{\textsf{prosp}}$ & $\piimp_{\textsf{prosp}}$  & $\pibase_{\sra}$ & $\piimp_{\sra}$ & $\pi_{\textsf{IV},1}$ & $\pi_{\textsf{IV},0}$ & $\pi_{\textsf{IV},1/2}$ \\ \\
   
   \multicolumn{10}{c}{$\xi = 1$} \\
  \multirow{2}{*}{$C_1 = 3$} & 1.00 & 1.10 & 1.03 & 1.21 & 1.12 & 1.18 & 1.11 & 1.16 & 1.17 \\ 
   & [1.00,1.00] & [1.08,1.13] & [1.03,1.03] & [1.18,1.21] & [1.08,1.15] & [1.14,1.20] & [1.09,1.15] & [1.14,1.22] & [1.14,1.22] \\ 
  \multirow{2}{*}{$C_1 = 4$} & 1.00 & 1.11 & 1.03 & 1.21 & 1.11 & 1.18 & 1.10 & 1.17 & 1.17 \\ 
   & [1.00,1.00] & [1.08,1.14] & [1.03,1.03] & [1.18,1.21] & [1.08,1.15] & [1.14,1.21] & [1.09,1.14] & [1.14,1.22] & [1.14,1.22] \\ 
  \multirow{2}{*}{$C_1 = 5$} & 1.00 & 1.11 & 1.03 & 1.21 & 1.11 & 1.18 & 1.11 & 1.16 & 1.17 \\ 
   & [1.00,1.00] & [1.08,1.14] & [1.03,1.03] & [1.18,1.25] & [1.08,1.15] & [1.14,1.21] & [1.09,1.14] & [1.14,1.22] & [1.13,1.22] \\  \\
   
   \multicolumn{10}{c}{$\xi = 2$} \\
   
  \multirow{2}{*}{$C_1 = 3$} & 1.00 & 1.06 & 0.94 & 1.14 & 1.08 & 1.12 & 1.09 & 1.13 & 1.14 \\ 
   & [1.00,1.00] & [1.03,1.09] & [0.94,0.94] & [1.13,1.14] & [1.05,1.11] & [1.11,1.13] & [1.08,1.11] & [1.12,1.16] & [1.12,1.17] \\ 
  \multirow{2}{*}{$C_1 = 4$} & 1.00 & 1.07 & 0.94 & 1.14 & 1.07 & 1.12 & 1.08 & 1.13 & 1.13 \\ 
   & [1.00,1.00] & [1.06,1.10] & [0.94,0.94] & [1.13,1.14] & [1.05,1.11] & [1.11,1.13] & [1.07,1.10] & [1.12,1.14] & [1.12,1.17] \\ 
  \multirow{2}{*}{$C_1 = 5$} & 1.00 & 1.08 & 0.94 & 1.14 & 1.07 & 1.12 & 1.08 & 1.13 & 1.13 \\ 
   & [1.00,1.00] & [1.06,1.10] & [0.94,0.94] & [1.13,1.14] & [1.05,1.11] & [1.11,1.14] & [1.07,1.10] & [1.12,1.16] & [1.12,1.17] \\     \\
  
   \multicolumn{10}{c}{$\xi = 3$} \\
   
  \multirow{2}{*}{$C_1 = 3$} & 1.00 & 1.03 & 0.88 & 1.09 & 1.05 & 1.09 & 1.07 & 1.08 & 1.10 \\ 
   & [1.00,1.00] & [1.00,1.06] & [0.88,0.88] & [1.09,1.10] & [1.04,1.08] & [1.08,1.10] & [1.07,1.07] & [1.08,1.10] & [1.10,1.11] \\ 
  \multirow{2}{*}{$C_1 = 4$} & 1.00 & 1.04 & 0.88 & 1.10 & 1.05 & 1.09 & 1.06 & 1.08 & 1.10 \\ 
   & [1.00,1.00] & [1.00,1.06] & [0.88,0.88] & [1.09,1.10] & [1.03,1.08] & [1.08,1.10] & [1.05,1.07] & [1.08,1.10] & [1.10,1.11] \\ 
  \multirow{2}{*}{$C_1 = 5$} & 1.00 & 1.04 & 0.88 & 1.10 & 1.05 & 1.09 & 1.05 & 1.08 & 1.10 \\ 
   & [1.00,1.00] & [1.00,1.06] & [0.88,0.88] & [1.10,1.10] & [1.02,1.08] & [1.08,1.10] & [1.04,1.07] & [1.08,1.10] & [1.10,1.11] \\ 
   \hline
\end{tabular}}
\end{table}

\begin{table}[h]
\caption{Simulation results for the cross-fitting version of the algorithm. All relevant conditional probabilities were estimated using random forests implemented in the \textsf{R} package \textsf{randomForest} with node size equal to $5$ and $n_{\textsf{train}} = 1000$.}
\label{tbl: simulation cross fit n=1000}
\centering
\resizebox{\textwidth}{!}{\begin{tabular}{cccccccccc}
  \hline
$n_{\textsf{train}} = 500$ & $\pibase_{\textsf{std}}$ & $\piimp_{\textsf{std}}$  & $\pibase_{\textsf{prosp}}$ & $\piimp_{\textsf{prosp}}$  & $\pibase_{\sra}$ & $\piimp_{\sra}$ & $\pi_{\textsf{IV},1}$ & $\pi_{\textsf{IV},0}$ & $\pi_{\textsf{IV},1/2}$ \\ \\
   
   \multicolumn{10}{c}{$\xi = 1$} \\
  \multirow{2}{*}{$C_1 = 3$} & 1.00 & 1.13 & 1.03 & 1.23 & 1.13 & 1.21 & 1.13 & 1.19 & 1.19 \\ 
   & [1.00,1.00] & [1.12,1.15] & [1.03,1.03] & [1.20,1.29] & [1.10,1.17] & [1.17,1.25] & [1.09,1.15] & [1.14,1.23] & [1.14,1.23] \\ 
  \multirow{2}{*}{$C_1 = 4$} & 1.00 & 1.14 & 1.03 & 1.23 & 1.13 & 1.21 & 1.13 & 1.18 & 1.18 \\ 
   & [1.00,1.00] & [1.13,1.15] & [1.03,1.03] & [1.20,1.29] & [1.10,1.16] & [1.18,1.25] & [1.09,1.15] & [1.14,1.23] & [1.14,1.23] \\ 
  \multirow{2}{*}{$C_1 = 5$} & 1.00 & 1.14 & 1.03 & 1.23 & 1.13 & 1.21 & 1.12 & 1.18 & 1.18 \\ 
   & [1.00,1.00] & [1.13,1.15] & [1.03,1.03] & [1.20,1.29] & [1.10,1.16] & [1.18,1.26] & [1.09,1.15] & [1.14,1.23] & [1.14,1.23] \\  \\
   
   \multicolumn{10}{c}{$\xi = 2$} \\
   
  \multirow{2}{*}{$C_1 = 3$} & 1.00 & 1.09 & 0.94 & 1.14 & 1.10 & 1.13 & 1.10 & 1.15 & 1.15 \\ 
   & [1.00,1.00] & [1.08,1.11] & [0.94,0.94] & [1.13,1.14] & [1.08,1.12] & [1.12,1.15] & [1.08,1.11] & [1.12,1.17] & [1.12,1.17] \\ 
  \multirow{2}{*}{$C_1 = 4$} & 1.00 & 1.10 & 0.94 & 1.14 & 1.10 & 1.13 & 1.10 & 1.14 & 1.15 \\ 
   & [1.00,1.00] & [1.09,1.11] & [0.94,0.94] & [1.13,1.17] & [1.08,1.12] & [1.12,1.16] & [1.08,1.11] & [1.12,1.17] & [1.12,1.17] \\ 
  \multirow{2}{*}{$C_1 = 5$} & 1.00 & 1.10 & 0.94 & 1.14 & 1.09 & 1.14 & 1.10 & 1.14 & 1.15 \\ 
   & [1.00,1.00] & [1.09,1.11] & [0.94,0.94] & [1.13,1.17] & [1.08,1.12] & [1.12,1.16] & [1.08,1.11] & [1.12,1.17] & [1.12,1.17] \\     \\
  
   \multicolumn{10}{c}{$\xi = 3$} \\
   
  \multirow{2}{*}{$C_1 = 3$} & 1.00 & 1.05 & 0.88 & 1.10 & 1.07 & 1.10 & 1.07 & 1.10 & 1.11 \\ 
   & [1.00,1.00] & [1.04,1.07] & [0.88,0.88] & [1.10,1.10] & [1.06,1.09] & [1.09,1.10] & [1.07,1.07] & [1.10,1.11] & [1.11,1.11] \\ 
  \multirow{2}{*}{$C_1 = 4$} & 1.00 & 1.06 & 0.88 & 1.10 & 1.07 & 1.10 & 1.07 & 1.10 & 1.11 \\ 
   & [1.00,1.00] & [1.05,1.07] & [0.88,0.88] & [1.10,1.11] & [1.05,1.09] & [1.09,1.11] & [1.06,1.07] & [1.10,1.11] & [1.10,1.11] \\ 
  \multirow{2}{*}{$C_1 = 5$} & 1.00 & 1.06 & 0.88 & 1.10 & 1.06 & 1.10 & 1.06 & 1.10 & 1.10 \\ 
   & [1.00,1.00] & [1.05,1.07] & [0.88,0.88] & [1.10,1.11] & [1.05,1.09] & [1.09,1.10] & [1.05,1.07] & [1.10,1.11] & [1.10,1.11] \\ 
   \hline
\end{tabular}}
\end{table}

\clearpage
\begin{figure}[ht]
    \centering
    \includegraphics[width = 0.85\textwidth]{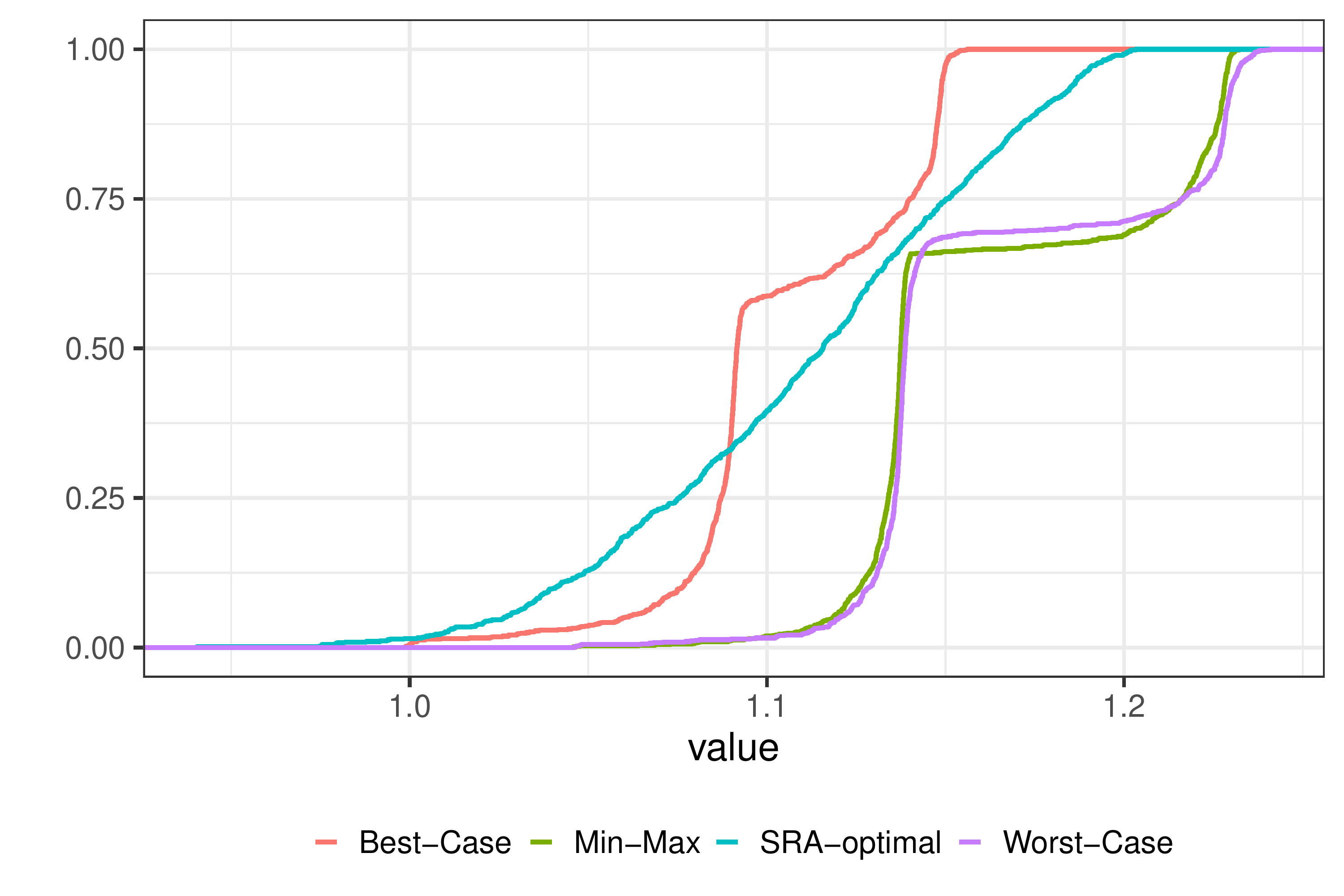}
    \caption{Cumulative distribution functions of value functions of SRA-optimal policy $\pibase_\sra$ and three IV-optimal policies $\pi_{\textsf{IV},0}$ (worst-case), $\pi_{\textsf{IV},1}$ (best-case), and $\pi_{\textsf{IV},1/2}$ (min-max) when $C_1 = 4$ and $\xi = 1$ and all relevant conditional expectations are estimated via random forests.}
    \label{fig: C=4 xi = 1 IV-optimal rules}
\end{figure}

\begin{figure}[ht]
    \centering
    \includegraphics[width = 0.85\textwidth]{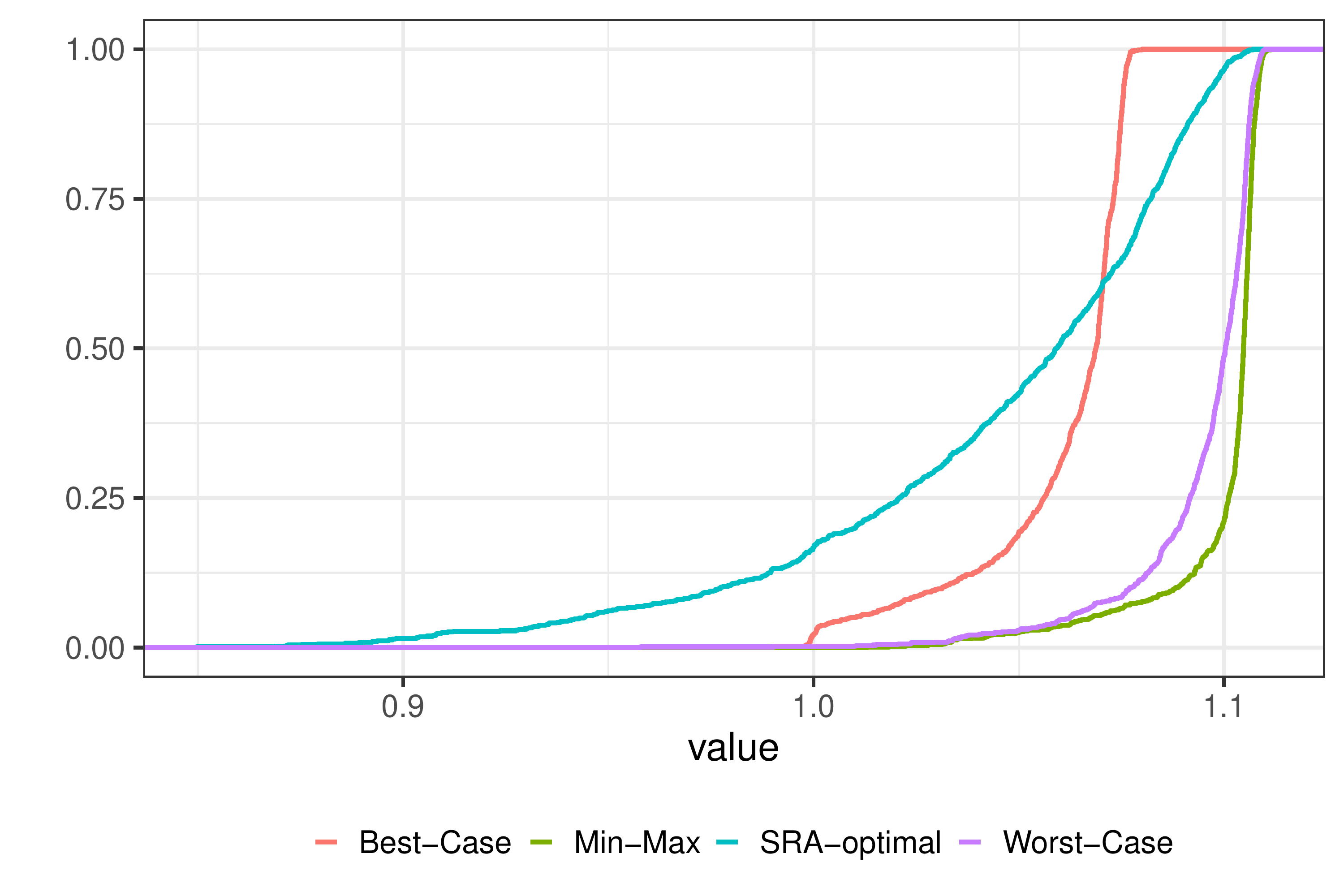}
    \caption{Cumulative distribution functions of value functions of SRA-optimal policy $\pibase_\sra$ and three IV-optimal policies $\pi_{\textsf{IV},0}$ (worst-case), $\pi_{\textsf{IV},1}$ (best-case), and $\pi_{\textsf{IV},1/2}$ (min-max) when $C_1 = 4$ and $\xi = 3$ and all relevant conditional expectations are estimated via random forests.}
    \label{fig: C=4 xi = 3 IV-optimal rules}
\end{figure}

\clearpage
\section{Details on the Real Data}
\subsection{Causal Direct Acyclic Graph (DAG)}
\label{subsec: appendix DAG}
We provide more details on the causal direct acyclic graph (DAG) with a time-varying IV that helps understand when a time-varying IV is necessary to identify relevant causal effects. Figure \ref{fig: time-varying IV DAG} exhibits a two-stage DAG: $Z_1$ and $Z_2$ are two time-varying IVs, $A_1$ and $A_2$ treatment received, $R_1$ and $R_2$ outcomes, and $U_1$, $U_2$ and $U$ unmeasured confounders. We omit observed covariates for clearer presentation. We explicitly differentiate between two types of unmeasured confounding in the DAG: $U_1$ and $U_2$ are unmeasured confounding specific to the first stage ($A_1$ and $R_1$) and the second stage ($A_2$ and $R_2$), and $U$ represents unmeasured confounding shared between both stages. When studying the effect of $A_2$ on $R_2$, $Z_2$ is a valid IV for $A_2$; however, in the presence of shared unmeasured confounding $U$, $Z_1$ is not a valid IV for $A_2$. This is because conditioning on $A_1$ induces association between $Z_1$ and $U$ (represented by the dashed line in the DAG) and hence the IV unconfoundedness assumption is violated. The association between $Z_1$ and $U$ induced by conditioning on $A_1$ is known as the \emph{collider bias} in the DAG literature (\citealp{hernan2004structural}).

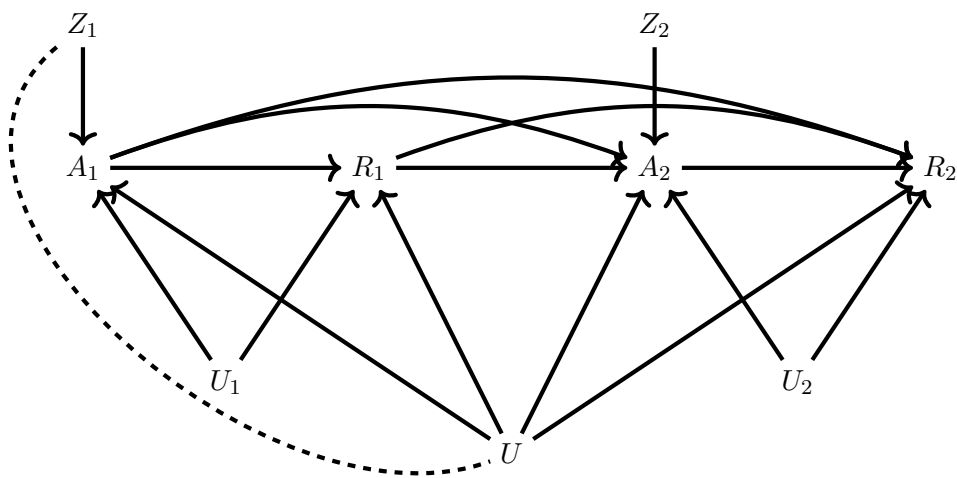
\begin{figure}[h]
\centering
\begin{tikzpicture}[scale=0.95, every text node part/.style={align=center}
]
\node at (0,0) (A1) {$A_1$};
\node at (4,0) (R1) {$R_1$};
\node at (8,0) (A2) {$A_2$};
\node at (12,0) (R2) {$R_2$};

\node at (2,-3) (U1) {$U_1$};
\node at (10,-3) (U2) {$U_2$};
\node at (6,-4) (U) {$U$};

\node at (0,2) (Z1) {$Z_1$};
\node at (8,2) (Z2) {$Z_2$};

\draw[ultra thick, ->] (A1) -- (R1);
\draw[ultra thick, ->] (R1) -- (A2);
\draw[ultra thick, ->] (A2) -- (R2);
\draw[ultra thick, ->] (A1) to [out=20, in=160] (A2);
\draw[ultra thick, ->] (A1) to [out=20, in=160](R2);
\draw[ultra thick, ->] (R1) to [out=20, in=160](R2);

\draw[ultra thick, ->] (U1) -- (A1);
\draw[ultra thick, ->] (U1) -- (R1);

\draw[ultra thick, ->] (U2) -- (A2);
\draw[ultra thick, ->] (U2) -- (R2);

\draw[ultra thick, ->] (U) -- (A1);
\draw[ultra thick, ->] (U) -- (R1);
\draw[ultra thick, ->] (U) -- (A2);
\draw[ultra thick, ->] (U) -- (R2);

\draw[ultra thick, ->] (Z1) -- (A1);
\draw[ultra thick, ->] (Z2) -- (A2);

\draw[ultra thick, -, dashed] (Z1) to [out=220, in=200](U);

\end{tikzpicture}
    \caption{\small A two-stage DAG with a time-varying instrumental variable. In the DAG, $U_1$ (and analogously $U_2$) represents stage-$1$ (and stage-$2$) unmeasured confounders. $U$ represents shared unmeasured confounding. $Z_1$ is an IV for $A_1$ and $Z_2$ an IV for $A_2$. Note that in the presence of the shared unmeasured confounding $U$, $Z_1$ is \emph{not} a valid for $A_2$ even after $A_1$ and $R_1$ are controlled for. This is because conditioning on $A_1$ induces association between $Z_1$ and $U$ so that $Z_1$ is no longer independent of the unmeasured confounder $U$.}
    \label{fig: time-varying IV DAG}
\end{figure}

\newpage
\subsection{More Details on the NICU Application}
Our raw data consist of all births in the Commonwealth of Pennsylvania between 1995 and 2009; there is one ID associated with each delivery and one associated with each mother, from which deliveries of multiple babies by the same mother were identified. The data combine information from birth and death certificates and the UB-92 form that hospitals provide. The American Academy of Pediatrics recognizes six levels of neonatal intensive care units (NICUs) of increasing technical expertise and capability, namely 1, 2, 3A, 3B, 3C, 3D and regional centers (\citealp{baiocchi2010building}). We followed \citet{baiocchi2010building} and defined an NICU as low-level if its designation is below $3A$ and high-level otherwise. Travel time was determined using the software ArcView as the time from the centroid of mothers' zip code to the closest low and high level hospitals. 

We considered mothers who delivered exactly two babies during the period under consideration, and excluded less than $1\%$ deliveries that missed at least one outcome (death) at two deliveries. There are a total of $183,487$ mothers and $183,487 \times 2 = 366,974$ deliveries in our final study cohort. Approximately $8.40\%$ of these $366,974$ babies were premature (less than $37$ weeks in gestational age). We considered the following covariates in our analysis: percentage of people having a college degree in mother's neighborhood (\textsf{college}), poverty rate of the neighborhood (\textsf{poverty}), home value of the neighborhood (\textsf{homeval}), percentage of people renting in the neighborhood (\textsf{rent}), median income of the neighborhood (\textsf{medincome}), and urban/rural of the neighborhood (\textsf{urban}), whether the mother is white (\textsf{white}), mother's age at delivery (\textsf{ageM}), mother's years of education (\textsf{educyrM}), how many months of prenatal care received (\textsf{precare}), gestational age in weeks (\textsf{GA}), and seven congenital anomalies. There are many meaningful causal questions concerning individualized treatment rules that can be answered with the NICU data. In this analysis, we are most interested in developing a system that assigns mothers who are about to deliver to an appropriate NICU so we included only covariates that are observed prior to the delivery. Covariates concerning mother's neighborhood are included in the study to ensure that the excess travel time IV is more likely to be a valid one.

Though $183,487 \times 2 = 366,974$ deliveries considered in the analysis were complete in the IV, treatment, and outcome data, some of them missed covariates data. We imputed the missing covariates data using the widely-used \emph{multiple imputation by chain equations} method (\citealp{buuren2010mice}) and repeated our analysis on each of the five imputed dataset. The estimated tree-structured DTRs were nearly identical for each imputed dataset.
\end{appendices}

\end{document}